\documentclass[10pt]{amsart}
\usepackage{fullpage} 
\usepackage{amssymb} 
\usepackage{graphicx} 
\usepackage{amsmath,amsthm,amssymb,latexsym} 
\usepackage{amstext, amsfonts,amsbsy} 
\usepackage{enumerate} 
\usepackage{upgreek} 
\usepackage{color} 
\usepackage{accents} 
\usepackage{mathtools} 
\usepackage{tikz}
\usetikzlibrary{matrix,graphs,arrows,positioning,calc,decorations.markings,shapes.symbols}
\usepackage[pdftex,bookmarks,colorlinks,breaklinks]{hyperref}  
\definecolor{dullmagenta}{rgb}{0.4,0,0.4}   
\definecolor{darkblue}{rgb}{0,0,0.4}
\hypersetup{linkcolor=red,citecolor=blue,filecolor=dullmagenta,urlcolor=darkblue} 
\usepackage{eucal}

\newtheorem{theorem}{Theorem}[section]
\newtheorem{proposition}[theorem]{Proposition}
\newtheorem{corollary}[theorem]{Corollary}
\newtheorem{lemma}[theorem]{Lemma}
\newtheorem{definition}[theorem]{Definition}
\newtheorem{assumption}[theorem]{Assumption}

\theoremstyle{definition}

\theoremstyle{remark}
\newtheorem{remark}[theorem]{Remark}

\numberwithin{equation}{section}
\numberwithin{theorem}{subsection}

\DeclareMathOperator*{\residue}{Res}
\newcommand{\lxzin}{\begin{tikzpicture}[>=stealth,scale=0.3,baseline=0pt,
        lozxz/.style={fill=red!60!white, draw=red!50!black}]
        \filldraw[lozxz]
            (${0}*({cos(deg(-pi/6))},{sin(deg(-pi/6))}) + {0}*({cos(deg(pi/6))},{sin(deg(pi/6))}) + {0}*(0,1)$) --
                (${1}*({cos(deg(-pi/6))},{sin(deg(-pi/6))}) + {0}*({cos(deg(pi/6))},{sin(deg(pi/6))}) + {0}*(0,1)$) --
                (${1}*({cos(deg(-pi/6))},{sin(deg(-pi/6))}) + {0}*({cos(deg(pi/6))},{sin(deg(pi/6))}) + {1}*(0,1)$) --
                (${0}*({cos(deg(-pi/6))},{sin(deg(-pi/6))}) + {0}*({cos(deg(pi/6))},{sin(deg(pi/6))}) + {1}*(0,1)$) -- cycle;
		\draw (${0}*({cos(deg(-pi/6))},{sin(deg(-pi/6))}) + {0}*({cos(deg(pi/6))},{sin(deg(pi/6))}) + {0}*(0,1)$) --
                (${1}*({cos(deg(-pi/6))},{sin(deg(-pi/6))}) + {0}*({cos(deg(pi/6))},{sin(deg(pi/6))}) + {1}*(0,1)$);
\end{tikzpicture}}
\newcommand{\lyzin}{\begin{tikzpicture}[>=stealth,scale=0.3,baseline=4pt,
        lozyz/.style={fill=blue!18!white, draw=blue!50!black}]
        \filldraw[lozyz]
            (${0}*({cos(deg(-pi/6))},{sin(deg(-pi/6))}) + {0}*({cos(deg(pi/6))},{sin(deg(pi/6))}) + {0}*(0,1)$) --
                (${0}*({cos(deg(-pi/6))},{sin(deg(-pi/6))}) + {1}*({cos(deg(pi/6))},{sin(deg(pi/6))}) + {0}*(0,1)$) --
                (${0}*({cos(deg(-pi/6))},{sin(deg(-pi/6))}) + {1}*({cos(deg(pi/6))},{sin(deg(pi/6))}) + {1}*(0,1)$) --
                (${0}*({cos(deg(-pi/6))},{sin(deg(-pi/6))}) + {0}*({cos(deg(pi/6))},{sin(deg(pi/6))}) + {1}*(0,1)$) -- cycle;
		\draw (${0}*({cos(deg(-pi/6))},{sin(deg(-pi/6))}) + {1}*({cos(deg(pi/6))},{sin(deg(pi/6))}) + {0}*(0,1)$) --
                (${0}*({cos(deg(-pi/6))},{sin(deg(-pi/6))}) + {0}*({cos(deg(pi/6))},{sin(deg(pi/6))}) + {1}*(0,1)$);		
\end{tikzpicture}}
\newcommand{\lxyin}{\begin{tikzpicture}[>=stealth,scale=0.3,baseline=-3pt,
        lozxy/.style={fill=green!20!white, draw=green!50!black}]
        \filldraw[lozxy]
            (${0}*({cos(deg(-pi/6))},{sin(deg(-pi/6))}) + {0}*({cos(deg(pi/6))},{sin(deg(pi/6))}) + {0}*(0,1)$) --
                (${1}*({cos(deg(-pi/6))},{sin(deg(-pi/6))}) + {0}*({cos(deg(pi/6))},{sin(deg(pi/6))}) + {0}*(0,1)$) --
                (${1}*({cos(deg(-pi/6))},{sin(deg(-pi/6))}) + {1}*({cos(deg(pi/6))},{sin(deg(pi/6))}) + {0}*(0,1)$) --
                (${0}*({cos(deg(-pi/6))},{sin(deg(-pi/6))}) + {1}*({cos(deg(pi/6))},{sin(deg(pi/6))}) + {0}*(0,1)$) -- cycle;
		\draw (${1}*({cos(deg(-pi/6))},{sin(deg(-pi/6))}) + {0}*({cos(deg(pi/6))},{sin(deg(pi/6))}) + {0}*(0,1)$) --
                (${0}*({cos(deg(-pi/6))},{sin(deg(-pi/6))}) + {1}*({cos(deg(pi/6))},{sin(deg(pi/6))}) + {0}*(0,1)$);		
\end{tikzpicture}}
\newcommand{\lxyinf}{\begin{tikzpicture}[>=stealth,scale=0.3,baseline=-3pt,
        lozxy/.style={fill=green!20!white, draw=green!50!black}]
        \filldraw[lozxy]
            (${0}*({cos(deg(-pi/6))},{sin(deg(-pi/6))}) + {0}*({cos(deg(pi/6))},{sin(deg(pi/6))}) + {0}*(0,1)$) --
                (${1}*({cos(deg(-pi/6))},{sin(deg(-pi/6))}) + {0}*({cos(deg(pi/6))},{sin(deg(pi/6))}) + {0}*(0,1)$) --
                (${1}*({cos(deg(-pi/6))},{sin(deg(-pi/6))}) + {1}*({cos(deg(pi/6))},{sin(deg(pi/6))}) + {0}*(0,1)$) --
                (${0}*({cos(deg(-pi/6))},{sin(deg(-pi/6))}) + {1}*({cos(deg(pi/6))},{sin(deg(pi/6))}) + {0}*(0,1)$) -- cycle;
\end{tikzpicture}}

\begin{document}

\title[$q$-Racah and Discrete Painlev\'e Equations]
{\MakeLowercase{$q$}-Racah ensemble and \MakeLowercase{$q$}-P$\left(E_7^{(1)}/A_{1}^{(1)}\right)$ Discrete Painlev\'e equation}

\author{Anton Dzhamay}
\address{Anton Dzhamay:\quad School of Mathematical Sciences\\
The University of Northern Colorado\\
Campus Box 122\\
501 20th Street\\
Greeley, CO 80639, USA}
\email{\href{mailto:adzham@unco.edu}{\texttt{adzham@unco.edu}}}
\thanks{}
\author{Alisa Knizel}
\address{Aliza Knizel:\quad Department of Mathematics\\
Columbia University\\
New York, NY, USA }
\email{\href{mailto:knizel@math.columbia.edu}{\texttt{knizel@math.columbia.edu}}}

\keywords{Gap probabilities, orthogonal polynomial ensembles, Askey-Wilson scheme, Painlev\'e equations, difference equations,
isomonodromic transformations, birational transformations}
\subjclass[2010]{33D45, 34M55, 34M56, 14E07, 39A13}

\date{}

\begin{abstract}
	
 The goal of this paper is to investigate the missing part of the story about the relationship between
 the orthogonal polynomial ensembles and Painlev\'e equations. Namely, we consider the $q$-Racah polynomial
 ensemble and show that the one-interval gap probabilities in this case can be expressed through a solution of the
 discrete $q$-P$\left(E_7^{(1)}/A_{1}^{(1)}\right)$ equation. Our approach also gives a new Lax pair for this equation.
 This Lax pair has an interesting additional involutive symmetry structure.
	
\end{abstract}

\maketitle


\section{Introduction} 
\label{sec:introduction}

The present paper is a continuation of the work on the relationship between the orthogonal polynomial ensembles and Painlev\'e equations \cite{K}, where the
$q$-analogue of methods introduced by Arinkin and Borodin in \cite{BA1} was developed.
This relationship in the continuous settings was first established in the 90's \cite{TW,HS,WF,BD}. First results in the discrete case were obtained in a
paper by Borodin and Boyarchenko \cite{BB} using the formalism of discrete integrable operators and discrete Riemann-Hilbert problems. That paper will be the
starting point of our investigation.

Our goal is to establish a certain recurrence procedure for computing the so-called {\it gap probability function} for the $q$-Racah
orthogonal polynomial ensemble and to show that this function can be expressed through a solution of a
$q$-P$\left(E_7^{(1)}/A_{1}^{(1)}\right)$ discrete Painlev\'e equation, as written in \cite{KajNouYam:2017:GAOPE}.
For us, the original motivation to study this ensemble comes from its relationship to an interesting tiling model that we describe next.

\subsection{The $q$-Racah tiling model} 
 \label{sub:tiling_model}

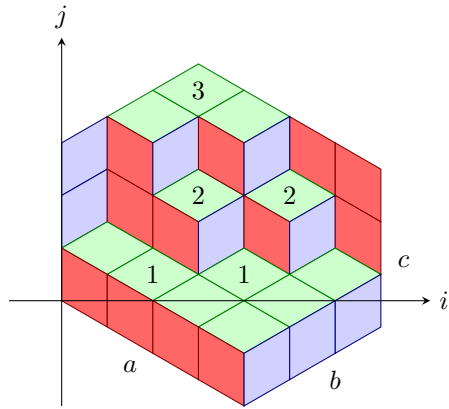
\begin{figure}[h]
\centering
\begin{tikzpicture}[>=stealth,
        lozxy/.style={fill=green!20!white, draw=green!50!black},
        lozyz/.style={fill=blue!18!white, draw=blue!50!black},
        lozxz/.style={fill=red!60!white, draw=red!50!black}, scale=0.7]

\def\lxz(#1:#2:#3){
        \filldraw[lozxz]
            (${(#1)+0}*({cos(deg(-pi/6))},{sin(deg(-pi/6))}) + {(#2)+0}*({cos(deg(pi/6))},{sin(deg(pi/6))}) + {(#3)+0}*(0,1)$) --
                (${(#1)+1}*({cos(deg(-pi/6))},{sin(deg(-pi/6))}) + {(#2)+0}*({cos(deg(pi/6))},{sin(deg(pi/6))}) + {(#3)+0}*(0,1)$) --
                (${(#1)+1}*({cos(deg(-pi/6))},{sin(deg(-pi/6))}) + {(#2)+0}*({cos(deg(pi/6))},{sin(deg(pi/6))}) + {(#3)+1}*(0,1)$) --
                (${(#1)+0}*({cos(deg(-pi/6))},{sin(deg(-pi/6))}) + {(#2)+0}*({cos(deg(pi/6))},{sin(deg(pi/6))}) + {(#3)+1}*(0,1)$) -- cycle;}

\def\lyz(#1:#2:#3){
        \filldraw[lozyz]
            (${(#1)+0}*({cos(deg(-pi/6))},{sin(deg(-pi/6))}) + {(#2)+0}*({cos(deg(pi/6))},{sin(deg(pi/6))}) + {(#3)+0}*(0,1)$) --
                (${(#1)+0}*({cos(deg(-pi/6))},{sin(deg(-pi/6))}) + {(#2)+1}*({cos(deg(pi/6))},{sin(deg(pi/6))}) + {(#3)+0}*(0,1)$) --
                (${(#1)+0}*({cos(deg(-pi/6))},{sin(deg(-pi/6))}) + {(#2)+1}*({cos(deg(pi/6))},{sin(deg(pi/6))}) + {(#3)+1}*(0,1)$) --
                (${(#1)+0}*({cos(deg(-pi/6))},{sin(deg(-pi/6))}) + {(#2)+0}*({cos(deg(pi/6))},{sin(deg(pi/6))}) + {(#3)+1}*(0,1)$) -- cycle;}

\def\lxy(#1:#2:#3)[#4]{
        \filldraw[lozxy]
            (${(#1)+0}*({cos(deg(-pi/6))},{sin(deg(-pi/6))}) + {(#2)+0}*({cos(deg(pi/6))},{sin(deg(pi/6))}) + {(#3)+0}*(0,1)$) --
                (${(#1)+1}*({cos(deg(-pi/6))},{sin(deg(-pi/6))}) + {(#2)+0}*({cos(deg(pi/6))},{sin(deg(pi/6))}) + {(#3)+0}*(0,1)$) --
                (${(#1)+1}*({cos(deg(-pi/6))},{sin(deg(-pi/6))}) + {(#2)+1}*({cos(deg(pi/6))},{sin(deg(pi/6))}) + {(#3)+0}*(0,1)$) --
                (${(#1)+0}*({cos(deg(-pi/6))},{sin(deg(-pi/6))}) + {(#2)+1}*({cos(deg(pi/6))},{sin(deg(pi/6))}) + {(#3)+0}*(0,1)$) -- cycle;
                \node at (${(#1)+0.5}*({cos(deg(-pi/6))},{sin(deg(-pi/6))}) + {(#2)+0.5}*({cos(deg(pi/6))},{sin(deg(pi/6))}) + {(#3)+0}*(0,1)$) {$#4$};}

        \lxy(0:1:3)[];\lxy(0:2:3)[3];\lxy(1:2:3)[];\lxz(2:3:2)[];\lxz(3:3:2)[];
        \lyz(0:0:2);\lxz(0:1:2);\lyz(1:1:2);\lxy(1:1:2)[2];\lxz(1:2:2);\lxy(2:2:2)[2];\lyz(2:2:2);
        \lyz(0:0:1);\lxz(0:1:1);\lxz(1:1:1);\lyz(2:1:1);\lxz(2:2:1);\lyz(3:2:1);\lxz(3:3:1);
        \lxy(0:0:1)[];\lxy(1:0:1)[1];\lxy(2:0:1)[];\lxy(2:1:1)[1];\lxy(3:0:1)[];\lxy(3:1:1)[];\lxy(3:2:1)[];
        \lxz(0:0:0);\lxz(1:0:0);\lxz(2:0:0);\lxz(3:0:0);\lyz(4:0:0);\lyz(4:1:0);\lyz(4:2:0);

        \node at (${2}*({cos(deg(-pi/6))},{sin(deg(-pi/6))}) + {-0.5}*({cos(deg(pi/6))},{sin(deg(pi/6))}) + {0}*(0,1)$) {$a$};
        \node at (${4.5}*({cos(deg(-pi/6))},{sin(deg(-pi/6))}) + {1.5}*({cos(deg(pi/6))},{sin(deg(pi/6))}) + {0}*(0,1)$) {$b$};
        \node at (${4.5}*({cos(deg(-pi/6))},{sin(deg(-pi/6))}) + {3}*({cos(deg(pi/6))},{sin(deg(pi/6))}) + {1.5}*(0,1)$) {$c$};

     \draw[->] (-1,0) -- (7,0) node[right] {$i$};
     \draw[->] (0,-2) -- (0,5) node[above] {$j$};
\end{tikzpicture}
\caption{A tiling of a $4\times 3\times 3$ hexagon} \label{fig:hex}
\end{figure}

Consider a hexagon, drawn on a regular triangular lattice, whose side lengths are given by integers $a,b,c \geq 1$, see Figure \ref{fig:hex}. We are
interested in random tilings of such a hexagon by rhombi, also called lozenges, that are obtained by gluing two neighboring triangles together.
There are three types of rhombi that arise in such a way: \lxzin, \lyzin, and \lxyin, and so, as can be clearly seen in Figure~\ref{fig:hex},
this model also has a natural interpretation as a random stepped surface formed by a stack of boxes or, equivalently,
as a boxed plane partition (that is also called a $3$-D Young diagram). In this way we can associate a tiling with a \emph{height function} $h$
that assigns to every lattice vertex 
inside the hexagon its ``height'' 
above the ``horizontal plane'', as shown on Figure~\ref{fig:hex}.
%
%

We are interested in the probability  measures on the set of such tilings that were introduced  in \cite{BGR}. These probability measures form
a two-parameter family generalization of the uniform distribution. If we denote these parameters by $q$ and $\kappa$,
the weight of a tiling is defined to be the  product of simple factors
$w(\lxyinf_{i,j}) = ( \kappa q^{  j -(c+1)/2 }- q^{-j + (c+1)/2}/\kappa)$ over all  horizontal rhombi $\lxyinf$, where
$(i,j)$ is the coordinate of the topmost point of the rhombus (the $i$ and $j$ axes are shown on Figure~\ref{fig:hex}).
The dependence of the factors on the location  of the lozenge makes the model \emph{inhomogeneous}.
In order to define a probability  measure, the weight of a tiling has to be non-negative.  This imposes certain restrictions on the
parameters $q$ and $\kappa$ that we discuss in Section~\ref{sec:ensemble}.

An important observation is that each lozenge tiling can be considered as time-dependent configuration of points on the line. To
make this connection, we perform a simple affine transformation of the hexagon to get the shifted hexagon
and the new coordinates $(x,t)$ as shown in Figure~\ref{fig:paths}.
Then each tiling naturally corresponds to  a family of $N = a$ non-intersecting up-right paths
(formed by the midlines of  the tiles of the first two
types). For each $0 \leq  t \leq b+c$ we draw a vertical line through  the point $(t, 0)$ and denote by $$x^t_1 <  x^t_2< \cdots < x^t_N$$ the
points of intersection of the line with the $N$ up-right paths. In  this way, we can view a tiling as an $N$-point configuration, which varies
in time. Define \emph{the gap probability function} on a slice $t$ as
\begin{equation}\label{eq:gap_prob}
D_t(s)=\textup {Prob}\left (x^t_N<s \right );
\end{equation}
this function is the main object of our study.

\begin{figure}[h]
\centering
\begin{tikzpicture}[>=stealth,
        brxy/.style={fill=green!20!white, draw=green!50!black},
        bryz/.style={fill=blue!18!white, draw=blue!50!black},
        brxz/.style={fill=red!60!white, draw=red!50!black},
		pth/.style={very thick, draw=black},
		intpt/.style={circle, draw=white!100, fill=purple!80!black, very thick, inner sep=1pt, minimum size=2.5mm},
		scale=0.7
]

\def\brxy(#1:#2:#3){
        \filldraw[brxy]
            (${(#1)+0}*(1,0) + {(#2)+0}*({sqrt(2)*cos(deg(pi/4))},{sqrt(2)*sin(deg(pi/4))}) + {(#3)-0.5}*(0,1)$) --
            (${(#1)+1}*(1,0) + {(#2)+0}*({sqrt(2)*cos(deg(pi/4))},{sqrt(2)*sin(deg(pi/4))}) + {(#3)-0.5}*(0,1)$) --
            (${(#1)+1}*(1,0) + {(#2)+1}*({sqrt(2)*cos(deg(pi/4))},{sqrt(2)*sin(deg(pi/4))}) + {(#3)-0.5}*(0,1)$) --
            (${(#1)+0}*(1,0) + {(#2)+1}*({sqrt(2)*cos(deg(pi/4))},{sqrt(2)*sin(deg(pi/4))}) + {(#3)-0.5}*(0,1)$) -- cycle;
       }

\def\bryz(#1:#2:#3){
        \filldraw[bryz]
            (${(#1)+0}*(1,0) + {(#2)+0}*({sqrt(2)*cos(deg(pi/4))},{sqrt(2)*sin(deg(pi/4))}) + {(#3)-0.5}*(0,1)$) --
            (${(#1)+0}*(1,0) + {(#2)+1}*({sqrt(2)*cos(deg(pi/4))},{sqrt(2)*sin(deg(pi/4))}) + {(#3)-0.5}*(0,1)$) --
            (${(#1)+0}*(1,0) + {(#2)+1}*({sqrt(2)*cos(deg(pi/4))},{sqrt(2)*sin(deg(pi/4))}) + {(#3)+0.5}*(0,1)$) --
            (${(#1)+0}*(1,0) + {(#2)+0}*({sqrt(2)*cos(deg(pi/4))},{sqrt(2)*sin(deg(pi/4))}) + {(#3)+0.5}*(0,1)$) -- cycle;
		\draw[pth] 	
			(${(#1)+0}*(1,0) + {(#2)+0}*({sqrt(2)*cos(deg(pi/4))},{sqrt(2)*sin(deg(pi/4))}) + {(#3)}*(0,1)$) --
            (${(#1)+0}*(1,0) + {(#2)+1}*({sqrt(2)*cos(deg(pi/4))},{sqrt(2)*sin(deg(pi/4))}) + {(#3)}*(0,1)$);
       }

\def\brxz(#1:#2:#3){
        \filldraw[brxz]
            (${(#1)+0}*(1,0) + {(#2)+0}*({sqrt(2)*cos(deg(pi/4))},{sqrt(2)*sin(deg(pi/4))}) + {(#3)-0.5}*(0,1)$) --
            (${(#1)+1}*(1,0) + {(#2)+0}*({sqrt(2)*cos(deg(pi/4))},{sqrt(2)*sin(deg(pi/4))}) + {(#3)-0.5}*(0,1)$) --
            (${(#1)+1}*(1,0) + {(#2)+0}*({sqrt(2)*cos(deg(pi/4))},{sqrt(2)*sin(deg(pi/4))}) + {(#3)+0.5}*(0,1)$) --
            (${(#1)+0}*(1,0) + {(#2)+0}*({sqrt(2)*cos(deg(pi/4))},{sqrt(2)*sin(deg(pi/4))}) + {(#3)+0.5}*(0,1)$) -- cycle;
		\draw[pth] 	
			(${(#1)+0}*(1,0) + {(#2)+0}*({sqrt(2)*cos(deg(pi/4))},{sqrt(2)*sin(deg(pi/4))}) + {(#3)}*(0,1)$) --
            (${(#1)+1}*(1,0) + {(#2)+0}*({sqrt(2)*cos(deg(pi/4))},{sqrt(2)*sin(deg(pi/4))}) + {(#3)}*(0,1)$);	
       }

	\brxy(0:1:3);\brxy(0:2:3);\brxy(1:2:3);\brxz(2:3:2);\brxz(3:3:2);
	\bryz(0:0:2);\brxz(0:1:2);\bryz(1:1:2);\brxy(1:1:2);\brxz(1:2:2);\brxy(2:2:2);\bryz(2:2:2);
	\bryz(0:0:1);\brxz(0:1:1);\brxz(1:1:1);\bryz(2:1:1);\brxz(2:2:1);\bryz(3:2:1);\brxz(3:3:1);
	\brxy(0:0:1);\brxy(1:0:1);\brxy(2:0:1);\brxy(2:1:1);\brxy(3:0:1);\brxy(3:1:1);\brxy(3:2:1);
	\brxz(0:0:0);\brxz(1:0:0);\brxz(2:0:0);\brxz(3:0:0);\bryz(4:0:0);\bryz(4:1:0);\bryz(4:2:0);

	\draw[->] (-0.5,0) -- (8,0) node[right] {$t$};
	\draw[->] (0,-0.5) -- (0,6.5) node[above] {$x$};

	\draw[-,line width=2pt, color=purple!80!black] (3,6) -- (3,-1) node[below] {$t=3$};
	\node[intpt] at (3,0)  {}; 	\node[intpt] at (3,2)  {}; \node[intpt] at (3,4)  {};

\end{tikzpicture}
\caption{Modified hexagon and corresponding up-right path
  configuration. The dots represent the particles at time $t = 3$ and we have $x^t_1 = 0$, $x^t_2 = 2$, and $x^t_3 = 4.$} \label{fig:paths}
\end{figure}
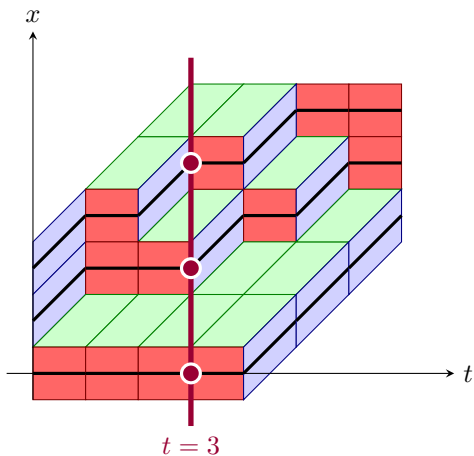

In the same way the Hahn orthogonal polynomial ensemble arises in the analysis of uniform lozenge tilings, our measures are related to the
$q$-Racah orthogonal polynomials. In this sense, the model goes all the way up to the top of the Askey scheme \cite{KLS}. The correspondence
goes as follows: for a fixed section $t,$ configurations $x^t_1 < x^t_2< \cdots < x^t_N$ form an $N$-point process. Under a suitable change
of variables this point process has the same distribution as the $q${\it-Racah orthogonal polynomial ensemble} for a set of parameters that
depend on the location of the vertical slice and the size of the hexagon. We elaborate more on this connection in Section \ref{sec:ensemble}.

An interesting aspect of this two-parameter family of probability measures is its various degenerations. For example, the uniform measure
on tilings is recovered in the limit $\kappa \rightarrow 0$ and $q \rightarrow 1$. Other interesting degenerations include $\kappa
\rightarrow 0$, in which case the weight becomes proportional to $q^{-V}$, where $V$ is the number of boxes in the $3$-D interpretation). On
one hand, these limits correspond to some arrows in the degeneration cascades in the Askey scheme of hypergeometric and basic hypergeometric
orthogonal polynomials. On the other hand, they seem to correspond to the degeneration cascades in Sakai's classification scheme of discrete
Painlev\'e equations \cite{Sak:2001:RSAWARSGPE}, as shown in Figure~\ref{fig:sakai-scheme}. Specifically, in \cite{BB} it was shown that gap
probabilities of the form~\eqref{eq:gap_prob} for many examples of discrete orthogonal polynomial ensembles can be computed using a certain
recurrence procedure that is essentially equivalent to the difference and $q$-difference discrete Painlev\'e equations; some cases are
labeled on Figure~\ref{fig:sakai-scheme}. This correspondence has been extended in \cite{K} to the $q$-Hahn case that corresponds to the
$q$-P$\left(E_6^{(1)}/A_{2}^{(1)}\right)$ discrete Painlev\'e equation. The $q$-Racah case considered in the present paper corresponds to the
$q$-P$\left(E_7^{(1)}/A_{1}^{(1)}\right)$ discrete Painlev\'e equation. Although we do not study these degenerations in detail (we plan to
consider this question separately), in Section~\ref{sec:degeneration_from_the_q_racah_to_the_q_hahn_case} we show that the weight
degeneration from the $q$-Racah case to the $q$-Hahn case is completely consistent with the degeneration of the $A_{1}^{(1)}$ surface
(with $E_{7}^{(1)}$ symmetry) into the $A_{2}^{(1)}$ surface (with $E_{6}^{(1)}$ symmetry)  in Sakai's approach.
%
\begin{figure}[ht]{\tiny\vskip-0.2in
	\begin{tikzpicture}[>=stealth,scale=0.85]
	\node (e8e) at (0,4) {$\left(E_{8}^{(1)}\right)^{e}$};
	\node (a1qa) at (15,4) {$\left(A_{1}^{(1)}\right)^{q}_{|\alpha|^{2}=8}$};
	\node (e8q) at (0,2) {$\left(E_{8}^{(1)}\right)^{q}$};
	\node (e7q) at (1.8,2) {$\left(E_{7}^{(1)}\right)^{q}$};
	\node (e6q) at (3.6,2) {$\left(E_{6}^{(1)}\right)^{q}$};
	\node (d5q) at (5.4,2) {$\left(D_{5}^{(1)}\right)^{q}$};
	\node (a4q) at (7.2,2) {$\left(A_{4}^{(1)}\right)^{q}$};
	\node (a21q) at (9.5,2) {$\left((A_{2} + A_{1})^{(1)}\right)^{q}$};
	\node (a11q) at (12.5,2) {$\left((A_{1} + A_{1})^{(1)}\right)^{q}$};
	\node  at (12,1.6) {$|\alpha|^{2}=14$};
	\node (a1q) at (15,2) {$\left(A_{1}^{(1)}\right)^{q}$};
	\node  at (14.4,1.6) {$|\alpha|^{2}=4$};
	\node (a0q) at (17,2) {$\left(A_{0}^{(1)}\right)^{q}$};
	\node (e8d) at (1.8,0) {$\left(E_{8}^{(1)}\right)^{\delta}$};
	\node (e7d) at (3.6,0) {$\left(E_{7}^{(1)}\right)^{\delta}$};
	\node (e6d) at (5.4,0) {$\left(E_{6}^{(1)}\right)^{\delta}$};
	\node (d4d) at (7.2,0) {$\left(D_{4}^{(1)}\right)^{c,\delta}$};
	\node (a3d) at (9.5,0) {$\left(A_{3}^{(1)}\right)^{c,\delta}$};
	\node (a11d) at (12.5,0) {$\left(2A_{1}^{(1)}\right)^{c,\delta}$};
	\node (a1d) at (15,0) {$\left(A_{1}^{(1)}\right)^{c,\delta}$};
	\node (a0d) at (17,0) {$\left(A_{0}^{(1)}\right)^{c}$};
	\node (a2d) at (12.5,-2) {$\left(2A_{1}^{(1)}\right)^{c,\delta}$};
	\node (a1da) at (15,-2) {$\left(A_{1}^{(1)}\right)^{c,\delta}$};
	\node (a0da) at (17,-2) {$\left(A_{0}^{(1)}\right)^{c}$};
	\draw[->] (e8e) -> (e8q);	\draw[->] (a1qa) -> (a0d); 	
	\draw[->] (e8q) -> (e7q); 	\draw[->] (e8q) -> (e8d); 	
	\draw[->] (e7q) -> (e6q); 	\draw[->] (e7q) -> (e7d); 	
	\draw[->] (e6q) -> (d5q); 	\draw[->] (e6q) -> (e6d); 	
	\draw[->] (d5q) -> (a4q); 	\draw[->] (d5q) -> (d4d); 	
	\draw[->] (a4q) -> (a21q); 	\draw[->] (a4q) -> (a3d); 	
	\draw[->] (a21q) -> (a11q); \draw[->] (a21q) -> (a11d); \draw[->] (a21q) -> (a2d);	
	\draw[->] (a11q) -> (a1q); 	\draw[->] (a11q) -> (a1d); 	\draw[->] (a11q) -> (a1qa); 	\draw[->] (a11q) -> (a1da);
	\draw[->] (a1q) -> (a0q); 	\draw[->] (a1q) -> (a0d);	\draw[->] (a1q) -> (a0da);
	\draw[->] (e8d) -> (e7d);
	\draw[->] (e7d) -> (e6d);
	\draw[->] (e6d) -> (d4d);
	\draw[->] (d4d) -> (a3d);	
	\draw[->] (a3d) -> (a11d);	\draw[->] (a3d) -> (a2d);	
	\draw[->] (a11d) -> (a1d);	\draw[->] (a11d) -> (a1da);	
	\draw[->] (a1d) -> (a0d);	\draw[->] (a1d) -> (a0da);	
	\draw[->] (a2d) -> (a1da);	\draw[->] (a1da) -> (a0da);	
		\node [red] at ($(e7q.north) + (0,0.2)$) {$q$-Racah};
		\node [red] at ($(e6q.north) + (0,0.2)$) {$q$-Hahn};
		\node [red] at ($(e7d.south) + (0,-0.1)$) {Racah};
		\node [red] at ($(e6d.south) + (0,-0.1)$) {Hahn};
		\node [red] at ($(d5q.north) + (0,0.5)$) {$q$-Krawtchouk};
		\node [red] at ($(d5q.north) + (0,0.8)$) {Little $q$-Jacobi};
		\node [red] at ($(d4d.south) + (0,-0.6)$) {Meixner};
		\node [red] at ($(d4d.south) + (0,-0.9)$) {Krawtchouk};
		\node [red] at ($(a3d.south) + (0,-0.9)$) {Charlier};

		\node [blue] at ($(d5q.north) + (0,0.2)$) {q-$P_{\text{VI}}$};
		\node [blue] at ($(a4q.north) + (0,0.2)$) {q-$P_{\text{V}}$};
		\node [blue] at ($(a21q.north) + (0,0.2)$) {q-$P_{\text{IV}}$, q-$P_{\text{III}}$};
		\node [blue] at ($(a1qa.north) + (-0.5,0.2)$) {q-$P_{\text{I}}$};
		\node [blue] at ($(d4d.south) + (0,-0.1)$) {$P_{\text{VI}}$,\! d-$P_{\text{V}}$};
		\node [blue] at ($(a3d.south) + (-0.2,-0.1)$) {$P_{\text{V}}$,\! d-$P_{\text{IV}}$};
		\node [blue] at ($(a3d.south) + (0,-0.5)$) {d-$P_{\text{III}}$};
		\node [blue] at ($(a11d.south) + (0,-0.1)$) {$P_{\text{III}}$};
		\node [blue] at ($(a11d.south) + (0.1,-0.5)$) {alt.\! d-$P_{\text{II}}$};
		\node [blue] at ($(a1d.south) + (0,-0.1)$) {$P_{\text{III}}$};
		\node [blue] at ($(a0d.south) + (0,-0.1)$) {$P_{\text{III}}$};
		\node [blue] at ($(a2d.south) + (0,-0.1)$) {$P_{\text{IV}}$,\! d-$P_{\text{II}}$};
		\node [blue] at ($(a1da.south) + (0,-0.1)$) {$P_{\text{II}}$,\! alt.d-$P_{\text{I}}$};
		\node [blue] at ($(a0da.south) + (0,-0.1)$) {$P_{\text{I}}$};
	\end{tikzpicture}}
	\caption{The degeneration cascade for the symmetry-type classification of Painlev\'e equations}
	\label{fig:sakai-scheme}
\end{figure}
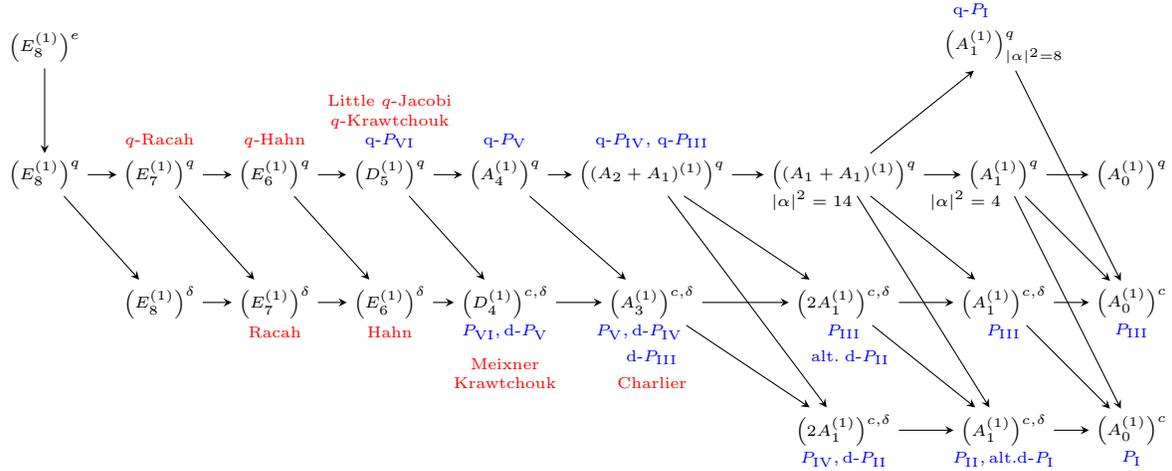

%

\begin{figure}[h]
    \includegraphics[width=0.35\textwidth]{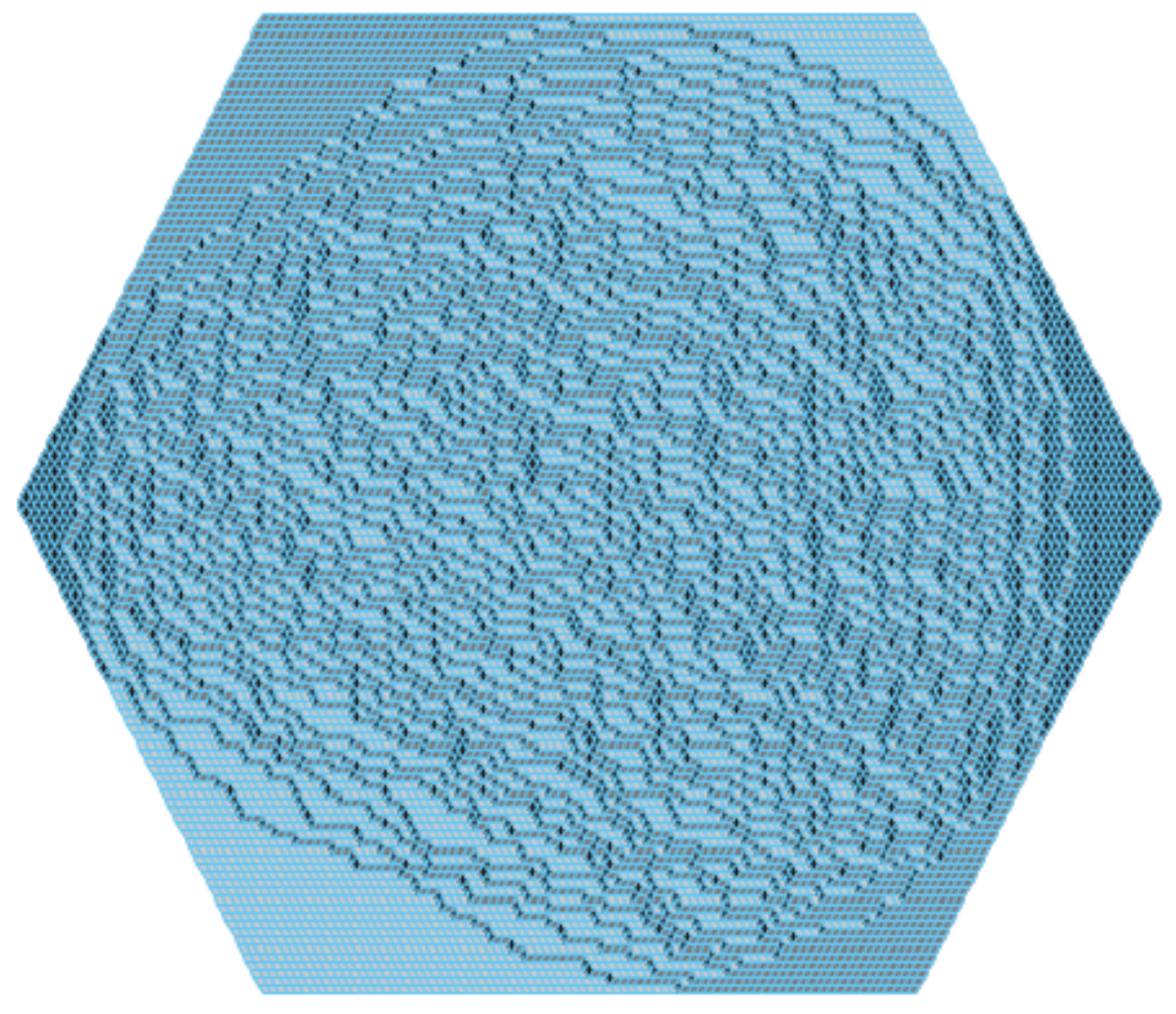}
\caption{ A simulation of a tiling for a hexagon with the sides $a =
  60$, $b = 80$, $c =60$ and parameters $\kappa^2=0,001,$ $q=0.995.$}
  \label{tiling_sample}
\end{figure}

We also want to point out that the $q$-Racah tiling model is a source of rich and interesting structures that are worth investigating. In
particular, the asymptotic behavior of the height function of the $q$-Racah tiling model when the sides of the hexagon become large and
simultaneously $q\rightarrow 1, \kappa\rightarrow \kappa_0$, where $\kappa_0 \in (0,1)$ is fixed, was studied in \cite{DK}, (see Figure~\ref{tiling_sample} for a sample tiling in this case), where
it was proved that there exists a deterministic limit shape $\hat{h}$ and the
random height functions $h$ concentrate near it with high probability as the parameters of the model scale to their critical values. An
important feature of that model is that the limit shape develops \emph{frozen facets} where the height function is linear. In addition, the
frozen facets are interpolated by a connected disordered \emph{liquid region}. In terms of the tiling, a frozen facet corresponds to a region
where asymptotically only one type of lozenge is present, and in the liquid region one sees lozenges of all three types, see
Figure~\ref{tiling_sample}. Similar concentration phenomena for the random height function in the case of the uniform measure and the measure
proportional to $q^{-V}$ are well-understood. In particular, in these cases convergence of the random height function to a deterministic
function for a large class of domains was established in \cite{JPS,CKP,D,DMB,KO}.

The results of the present paper predict the appearance of the Painlev\'e transcendents in the limit regime for the fluctuations of the
height function near the boundary of the limit shape.

\subsection{Moduli spaces of $q$-connections} 
\label{sub:moduli_spaces_of_q_connections}

Our approach is based on the ideas introduced in \cite{BB} and \cite{BA1}. First, using Discrete Riemann-Hilbert Problem formalism
of \cite{B1,B2},  we express the gap probability function in terms of the matrix entries of a sequence of matrices $A_s(z)$ of a
certain form. We then describe the general moduli space of matrices of this form (equivalently, the moduli space of $q$-connections)
and show that its smallest compactification is isomorphic to a $A_{1}^{(1)}$-surface in Sakai's approach. The evolution
$A_s(z)\mapsto A_{s+1}(z)$ is given by an isomonodromy transformation that can be thought of as an isomorphism between two
different surfaces in the $A_{1}^{(1)}$-family, and so it is not surprising that it is given by a discrete
$q$-P$\left(E_7^{(1)}/A_{1}^{(1)}\right)$ Painl\'eve equation. We first identify this equation indirectly through the action of the
isomonodromic dynamics on the parameters of the moduli space, and then show how to change coordinates to explicitly
transform this equation into the standard form.

One new and interesting aspect of the $q$-Racah case is a certain
involutive symmetry of the problem. Following the ideas
of D.~Arinkin and A.~Borodin, see also \cite{OrmRai:2017:ASDLPFPV},
we formalize this involutive symmetry structure via the notion of an \emph{elliptic connection}.

\begin{definition} Let $\mathcal{E}$ be a symmetric bi-quadratic curve in $\mathbb{P}^1\times\mathbb{P}^1$ i.e., $\mathcal{E}$ a zero locus of a symmetric
bi-degree $(2,2)$ polynomial. Note that generically $\mathcal{E}$ is elliptic. An \emph{$\mathcal{E}$-connection} (or an \emph{elliptic connection}) is a pair
$(\mathcal{L}, \mathcal{A}),$ where $\mathcal{L}$ is a vector bundle on $\mathbb{P}^1$, and where for any point $(x,y)\in \mathcal{E}$, we have a map
$\mathcal A(x,y):\mathcal{L}_y \to \mathcal{L}_x$ such that $\mathcal{A}(x,y)$ is a rational function of $(x,y)\in \mathcal{E}$ satisfying the involutivity
condition $\mathcal{A}(y,x)=\mathcal{A}(x,y)^{-1}$. \end{definition}

For our purposes we need to consider a degenerate case when $\mathcal{E}$ is a nodal rational curve. Namely, let
$u, q\in (0, 1)$ be two fixed parameters, and let the curve $\mathcal E_u$ be given by the following equation
(in the affine chart $\mathbb{C}^{1}\times \mathbb{C}^{1}\subset \mathbb{P}^{1}\times \mathbb{P}^{1}$):
\begin{equation*}
	\mathcal E_u:\quad (x-qy)(y-qx)=\frac{u^2}{q^2}(1-q)^2(1+q).
\end{equation*}
This curve has the following rational parameterization in terms of a parameter $z\in\mathbb{P}^{1}$:
\begin{equation*}
	x(z)=q^{-1}z+u^2/z,\quad y(z)=z+u^2/(q z) = x(qz).
\end{equation*}
In this way we can identify $\mathcal{A}(x,y)=\mathcal{A}(z)$, $\mathcal{A}(y,x)=\mathcal{A}(u^2/z)$, and the mapping
$\mathcal A(x,y):\mathcal{L}_y \to \mathcal{L}_x$ induces the mapping $\mathcal{A}(z): \mathcal{L}_{z}\to \mathcal{L}_{qz}$.
The latter mapping is the usual definition of a $q$-connection, but the above formalism allows us to incorporate into it the
symmetry condition.

\begin{definition}
 We say that a point $z_0 \in \mathbb{P}^1$ is a \emph{pole} of $\mathcal{A}$ if $\mathcal{A}(z)$ is not regular at $z=z_0$. We say that $z_0\in
\mathbb{P}^1$ is a \emph{zero} of $\mathcal{A}$ if the map $\mathcal{A}^{-1}(z) : \mathcal{L}_{q^{-1}z} \rightarrow \mathcal{L}_z$ is not regular at
$z = z_0$. Note that $\mathcal{A}$ can have a zero and a pole at the same point.
\end{definition}

\begin{definition}\label{def:mod}
Suppose $\mathcal{R}: \mathcal{L}\rightarrow\hat{\mathcal{L}}$ is a rational isomorphism between two vector bundles on $\mathcal{E}_{u}\simeq
\mathbb{P}^1$. We say that $\hat{\mathcal{L}}$ is a \emph{modification} of $\mathcal{L}$ on a finite set $S\subset \mathcal{E}_u$ if
$\mathcal{R}(z)$ and $\mathcal{R}^{-1}(z)$ are regular outside $S$. We call $\hat{\mathcal{L}}$ an \emph{upper modification} of $\mathcal{L}$ if
$\mathcal{R}$ is regular (then $\mathcal{L}$ is called a \emph{lower modification} of $\hat{\mathcal{L}}$).
A $\mathcal{E}_u$-connection $(\mathcal{L},\mathcal{A}(z))$ induces a $\mathcal{E}_u$-connection
$(\hat{\mathcal{L}}, \hat{\mathcal{A}}(z))$ that we also call a \emph{modification} of $(\mathcal{L},\mathcal{A}(z))$.
\end{definition}

A class of $\mathcal{E}_u$-connections that we consider depends on $8$ complex parameters $(z_1, z_2,z_3,z_4,z_5,z_6, d_1, d_2)$. After
choosing a trivialization of $\mathcal{L}$ over the affine chart $\mathbb{C}\subset \mathbb{P}^{1}$, the matrix $A(z)$ of the connection
$\mathcal{A}$ has the following form (see Section~\ref{sec:ell}):
\begin{equation*}
	A(z)= \frac{1}{P(z)}\begin{bmatrix}
		b_{11}(z) & \frac{b_{12}(z)}{z}\\ b_{21}(z) & b_{22}(z)
	\end{bmatrix} ,\quad  b_{21}(0)= 0,
\end{equation*}
where $b_{ij}(z)$ are polynomials with  $\deg(b_{11}(z))\leq 6$, $\deg(b_{12}(z))\leq 8$, $\deg(b_{21}(z))\leq 5$, $\deg(b_{22}(z))\leq 6$, and
\begin{equation*}
	\det A(z)= \frac{Q(z)}{P(z)},\qquad
\begin{aligned}
	P(z) &= (z-z_{1})(z - u^{2}/z_{2})(z-z_{3})(z - u^{2}/z_{4})(z-z_{5})(z - u^{2}/z_{6}),\\
	Q(z) &= \frac{z_{1}z_{3}z_{5}}{z_{2}z_{4}z_{6}}(z - u^{2}/z_{1})(z-z_{2})(z - u^{2}/z_{3})(z-z_{4})(z - u^{2}/z_{5})(z - z_{6}).
\end{aligned}
\end{equation*}
		We also require that $A(z)$ satisfies the \emph{asymptotic condition}
		\begin{equation*}
			S\left(\frac{z}{q} + \frac{u^{2}}{z}\right) A(z) S^{-1}\left(z + \frac{u^{2}}{q z}\right) \sim
			\begin{bmatrix}	d_{1} & 0 \\ 0 & d_{2} \end{bmatrix},
				\qquad\text{where}\quad
			S(z) = \begin{bmatrix}1 & 0 \\ 0 & z	\end{bmatrix},
		\end{equation*}
      and the \emph{involution condition}
\begin{equation*}
	A(u^2/z)=A^{-1}(z).
\end{equation*}

We consider $A(z)$ modulo gauge transformations of the form
\begin{equation*}
	\hat{A}(z) = R\left(\frac{z}{q} + \frac{u^{2}}{z}\right)A(z)R^{-1}\left(z + \frac{u^{2}}{q z}\right),\qquad R(z) = \begin{bmatrix}
		r_{11}(z) & r_{12}(z) \\ 0 & r_{22}(z)
	\end{bmatrix},
\end{equation*}
where $r_{ij}(z)$ are polynomials with $\deg(r_{11}(z))=\deg(r_{22}(z))=0$ and  $\deg(r_{12}(z))\leq 1$.

\begin{lemma}\label{lem:mod}
Under certain non-degeneracy conditions on the parameters $\lambda=(z_1,z_2,\dots,z_6, d_1, d_2)$ of a $\mathcal E_u$-connection
$\mathcal{A},$ there exits its unique modification $\overline{\mathcal{A}}$ of type
$\overline\lambda=(z_1,q z_2, z_{3},q z_{4}, z_{5},z_6, q^{-1} d_1, q^{-1} d_2)$.
\end{lemma}

Let us assume that the parameters $(z_1,\dots,z_6, d_1, d_2)$ are generic; the precise meaning of this condition is explained in
Section~\ref{sec:ell}. We show that the moduli space $M_{\lambda}$ of $q$-connections of type $\lambda=(z_1,z_2,\dots,z_6, d_1, d_2)$ modulo
$q$-gauge transformations is two-dimensional and its smallest smooth compactification can be identified with $\mathbb{P}^{1} \times
\mathbb{P}^{1}$ blown-up at eight points; more precisely, it is a Sakai surface of type $A_{1}^{(1)}$. We denote the parameters on this
surface by $(f, g),$ they are described in \eqref{eq:chvars-fg-xy-a1} in terms of the usual spectral coordinates.

\begin{theorem}
Consider the modification of $\mathcal{L}$ to $\hat{\mathcal{L}}$ from Lemma~\ref{lem:mod} that shifts
\begin{equation*}
	z_{2}\rightarrow q z_{2}, \quad z_{4}\rightarrow q z_{4},\quad d_{1}\rightarrow q^{-1}d_{1},
	\quad d_{2}\rightarrow q^{-1}d_{2}.	
\end{equation*}
Then this modification defines a regular morphism between two moduli spaces $M_{\lambda}$ and $M_{\overline\lambda}$.
Moreover, the coordinates $(\overline{f},\overline{g})$ on the moduli space $M_{\overline\lambda}$ are related to $(f, g)$ by the
$q$-P$\left(E_7^{(1)}/A_{1}^{(1)}\right)$ Painlev\'e equation
\begin{equation}\label{eq:painl}
	\left\{
	\begin{aligned}
		\frac{\left(fg - \frac{\kappa_{1}}{\kappa_{2}}\right)(\overline{f} g - \frac{\kappa_{1}}{q \kappa_{2}})}{(fg - 1) (\overline{f}g - 1)} &=
		\frac{\left(g - \frac{\nu_{5}}{\kappa_{2}}\right) \left(g - \frac{\nu_{6}}{\kappa_{2}}\right) \left(g - \frac{\nu_{7}}{\kappa_{2}}\right)
		\left(g - \frac{\nu_{8}}{\kappa_{2}}\right)}{ \left(g - \frac{1}{\nu_{1}}\right) \left(g - \frac{1}{\nu_{2}}\right) \left(g - \frac{1}{\nu_{3}}\right)
		\left(g - \frac{1}{\nu_{4}}\right)},\\
		\frac{\left(fg - \frac{\kappa_{1}}{\kappa_{2}}\right)(f \underline{g} - \frac{q \kappa_{1}}{\kappa_{2}})}{(fg - 1) (f\underline{g} - 1)} &=
		\frac{\left(f - \frac{\kappa_{1}}{\nu_{5}}\right) \left(f - \frac{\kappa_{1}}{\nu_{6}}\right) \left(f - \frac{\kappa_{1}}{\nu_{7}}\right)
		\left(f - \frac{\kappa_{1}}{\nu_{8}}\right)}{ \left(f - \nu_{1}\right) \left(f - \nu_{2}\right) \left(f - \nu_{3}\right)
		\left(f - \nu_{4}\right)}.
	\end{aligned}
	\right.,
\end{equation}
where we have the following matching of parameters:
\begin{equation*}
	\nu_{1} = \frac{1}{z_{6}},\,
				\nu_{2} = \frac{1}{z_{1}},\,
				\nu_{3} = \frac{1}{z_{3}},\,
				\nu_{4} = \frac{1}{z_{5}},\,
				\nu_{5} = \frac{u z_{4}}{z_{2}},\,
				\nu_{6} = u,\,
				\nu_{7}= \frac{d_{1} z_{4} z_{6}}{u},\,
				\nu_{8}=  \frac{d_{2} z_{4} z_{6}}{u},\,
				\kappa_{1} = \frac{u}{z_{2}},\,
						\kappa_{2} = \frac{z_{4}}{u}.
\end{equation*}
\end{theorem}

\begin{remark} The form~\eqref{eq:painl} of the standard  $q$-P$\left(E_7^{(1)}/A_{1}^{(1)}\right)$ equation here follows the recent survey
	monograph~\cite{KajNouYam:2017:GAOPE} (equation~(8.7) in 8.1.3). It is given as two maps $(f,g)\mapsto (\overline{f},g)$ and
	$(f,\underline{g})\mapsto(f,g)$, which reflects the QRT origin of this equation, but it is easy to rewrite it as a mapping
	$(f,g)\mapsto(\overline{f},\overline{g})$.
	
	We also want to point out that this equation was originally obtained by Grammaticos and Ramani \cite{GraRam:1999:THFTDPVIO} (equations (14a) and (14b)),
	where it is called the the asymmetric $q$-$P_\text{VI}$ equation.
	
\end{remark}		
%
%


We believe that important new aspects of the present paper are the following. First, it is a good illustration of the power
of Sakai's geometric theory for applications. Here we show how from just the minimal knowledge of
the singularity structure of the connection and the evolution of parameters we can identify our dynamics with the standard discrete
Painlev\'e dynamics and produce the required non-trivial change of coordinates that significantly simplifies the further computations. On one hand, our approach is algorithmic and it is
adaptable to other applications, but on the other hand it uses the full power of algebro-geometric theory of discrete Painlev\'e equations.
Second, the $q$-Racah weight that we consider is at the top of the degeneration cascade, so other models can be obtained from it through
degenerations, as we showed for the $q$-Hahn limit. Further, the $q$-Racah case was not considered in \cite{BB}, so we needed to adopt the computation of the gap probabilities
through the discrete Riemann-Hilbert problems from \cite{BB} to work for this model.

%
%
\subsection*{Acknowledgements} 
\label{sub:acknowledgements}
The authors want to thank Dima Arinkin, Alexei
Borodin, Kenji Kajiwara, and Tomoyuki Takenawa for many helpful discussions and suggestions.
A part of the work was completed
when the authors attended the 2017 IAS PCMI Summer Session on Random Matrices, and we are
grateful to the organizers for the hospitality and support. AK was partially supported by the NSF
grant DMS-1704186. AD was partially supported by the UNCO grant SSI-2018.



\section{The $q$-Racah Orthogonal Polynomial Ensemble} 
\label{sec:ensemble}
\subsection{Orthogonal Polynomials} 
\label{sub:orthogonal_polynomials}

Let $\mathfrak X$ be a finite subset of $\mathbb C$ such that $\mbox{card}(\mathfrak X)=M+1<\infty$ and $\omega: \mathfrak{X} \rightarrow \mathbb R_{>0}$ be
any function. Using $\omega$ as a weight function, we can define an inner product on the space $\mathbb{C}[z]$ of complex polynomials via
\begin{equation*}
	\left(f,g\right )_{\omega}:=\sum\limits_{x\in\mathfrak X}^{ }f(x)g(x)\omega(x),\quad f,g\in \mathbb{C}[z].
\end{equation*}
Given this inner product, a set $\{ P_n\}_{n = 0}^M$ of complex polynomials is called a collection of \emph{orthogonal polynomials associated to the weight
function} $\omega$ if \begin{itemize} \item $P_n$ is a polynomial of degree $n$ for all $n = 1, \dots, M$ and $P_0 \equiv const$; \item if $m \neq n$ then $(
P_m, P_n)_{\omega} = 0$. \end{itemize} We always take $P_n$ to be monic, i.e. $P_n(z) = z^n + \text{lower order terms}$.

It is clear that a collection of orthogonal polynomials $\{P_n\}_{n = 0}^M$ associated to $\omega$ and satisfying the condition $(P_n, P_n)_\omega \neq 0$ for
all $n = 0, \dots, M$ exists if and only if the restriction of $(\cdot, \cdot)_\omega$ to the space $\mathbb{C}[z]^{\leq d}$ of polynomials of degree at most
$d$ is nondegenerate for all $d = 0, \dots, M$. If this condition holds we say that the weight function $\omega$ is \emph{nondegenerate}, and in that case it
is clear that the collection $\{P_{n}\}_{n = 0}^{M}$ (with the monic normalization) is unique.

\begin{definition}\label{def:ensemble} Fix $N\in \{1,\dots, M+1\}.$ Under the above assumptions, \emph{an $N$-orthogonal discrete polynomial ensemble on
$\mathfrak X$ with the weight function $\omega$} is a probability distribution on $N$-tuples $(\chi_1,\dots, \chi_N),$ $\chi_i\in \mathfrak X,$ that is
defined by
\begin{equation}
	\mathbb{P}(\chi_1,\dots ,\chi_N)= \frac{1}{Z(N,M)} \prod_{1\leq i<j \leq N} \left(\chi_i-\chi_j\right)^2 \cdot
\prod_{i=1}^{N}{\omega(\chi_i)},
\end{equation}
where $Z(N, M)$ is the usual normalization constant.
\end{definition}

It is well known (see, e.g., \cite{Joh:2006:RMADP} or \cite{K1}) that such an ensemble is a determinantal point
process whose correlation kernel $K(x,y)$ can be written in terms of the orthogonal polynomials,
\begin{equation}\label{eq:kernel}
K(x,y)=\sqrt{\omega(x)\omega(y)}\sum_{i=0}^{N-1}\frac{P_i(x)P_i(y)}{(P_i, P_i)_\omega} =
\begin{dcases}
        \sqrt {\omega(x) \omega(y)} \frac{\phi(x)\psi(y)-\psi(x)\phi(y)}{x-y},
            & \quad x\neq y, \\
     \omega(x)(\phi'(x)\psi(x)-\phi(x) \psi'(x)),&\quad x=y,
      \end{dcases}	
\end{equation}
where $\phi(z)=P_N(z)$ and $\psi(z)=(P_{N-1},P_{N-1})^{-1}_{\omega}\cdot P_{N-1}(z)$. The second equality here follows from the observation that
$K(x, y)$ is equal to the product of $\sqrt{\omega(x)\omega(y)}$ with the $N^{\text{th}}$ Christoffel-Darboux kernel for this system of orthogonal
polynomials, see \cite{Sze:1967:OP}. 


Let us parametrize the set $\mathfrak{X}$ as $\mathfrak{X}=\{\pi_x\}^M_{x=0}$, where $\pi_x<\pi_{x+1},$ $x=0,\dots, M$.
For any $s\in \mathbb{N}$, $N\leq s \leq M$, let $\mathfrak{Z}_s=\{\pi_j\}_{j=0}^{s-1}$ and let
$\mathfrak{N}_s = \mathfrak{X} \setminus \mathfrak{Z}_{s} = \{\pi_j\}_{j=s}^{M}$. It is well-known, see \cite{BB} or \cite{AGZ},
that the so-called \emph{gap probabilities} $D_{s}$ for to this ensemble, defined below, can be expressed as a Fredholm
determinant of the correlation kernel $K(x,y)$ given by \eqref{eq:kernel},
\begin{equation*}
	D_s=\text{Prob}\left(\max\{\pi_{i}\}_{i=1}^{N} < \pi_s\right)=\det(1-K_s), \quad \text{where } K_s=K|_{\mathfrak N_s\times \mathfrak N_s}.
\end{equation*}
These are the quantities that we are interested in computing.

%
%

\subsubsection{$q$-Racah Orthogonal Polynomial Ensemble} 
\label{ssub:_q_racah_orthogonal_polynomial_ensemble}

In this section we recall some basic properties of the $q$-Racah orthogonal polynomials, cf. \cite[Section 3.2]{KS}.

\begin{definition}\label{def:qrw}
Let $q\in(0,1),$ $M\in \mathbb{Z}_{\geq 0}$, $\alpha,\beta,\gamma,\delta\in \mathbb R$ and $\gamma=q^{-M-1}.$ For 
$x = 0,1,\ldots,M$, the $q$-Racah weight function $\omega^{\textup{qR}}(x)$ is defined by
\begin{equation} \label{eq:w_qR}
\omega^{\textup{qR}}(x)=\frac{(\alpha q,\beta\delta q,\gamma q,\gamma \delta q;q)_x}{(q,\alpha^{-1}\gamma \delta q, \beta^{-1}\gamma q, \delta q;q )_x}
\frac{(1-\gamma\delta q^{2x+1})}{(\alpha \beta q)^{x}(1-\gamma\delta q)},
\end{equation}
where $(y_1,\dots, y_i;q)_k:=(y_1;q)_k\cdots (y_i;q)_k$ and $(y;q)_k:=(1-y)(1-yq)\cdots (1-yq^{k-1})$ is the usual $q$-Pochhammer symbol.
\end{definition}
 \begin{remark} \label{rem:qH} The condition $\gamma=q^{-M-1}$ can be replaced by $\alpha=q^{-M-1}$ or $\beta\delta=q^{-M-1}$.
	 Our choice is due to the fact that under the substitutions $\gamma=q^{-M-1}$ and $\delta=0$ the $q$-Racah weight
	 reduces to the $q$-Hahn weight $\omega^{\textup{qH}}(x) = \frac{(\alpha q, q^{-M}; q)_{x}}{(q, \beta^{-1} q^{-M}; q)_{x} (\alpha \beta q)^{x}}$.
\end{remark}
%
\begin{definition} \label{def:qRacah} Fix $N \in \mathbb{Z}_{\geq1}$ and let $\alpha, \beta, \gamma, \delta, q$ and $M$ be as in
Definition~\ref{def:qrw} with $M \geq N-1$. Denote by $\mathfrak{X}^N$ a collection of $N$-tuples of non-negative integers,
\begin{equation*}
\mathfrak{X}^N = \{ (\lambda_1, \dots, \lambda_N) \in \mathbb{Z}^N : 0\leq \lambda_1<\lambda_2<\dots<\lambda_N\leq M \}.	
\end{equation*}
The $q$-Racah ensemble is a probability measure $\mathbb{P}^{\textup{qR}}$ on the set $\mathfrak{X}^N$ that is given by
\begin{equation} \label{eq:distr}
\mathbb{P}^{\textup{qR}}(\lambda_1,\dots ,\lambda_N)= \frac{1}{Z(N,M, \alpha, \beta, \gamma, \delta,q)} \prod_{1\leq i<j \leq N} \left(\sigma(q^{-\lambda_i})-\sigma(q^{-\lambda_j}) \right)^2 \cdot \prod_{i=1}^{N}{\omega^{qR}(\lambda_i)},
\end{equation}
where  $\sigma(z) = z + \gamma\delta q z^{-1}$ and $Z(N,M, \alpha, \beta, \gamma, \delta,q)$ is the usual probabilistic normalization constant.
\end{definition}
For $\mathbb{P}^{\textup{qR}}$ to be an actual probability measure, expressions in~\eqref{eq:distr} have to be non-negative, and this is not
necessarily the case for a generic choice of parameters. Thus, some restrictions on the space of parameters have to be imposed and we make one
such possible choice in the following assumption.

\begin{assumption}\label{assume:ParSetQR}
We assume that parameters $\alpha, \beta, \gamma, \delta, q \in \mathbb{R}$ and $M, N \in \mathbb{Z}$ are such that
$$ M \geq N-1 \geq 0, \hspace{2mm} 1 > q > 0, \hspace{2mm} \alpha, \beta > 0, \hspace{2mm} \delta \geq 0, \hspace{2mm} \gamma = q^{-M-1}, \hspace{2mm} 1 > \beta \delta, \hspace{2mm} \beta \geq \gamma, \hspace{2mm} \alpha \geq \gamma.$$
Then  expressions in~\eqref{eq:distr} are non-negative on all of $\mathfrak{X}^N$ and indeed define a probability measure $\mathbb{P}^{\textup{qR}}$.
\end{assumption}

\begin{remark}
Although we chose to consider $q$-Racah ensemble as a probability measure on $N$-tuples of $(\lambda_1, \dots, \lambda_N)$, it can also be viewed as a measure
on $(\sigma(q^{-\lambda_1}), \dots, \sigma(q^{-\lambda_N}))$, to agree with Definition~\ref{def:ensemble}.
\end{remark}

It is well known that $\omega^{\textup{qR}}$ is a nondegenerate weight function.  Orthogonal polynomials $\{ P_n(z) \}_{n = 0}^M$  associated to it are called the \emph{$q$-Racah orthogonal polynomials}. They
satisfy the following orthogonality relation, written in the argument
$\sigma(q^{-x}) := q^{-x} + \gamma \delta q^{x+1}$:
\begin{equation}\label{eq:QRorthog}
\begin{split}
&\sum_{ x = 0}^M \omega^{\textup{qR}}(x) P_m\left(\sigma(q^{-x})\right)
P_n\left(\sigma(q^{-x})\right) = c_n \cdot \delta_{mn}, \mbox{ where }\\
&c_n = \frac{(\gamma \delta q^2, \alpha^{-1} \beta^{-1} \gamma, \alpha^{-1} \delta, \beta^{-1}; q)_\infty}{(\alpha^{-1} \gamma \delta q, \beta^{-1}\gamma q, \delta q, \alpha^{-1} \beta^{-1} q^{-1}; q)_\infty } \frac{(1 - \alpha \beta q) (\gamma \delta q)^n}{ ( 1 - \alpha \beta q^{2n+1})} \frac{(q, \beta q, \alpha \delta^{-1} q, \alpha \beta \gamma^{-1} q; q)_n }{(\alpha \beta q, \alpha q, \beta \delta q, \gamma q;q)_n}.
\end{split}
\end{equation}
A connection between $q$-Racah ensemble and the tiling model described in Section~\ref{sub:tiling_model} is given by the following Theorem, see
\cite{BGR}.
\begin{theorem}\label{TilingtoParticle}
Consider the tiling of a hexagon with side lengths $a,b,c$. Let $N=a, T=b+c, S=c$, and let $q \in (0,1)$, $\kappa \in \left[0,q^{(T-1)/2} \right)$.
Fix $t \in \{0,1, \dots, T\}$ and let $(x^t_1, \dots, x^t_N)$ be the corresponding random $N$-point configuration, see Figure~\ref{fig:paths}. Then
\begin{equation*}
	\mathbb P(x^t_1, \dots, x^t_N)=\mathbb{P}^{\textup{qR}}(x^t_1, \dots,x^t_N),
\end{equation*}
where the parameters of the $q$-Racah ensemble are as follows:
\begin{enumerate}
\item for $t<S $, $t < T - S$, and $0\leq x \leq M = t + N - 1$,
\begin{equation*}
\alpha = q^{-S-N}, \quad \beta = q^{S - T - N}, \quad \gamma = q^{-t - N},\quad\delta = \kappa^2 q^{-S + N};	
\end{equation*}
\item for $S- 1 < t < T-S+1$ and $x\leq 0 \leq M = S + N - 1$,
\begin{equation*}
\alpha = q^{-t-N}, \quad \beta = q^{t - T - N}, \quad \gamma = q^{-S- N},\quad\delta = \kappa^2 q^{-t + N};	
\end{equation*}
\item for $T-S+1 < t < S$ and $0\leq x - (t + S - T)\leq M = T-S + N - 1$,
\begin{equation*}
\alpha = q^{-T-N+t}, \quad \beta = q^{- T - N}, \quad \gamma = q^{-T - N+S},\quad\delta = \kappa^2 q^{-T+t+ N};	
\end{equation*}
\item for $S- 1 < t$, $T-S-1 < t$, and $0\leq x - (t + S - T)\leq M = T-t + N - 1$,
\begin{equation*}
\alpha = q^{-T-N+S}, \quad \beta = q^{-S  - N}, \quad \gamma = q^{-T - N+t},\quad\delta = \kappa^2 q^{-T+S + N}.	
\end{equation*}
\end{enumerate}
\end{theorem}

In particular, we can treat the gap probability function for the tiling model as the gap probability function for the $q$-Racah ensemble.


\subsection{Discrete Riemann-Hilbert Problems and Gap Probabilities} 
\label{sub:discrete_riemann_hilbert_problems_and_gap_probabilities}
The connection between \emph{Discrete Riemann-Hilbert Problems} (DRHP) and gap probabilities goes back to \cite{B1,B2, BB}.
In this section we review some relevant results from \cite{BB} and also establishes an easier way (compared to \cite{BB}) to compute gap
probabilities through the solution to the corresponding DRHP.

Let $\mathfrak{X}$ and $\omega$ be as in Section~\ref{sub:orthogonal_polynomials} and define
$w: \mathfrak{X}\to \operatorname{Mat}(2, \mathbb{C})$ in terms of the weight function $\omega$ as
\begin{equation}\label{eq:jump}
	w(x) = \begin{bmatrix}
		0 & \omega(x) \\ 0 & 0
	\end{bmatrix}.
\end{equation}

\begin{definition} An analytic function
	\begin{equation*}
		m: \mathbb{C} \setminus \mathfrak{X} \to \operatorname{Mat}(2, \mathbb{C})
	\end{equation*}
	is a solution of the DRHP $(\mathfrak X, w)$ if $m$ has simple poles at the points of $\mathfrak{X}$ and its residues at these points are given by
	the \emph{residue} (or \emph{jump}) condition
	\begin{equation}\label{eq:jumpCond}
	\residue\limits_{z=x} m(z)=\lim\limits_{z\rightarrow x} \left(m(z)w(x)\right),\quad x\in \mathfrak X.
	\end{equation}
\end{definition}

Let us introduce the notation
\begin{equation}\label{eq:DPEN}
c_n := \left( P_n, P_n\right)_{\omega}, \quad H_{n}(z) : = \sum\limits_{x\in\mathfrak X}\frac{P_{n}(x)\omega(x)}{z - x}, \qquad n = 0, \dots, M.
\end{equation}

The connection between the collection of orthogonal polynomials  $\{P_n(z)\}_{n = 0}^M$ on $\mathfrak X$ with the weight function
$\omega$ and solutions to  DRHP$(\mathfrak X, w)$ was established in \cite{BB}.

\begin{theorem} \cite[Lemma 2.1 and Theorem 2.4]{BB}\label{thm:RHPsol}
	Let $\mathfrak X$  be a finite subset of $\mathbb C$,
	$\operatorname{card}(\mathfrak X)=M+1<\infty$,
	$\omega \colon \mathfrak X\rightarrow \mathbb C$ a nondegenerate weight function,
 	and $w$ given by (\ref{eq:jump}). Then for any $N=1,2,\dots,M$ the DRHP $(\mathfrak X, w)$ has a unique solution
	$m_{\mathfrak X}(z)$ satisfying an asymptotic condition
\begin{equation}\label{eq:DRHP1}
m_{\mathfrak X}(z)\cdot \begin{bmatrix}
	z^{-N} & 0 \\ 0 & z^{M}
\end{bmatrix} = I+O\left( z^{-1} \right)\text{ as } z\rightarrow \infty,
\end{equation}
where $I$ is the identity matrix. This solution is explicitly given by
\begin{equation*}
	m_{\mathfrak X}(z)= \begin{bmatrix}
		P_N(z) & H_N(z) \\ c_{N-1}^{-1} P_{N-1}(z)   & c_{N-1}^{-1}H_{N-1}(z)
	\end{bmatrix},\quad \text{where $c_{n}, H_n$ are as in \eqref{eq:DPEN}}.
\end{equation*}
Since $w(x)$ is nilpotent, $\det m_{\mathfrak X}(z)$ is entire.  Moreover, since $\det m_{\mathfrak X}(z)\to 1$
as $z\to \infty$,  $\det m_{\mathfrak X}(z) \equiv 1$.
\end{theorem}


Recall that for $N\leq s \leq M$, $\mathfrak{Z}_s=\{\pi_j\}_{j=0}^{s-1}$ and
$\mathfrak{N}_s = \mathfrak{X} \setminus \mathfrak{Z}_{s} = \{\pi_j\}_{j=s}^{M}$.
Let
\begin{equation*}
	m_{s}(z) = \begin{bmatrix}
		m_{s}^{11}(z) & m_{s}^{12}(z) \\ m_{s}^{21}(z) & m_{s}^{22}(z)
	\end{bmatrix}
\end{equation*}
be the unique solution of DRHP $(\mathfrak{Z}_s, \omega|_{\mathfrak{Z}_s})$ such that
\begin{equation}\label{eq:m_cond}
m_{s}(z)\cdot \begin{bmatrix}
	z^{-N} & 0\\0 & z^{N}
\end{bmatrix} = I + O(z^{-1})\text{ as } z\rightarrow \infty.
\end{equation}
Note that $m_{s}(z)$ is analytic on $\mathfrak{N}_{s}$.

\begin{lemma}[{\cite[Theorem 3.1(a)]{BB}}] For each $s\in \mathbb{N}$, $N\leq s\leq M$,
  there exists a constant nilpotent matrix $T_s$ such that
\begin{equation}\label{eq:next_m}
m_{s+1}(z)=\left(I+\frac{T_s}{z-\pi_s} \right)m_s(z).
\end{equation}
\end{lemma}

\begin{remark}\label{rem:Ts} Note that any $2\times 2$ nilpotent matrix can be written in the form
	\begin{equation}\label{eq:nilp-par}
		T_{s} = \begin{bmatrix}
			t_{s}^{11} & t_{s}^{12} \\ t_{s}^{21} & - t_{s}^{11}
		\end{bmatrix},\qquad (t_{s}^{11})^{2} + t_{s}^{12} t_{s}^{21} = 0.
	\end{equation}
	As explained in \cite[Proposition 5.5]{BB}, we can assume that $t_{s}^{11}\neq 0$ (and hence $t_{s}^{12}, t_{s}^{21}\neq 0$ as well).
\end{remark}

\begin{proposition}\label{prop:gap} The following formula holds
\begin{equation}\label{eq:first_gap}
\frac{D_{s+1}}{D_s}=\omega(\pi_{s})\cdot \frac{(m_s^{11}(\pi_s))^2}{t_s^{12}}.
\end{equation}
\end{proposition}

\begin{proof}
	From \cite[Lemma~4.1]{BB} it follows that the operator $(1 - K_{s})$ is invertible, $D_{s} = \det(1 - K_{s})\neq 0$, and
	the resolvent $R_{s} = K_{s}(1 - K_{s})^{-1}$ is well-defined. Moreover, the diagonal values of the resolvent
	$R_{s}$ satisfy the important identity (see, for example, \cite[Section~3.4.2]{AGZ})
	\begin{equation*}
		1 + R_{s}(\pi_{s},\pi_{s}) = \frac{\det(1 - K_{s+1})}{\det(1 - K_{s})} = \frac{D_{s+1}}{D_{s}}.
	\end{equation*}
	Finally, in \cite[Theorem~2.3 applied in Situation~2.2]{B2}, it was shown that the diagonal values of the resolvent
	can be computed explicitly by
	\begin{equation}\label{eq:resolvent}
		R_{s}(\pi_{s},\pi_{s}) = - \begin{bmatrix} 0 & \sqrt{\omega(\pi_{s})} \end{bmatrix}
		m_{s}^{-1}(\pi_{s}) m_{s}'(\pi_{s})
		\begin{bmatrix}\sqrt{\omega(\pi_{s})} & 0\end{bmatrix}^{t} =
			-\operatorname{Tr}\left( m_{s}^{-1}(\pi_{s}) m_{s}'(\pi_{s})  \begin{bmatrix}
				0 & \omega(\pi_{s}) \\ 0 & 0\end{bmatrix}\right).
	\end{equation}

From \eqref{eq:next_m} taking residue at $z=\pi_s$ we get
\begin{equation}\label{eq:ttt}
	T_{s} \begin{bmatrix} m^{11}_{s}(\pi_{s}) \\
              m^{21}_{s}(\pi_{s}) \end{bmatrix} = 0,\quad\text{in
              particular, }\frac{m_s^{11}(\pi_s)}{m_s^{21}(s)}=-\frac{t_s^{12}}{t_s^{11}}.
\end{equation}
Second, multiplying by $w(\pi_{s})$ we get $T_{s} m_{s}(\pi_{s})w(\pi_{s}) = 0.$
Note that since $t^{21}_s\neq0$ we have
\begin{equation*}
\operatorname{Ker}(T_{s}) =
			\operatorname{Span}_{\mathbb{C}} \left\{T_{s}\begin{bmatrix} 1 \\0\end{bmatrix} =
			\begin{bmatrix} t^{11}_{s} \\t^{21}_{s}\end{bmatrix} \right\},\quad
			\begin{bmatrix} m^{11}_{s}(\pi_{s}) \\ m^{21}_{s}(\pi_{s}) \end{bmatrix} = \lambda 	\begin{bmatrix} t^{11}_{s} \\t^{21}_{s}\end{bmatrix}.	
\end{equation*}
Now, on the one hand, using the DRHP residue condition \eqref{eq:jumpCond}
and  \eqref{eq:next_m}, we get
	\begin{equation}\label{eq:resPs}
	\lim\limits_{z\to \pi_{s}}  m_{s+1}(z) w(\pi_s) =
		\residue\limits_{z = \pi_{s}} m_{s+1}(z)= T_{s} m_{s} (\pi_{s}).
	\end{equation}
On the other hand, using \eqref{eq:next_m}
\begin{equation}
\lim\limits_{z\to \pi_{s}} m_{s+1}(z) w(\pi_s)=\lim\limits_{z\to \pi_{s}} \left(\left( I + \frac{T_{s}}{z - \pi_{s}}\right) m_{s}(z)
		\begin{bmatrix}
			0 & \omega(\pi_{s}) \\ 0 & 0
		\end{bmatrix} \right )=m_{s}(\pi_{s}) w(\pi_{s}) + T_{s} m'_{s}(\pi_{s}) w(\pi_{s}).
\end{equation}
Therefore, we get
	\begin{equation}\label{eq:Ts}
		T_{s} m_{s}(\pi_{s}) - m_{s}(\pi_{s}) w(\pi_{s}) =T_{s} m'_{s}(\pi_{s}) w(\pi_{s}).
	\end{equation}
Since $T_{s}$ is nilpotent, we can not invert it to find
        $m'_{s}(\pi_{s}) w(\pi_{s})$. However, we see that
	\begin{equation*}
		T_{s}m'_{s}(\pi_{s}) \begin{bmatrix} \omega(\pi_{s}) \\ 0 \end{bmatrix} =
		T_{s} \begin{bmatrix} m^{12}_{s}(\pi_{s}) \\ m^{22}_{s}(\pi_{s})\end{bmatrix} -
		\omega(\pi_{s}) \begin{bmatrix} m^{11}_{s}(\pi_{s}) \\ m^{21}_{s}(\pi_{s}) \end{bmatrix}
		= T_{s}\left( \begin{bmatrix} m^{12}_{s}(\pi_{s}) \\ m^{22}_{s}(\pi_{s})\end{bmatrix} - \omega(\pi_{s})
		\lambda \begin{bmatrix} 1 \\0\end{bmatrix}\right).
	\end{equation*}
		Therefore,
	\begin{equation*}
		m'_{s}(\pi_{s}) \begin{bmatrix} \omega(\pi_{s}) \\ 0 \end{bmatrix} =
		\begin{bmatrix} m^{12}_{s}(\pi_{s}) \\ m^{22}_{s}(\pi_{s})\end{bmatrix} - \omega(\pi_{s})
		\lambda \begin{bmatrix} 1 \\0\end{bmatrix} + k \begin{bmatrix}
				t_{s}^{11} \\ t_{s}^{21}
		\end{bmatrix},
	\end{equation*}
	where the last vector is some vector in the kernel of $T_{s}$. Substituting this in \eqref{eq:resolvent} and using the fact that
	$\det m_{s}(z) \equiv 1$ gives, again using~\eqref{eq:nilp-par},
	\begin{align*}
		R_{s}(\pi_{s},\pi_{s}) &= - \operatorname{Tr}\left(
		\begin{bmatrix}
			m^{22}_{s}(\pi_{s}) & - m^{12}_{s}(\pi_{s}) \\ -m^{21}_{s}(\pi_{s}) & m^{11}_{s} (\pi_{s})
 		\end{bmatrix} \cdot
		\begin{bmatrix}
			0 & m^{12}_{s}(\pi_{s}) - \omega(\pi_{s}) \lambda + k t_{s}^{11} \\
			0 & m^{22}_{s}(\pi_{s}) + k t_{s}^{21}
		\end{bmatrix}
		\right)\\
		& = - \det m_{s}(\pi_{s}) - \omega(\pi_s) \lambda m^{21}(\pi_{s}) + k\left(m^{11}_{s}(\pi_{s})t_{s}^{21} - m^{21}_{s}(\pi_{s}) t_{s}^{11}\right)
		\\
		&= - 1  + \omega(\pi_{s}) \frac{m^{11}_{s}(\pi_{s})}{t_{s}^{11}}\cdot \frac{m^{11}_{s}(\pi_{s})t_{s}^{11}}{t_{s}^{12}}
		= -1 + \omega(\pi_{s})\frac{(m^{11}_{s}(\pi_{s}))^{2}}{t_{s}^{12}}.
	\end{align*}
\end{proof}


\subsection{Connection matrix for the $q$-Racah ensemble} 
\label{sub:connection_matrix_for_the_q_racah_ensemble}

In this section following the steps of \cite{BB},\cite{BA2} and
\cite{K} we introduce a connection matrix for the $q$-Racah ensemble
which captures all essential information about the gap probability function.

\begin{assumption}\label{assume:gdq}
We assume that $\gamma \delta q\in (0, 1)$. Then $\sigma(q^{-x})$ is an increasing function and
$\pi_{x} = \sigma(q^{-x})$ is an ordered set  for $x=0, \dots, M$. Therefore,
the framework of the previous two sections, including formula
\eqref{eq:first_gap} for computing gap probabilities, is applicable.
%
\end{assumption}

\begin{remark}For the $q$-Racah weight, the DHRP condition \eqref{eq:jumpCond} has to be slightly changed, it now takes the form
	\begin{equation}\label{eq:jumpCondqR}
	\residue\limits_{z=\pi_{x}} m(z)=\lim\limits_{z\rightarrow \pi_{x}} \left(m(z)w(x)\right),\quad \pi_{x} = \sigma(q^{-x})\in \mathfrak X.
	\end{equation}

\end{remark}

Let $u>0$ be defined via $u^2 = \gamma \delta q^2$. Then $\sigma(z) = z + \gamma\delta q z^{-1} = z + u^{2}/(qz)$, and it is easy to
see that
\begin{equation}\label{eq:prop-sigma}
	\sigma\left(\frac{u^{2}}{z}\right) = \sigma\left(\frac{z}{q}\right),\qquad \sigma(z) - \sigma(y) = (z - y)\left(1 - \frac{u^{2}}{y z q}\right).
\end{equation}

\begin{remark} Note that
\begin{equation}\label{eq:qR-prop}
	\begin{aligned}
	\frac{\omega^{\textup{qR}}(x + 1)}{\omega^{\textup{qR}}(x)} &= \frac{(q^{-x} - \alpha q)(q^{-x} - \beta \delta q)(q^{-x} - \gamma q)
	(q^{-x} - \gamma \delta q)(q^{-2x} - \gamma \delta q^{3})}{(q^{-x} -  q)(q^{-x} -\alpha^{-1} \gamma \delta q)(q^{-x} -\beta^{-1} \gamma q)
	(q^{-x} - \delta q)(\alpha \beta q)(q^{-2x} - \gamma \delta q)} \\
	&=	\frac{\Phi^{+}(q^{-x})(q^{-2x} - q u^{2})}{q \Phi^{-}(q^{-x}) (q^{-2x} - u^{2}/q)},	
	\end{aligned}
\end{equation}	
where
\begin{equation}\label{eq:phis}
	\begin{aligned}
		\Phi^+(z) &=  (z-\alpha q)(z-  \beta\delta q)(z- \gamma q)(z- \gamma \delta q),\\
		\Phi^{-}(z) &=\alpha \beta(z - \alpha^{-1} \gamma \delta q)(z - \beta^{-1} \gamma q)(z - \delta q) (z-q).
	\end{aligned}
\end{equation}

The functions $\Phi^{\pm}(z)$ appear as coefficients in the Nekrasov's equation for $q$-Racah ensemble, see \cite{DK} for details.
\end{remark}

We are interested in computing the gap probability function for large $N$. The degrees of the diagonal entries of $m_s(z)$ grow with $N$ and that
presents a serious computational difficulty. To bypass it, we introduce matrix functions $A_s(z)$, $N\leq s\leq M$, as follows:
\begin{equation}\label{eq:As}
A_s(z):=m_{s} \left( \sigma(q^{-1}z)\right)\,  D(z)\,\, m_s^{-1} \left ( \sigma(z) \right),
\qquad\text{where}\quad D(z) = \begin{bmatrix}\frac{\Phi^+(z)}{ \Phi^-(z)}  & 0\\0 &  1\end{bmatrix}.
\end{equation}
In this definition we used the fact that $\det m_s(z)= 1$, see Theorem~\ref{thm:RHPsol}, and so $m_s(z)$ is invertible.
Matrices $A_{s}(z)$ play a central role in the arguments below.
In what follows we show that the evolution $A_{s}(z)\mapsto A_{s+1}(z)$ can be
effectively computed using discrete Painlev\'e equations and also explain how to extract from this dynamics the relevant information about the
recursion on gap probabilities $D_{s}$.

\begin{remark}
In \cite{DK} the trace of matrix $A_M(z)$ is linked to the explicit computation of the frozen boundary in the tiling model.
\end{remark}

%
%

\begin{proposition}\label{prop:As}
Let
\begin{equation*}
z_{1}(s) = z_{2}(s) = q^{-s + 1},\quad z_{3} = q, \quad z_{4} = \alpha q, \quad z_{5} = \delta q, \quad z_{6} = \beta \delta q.	
\end{equation*}
Then $A_{s}(z)$ has the following properties:
\begin{enumerate}[(i)]
	\item $A_{s}^{-1}(z) = A_{s}(u^{2}/z)$ and $A(u)$ is an
          identity matrix;
	\item $\displaystyle\det A_{s}(z) = \frac{Q_{s}(z)}{P_{s}(z)}$, where
	\begin{align*}
		P_{s}(z) &= \bigg(z-z_{1}(s)\bigg)\bigg(z-\frac{u^{2}}{z_{2}(s)}\bigg) \bigg(z-z_{3}\bigg)\bigg(z-\frac{u^{2}}{z_{4}}\bigg)
		 \bigg(z-z_{5}\bigg)\bigg(z-\frac{u^{2}}{z_{6}}\bigg),\\
		Q_{s}(z) &= \frac{z_{1}(s)z_{3}z_{5}}{z_{2}(s)z_{4}z_{6}}\bigg(z-\frac{u^{2}}{z_{1}(s)}\bigg)
		 \bigg(z-z_{2}(s)\bigg)\bigg(z-\frac{u^{2}}{z_{3}}\bigg)
		 \bigg(z-z_{4}\bigg)\bigg(z-\frac{u^{2}}{z_{5}}\bigg)\bigg(z-z_{6}\bigg).
	\end{align*}
	\item  Matrices $A_{s}(z)$ has the  form
	\begin{align*}
		A_{s}(z) &= \frac{1}{P_{s}(z)} B_{s}(z),\qquad\text{where}\\
		B_{s}(z) &= \begin{bmatrix} \sum\limits_{i=0}^{3} n_{i} z^{6-i} + n_{4} u^{2} z^{2} + n_{5} u^{4} z + n_{6} u^{6}
			& z (z^{2} - u^{2})\left(m_{0} z^{2}+ m_{1} z+ m_{0} u^{2} \right) \\
			z(z^{2} - u^{2})\left(k_{0} z^{2}+ k_{1} z + k_{0}u^{2} \right) &
			\sum\limits_{i=0}^{3} n_{6-i} z^{6-i} + n_{2} u^{2} z^{2} + n_{1} u^{4} z + n_{0} u^{6}
		\end{bmatrix}.
	\end{align*}

\end{enumerate}
\end{proposition}

\begin{proof}
	Using the fact that $\det m_s(z)= 1$, we see that
	\begin{equation*}
		\det A_{s}(z) = \det D(z) = \frac{\Phi^{+}(z)}{\Phi^{-}(z)} =
		\frac{\big(z-\frac{u^{2}}{z_{3}}\big)\big(z-z_{4}\big)\big(z-\frac{u^{2}}{z_{5}}\big)\big(z-z_{6}\big)}{
		\frac{z_{4}z_{6}}{z_{3}z_{5}}\big(z-z_{3}\big)\big(z-\frac{u^{2}}{z_{4}}\big)\big(z-z_{5}\big)\big(z-\frac{u^{2}}{z_{6}}\big)}=
		\frac{Q_{s}(z)}{P_{s}(z)},
	\end{equation*}
	where in the last equation we note some cancelations since $z_{1}(s) = z_{2}(s)$. Further, since
	\begin{equation*}
		\left.\frac{z - z_{i}}{z - \frac{u^{2}}{z_{i}}}\right|_{z\rightsquigarrow \frac{u^{2}}{z}} =
		\frac{z_{i}^{2}}{u^{2}} \frac{z - \frac{u^{2}}{z_{i}}}{z-z_{i}}\quad\text{and so}\quad
		\frac{\Phi^{+}\left(u^{2}/z\right)}{\Phi^{-}\left(u^{2}/z\right)} = \frac{\Phi^{-}(z)}{\Phi^{+}(z)},
		\qquad\text{and}\quad
		\sigma\left(\frac{u^{2}}{z}\right) = \sigma\left(\frac{z}{q}\right),
	\end{equation*}
	we immediately see that $D^{-1}(z) = D(u^{2}/z)$ and $A_{s}^{-1}(z) = A_{s}(u^{2}/z)$.
	
	To complete the proof, we need to understand the singularity structure of $A_{s}(z)$. We see that $D(z)$ has simple poles at
	$z = q = z_{3}$, $z = \delta q = z_{5}$, $z = \beta^{-1} \gamma q = u^{2}/z_{6}$, and $z = \alpha^{-1} \gamma \delta q = u^{2}/z_{4}$.
	Thus, to show that $A_{s}(z) = \frac{1}{P_{s}(z)} B_{s}(z)$ where $B_{s}(z)$ is regular, we need to show that the only remaining possible
	poles of $A_{s}(z)$ are $z_{1}(s) = q^{-s + 1} $ and $u^{2}/z_{2}(s) =  u^{2} q^{s-1}$.


	Recall that from the DRHP, $m_{s}(z)$ has simple poles at $\pi_{x} = \sigma(q^{-x})$ for $0\leq x\leq s-1$,
	\begin{equation*}
		m_{s}(z) = A_{-1}(z - \pi_{x})^{-1} + A_{0} + A_{1}(z - \pi_{x}) + \cdots.
	\end{equation*}
	Moreover,
	\begin{equation*}
		\residue\limits_{z=\pi_{x}} m_{s}(z) = A_{-1} = \lim\limits_{z\rightarrow \pi_{x}} \left(m(z)w(x)\right)  = A_{0} w(x)\qquad
		\text{and}\quad A_{-1} w(x) = A_{0} w(x)^{2}  = 0,
	\end{equation*}
	since $w(x)$ is nilpotent. Thus,
	\begin{equation*}
		m_{s}(z) = A_{0} w(x)(z - \pi_{x})^{-1} + A_{0} + A_{1}(z - \pi_{x}) + \cdots = F(z)\left(I + \frac{w(x)}{z - \pi_{x}}\right),
	\end{equation*}
	where $F(z)$ is regular at $\pi_{x}$ and $F_{0} = A_{0}$. Since $\det m_{s}(z) = 1$,
	$m_{s}^{-1}(z)$ also has simple poles at $\pi_{x} = \sigma(q^{-x})$ for $0\leq x\leq s-1$ and
	\begin{equation*}
		m_{s}^{-1}(z) = \left(I - \frac{w(x)}{z - \pi_{x}}\right)F^{-1}(z),\qquad\text{where}
	\end{equation*}	$F^{-1}(z)$ is regular at $\pi_{x}$.

	From \eqref{eq:prop-sigma} we see that $\sigma(z) = \sigma(y)$ if $z = y$ or $z = u^{2}/(yq)$. Thus, the first factor
	$m_{s} \left( \sigma(q^{-1}z)\right)$ has simple poles when $z = q^{-(x-1)}$, then $\sigma(q^{-1}z) = \pi_{x}$ and $\sigma(z) = \pi_{x-1}$,
	or when $z = u^{2}q^{x}$ and then $\sigma(q^{-1}z) = \pi_{x}$ and $\sigma(z) = \pi_{x+1}$. Similarly, the second factor
	$m_{s}^{-1} \left( \sigma(z)\right)$ has simple poles when $z = q^{-x}$, then $\sigma(q^{-1}z) = \pi_{x+1}$ and $\sigma(z) = \pi_{x}$,
	or when $z = u^{2}q^{x-1}$ and then $\sigma(q^{-1}z) = \pi_{x-1}$ and $\sigma(z) = \pi_{x}$.
	
	We need to distinguish between the 	situation when both factors are singular, which happens when either $z = q^{-x}$ or $z = u^{2} q^{x}$,
	$0\leq x \leq s-2$, and the boundary case when only one factor is singular.
	
	For $z$ near $q^{-x}$, $0\leq x \leq s-2$, the matrix $A_{s}(z)$ takes the form
	\begin{align*}
		A_{s}(z) &= F(\sigma(q^{-1}z)) \left(I + \frac{w(x+1)}{\sigma(q^{-1}z) - \sigma(q^{-(x+1)})}\right) D(z)
		\left(I - \frac{w(x)}{\sigma(z) - \sigma(q^{-x})}\right)F^{-1}(\sigma(z))\\
		& = F(\sigma(q^{-1}z))
		\begin{bmatrix}
			\frac{\Phi^{+}(z)}{\Phi^{-}(z)} & \frac{q^{-x}z h(z)}{z - q^{-x}} \\ 0 & 1
		\end{bmatrix}
		F^{-1}(\sigma(z)),\qquad\text{where}\\
		h(z) &= \frac{\omega^{\textup{qR}}(x + 1)}{q^{-1}(z q^{-x}-u^{2}q)} -
		\frac{\Phi^{+}(z)}{\Phi^{-}(z)}\frac{\omega^{\textup{qR}}(x)}{(z q^{-x} - u^{2}/q)}.
	\end{align*}
	From \eqref{eq:qR-prop} we see that $h(q^{-x}) = 0$, and so $A_{s}(z)$ is regular at $q^{-x}$.

	Similarly, for $z$ near $u^{2}q^{x}$, $0\leq x \leq s-2$, the matrix $A_{s}(z)$ takes the form
	\begin{align*}
		A_{s}(z) &= F(\sigma(q^{-1}z)) \left(I + \frac{w(x)}{\sigma(q^{-1}z) - \sigma(q^{-x})}\right) D(z)
		\left(I - \frac{w(x+1)}{\sigma(z) - \sigma(q^{-(x+1)})}\right)F^{-1}(\sigma(z))\\
		& = F(\sigma(q^{-1}z))
		\begin{bmatrix}
			\frac{\Phi^{+}(z)}{\Phi^{-}(z)} & \frac{q^{-x}z h(z)}{z q^{-x} - u^{2}} \\ 0 & 1
		\end{bmatrix}
		F^{-1}(\sigma(z)),\qquad\text{where}\\
		h(z) &= \frac{\omega^{\textup{qR}}(x)}{q^{-1}z - q^{-x}} -
		\frac{\Phi^{+}(z)}{\Phi^{-}(z)}\frac{\omega^{\textup{qR}}(x+1)}{(z -q^{-(x+1)})},\quad
		h\left(\frac{u^{2}}{z}\right) = z \left(\frac{\omega^{\textup{qR}}(x)}{q^{-1}u^{2} - q^{-x}z} -
				\frac{\Phi^{-}(z)}{\Phi^{+}(z)}\frac{q \omega^{\textup{qR}}(x+1)}{(q u^{2} -q^{-x}z)}\right).
	\end{align*}
	Using \eqref{eq:qR-prop} again we see that $h(u^{2} q^{x}) = h(u^{2}/q^{-x}) = 0$, and so $A_{s}(z)$ is regular at $u^{2} q^{x}$.
	
	Consider now the boundary cases. There are four possibilities: when $z = q$ (resp.~$z = u^{2}q^{s-1}$), the first factor
	$m_{s}(\sigma(q^{-1}z))$ has a simple pole at $\pi_{0}$ (resp.~$\pi_{s-1}$) and the last factor $m_{s}^{-1}(\sigma(z))$ is regular;
	when $z = q^{-(s-1)}$ (resp.~$z = u^{2} q^{-1}$), the last factor has poles at $\pi_{s-1}$ (resp.~$\pi_{0}$)
	and the first factor is regular.
	
	Near $z = q$, we have
	\begin{align*}
		A_{s}(z) &= F(\sigma(q^{-1}z)) \begin{bmatrix}
			1 & \frac{\omega^{\textup{qR}}(0)}{\sigma(q^{-1}z) - \pi_{0}}\\0 & 1
		\end{bmatrix} \begin{bmatrix}
			\frac{\Phi^{+}(z)}{\Phi^{-}(z)} & 0 \\ 0 & 1
		\end{bmatrix} m_{s}^{-1}(\sigma(z)) = F(\sigma(q^{-1}z)) \begin{bmatrix}
			\frac{\Phi^{+}(z)}{\Phi^{-}(z)} & \frac{\omega^{\textup{qR}}(0)}{\sigma(q^{-1}z) - \pi_{0}}\\0 & 1
		\end{bmatrix}m_{s}^{-1}(\sigma(z)),
	\end{align*}
	and so both matrix elements in the top row of the central matrix have a simple pole at $q = z_{3}$ and it is already accounted for.
	Near $z = u^{2}q^{-1}$ we have
	\begin{align*}
		A_{s}(z) &= m_{s}(\sigma(z)) \begin{bmatrix}
			\frac{\Phi^{+}(z)}{\Phi^{-}(z)} & 0 \\ 0 & 1
		\end{bmatrix} \begin{bmatrix}
			1 & \frac{-\omega^{\textup{qR}}(0)}{\sigma(z) - \pi_{0}}\\0 & 1
		\end{bmatrix}  F^{-1}(\sigma(z)) =	m_{s}(\sigma(z))  \begin{bmatrix}
			\frac{\Phi^{+}(z)}{\Phi^{-}(z)} & \frac{- \Phi^{+}(z)\omega^{\textup{qR}}(0)}{\Phi^{-}(z)(\sigma(z) - \pi_{0})}\\0 & 1
		\end{bmatrix}F^{-1}(\sigma(z)),
	\end{align*}
	and the zero of $\Phi^{+}(z)$ at $u^{2}/z^{3} = u^{2} q^{-1}$ cancels the corresponding simple zero of $(\sigma(z) - \pi_{0})$
	and so $A_{s}(z)$ is regular at that point.
	Thus, the only \emph{new} remaining possible poles are at $u^{2}/z_{2} =  u^{2}q^{s-1}$ and $z_{1} = q^{-(s-1)}$, as we claimed.
	
	Thus, matrix entries of $B_{s}(z)$ are polynomials and the condition $A_{s}^{-1}(z) = A_{s}(u^{2}/z)$ becomes
	%
	\begin{equation*}
		\begin{bmatrix}
			b_{s}^{22}(z) & - b_{s}^{12}(z) \\ - b_{s}^{21}(z) & b^{11}(z)
		\end{bmatrix} = \frac{z^{6}}{u^{6}} \begin{bmatrix}
			b_{s}^{11}(u^{2}/z) & b_{s}^{12}(u^{2}/z) \\ b_{s}^{21}(u^{2}/z) & b_{s}^{22}(u^{2}/z)
		\end{bmatrix}.
	\end{equation*}
	From here it is immediate that $\deg b_{s}^{ij}(z)\leq 6$. Moreover, if we slightly adjust the coefficients and write
	\begin{align*}
		b_{s}^{11}(z) & = n_{0}z^{6} + n_{1} z^{5} + n_{2} z^{4} + n_{3} z^{3} + n_{4} u^{2} z^{2} + n_{5} u^{4} z + n_{6} u^{6}\\
		b_{2}^{22}(z) &= \frac{z^{6}}{u^{6}} b_{s}^{11}(z) =
		n_{6}z^{6} + n_{5} z^{5} + n_{4} z^{4} + n_{3} z^{3} + n_{2} u^{2} z^{2} + n_{1} u^{4} z + n_{0} u^{6},
	\end{align*}
	as claimed. In the same way we can see that
	\begin{equation*}
		b_{s}^{12} = (z^{2} - u^{2})(m_{4} (z^{4} + u^{2} z^{2} + u^{4}) + m_{0}z(z^{2} + u^{2}) + m_{1}z^{2}),
	\end{equation*}
	but since $A_{s}(z)$ is asymptotic to a diagonal matrix when $z\to\infty$, $\deg b_{s}^{12}(z) \leq 5$, and hence $m_{4}=0$ and we get
	$b_{s}^{12}(z) = z (z^{2} - u^{2})\left(m_{0} z^{2}+ m_{1} z+ m_{0} u^{2} \right)$. The argument for $b_{s}^{21}(z)$ is similar, and this
	completes the proof.
	
	
	%
\end{proof}

\subsection{Initial conditions} 
\label{sub:initial_conditions}
Our strategy for computing gap probabilities $D_{s}$ is to use the recursion. In this section we compute the initial conditions for this recursion
in the $q$-Racah case. First, we need the following result.

\begin{lemma}[{\cite[Proposition 6.1]{BB}}]\label{lem:init_cond}
The solution $m_N(z)$ of the DRHP($\{\pi_0, \dots, \pi_{N-1}\}, \omega|_{\{\pi_0, \dots, \pi_{N-1}\}}$) with the asymptotics
$m_N(z)\sim \begin{bmatrix} z^{N} & 0 \\ 0 & z^{-N}\end{bmatrix}$ as $z\to\infty$ is given by
\begin{equation}\label{eq:initial_sol}
m_N(z)= \begin{bmatrix}
	\Pi(z) & 0 \\
	\Pi(z)\,\sum\limits_{x=0}^{N-1} \frac{\rho_{x}}{z - \pi_{x}} & \frac{1}{\Pi(z)}
\end{bmatrix},\qquad\text{where}\quad \Pi(z) = \prod\limits_{m=0}^{N-1}(z - \pi_{m}),
%
%
\end{equation}
and where
\begin{equation*}\label{eq:rho}
\rho_{x}=\omega(x)^{-1}\cdot \prod\limits_{\begin{smallmatrix}
	0\leq m\leq N-1\\
	m\neq x
\end{smallmatrix}
} (\pi_x-\pi_m)^{-2},\quad 0\leq x\leq N-1.
\end{equation*}
\end{lemma}

Now we are ready to compute the initial conditions for the $q$-Racah case by evaluating the matrix $A_{N}(z)$. It still has the overall structure
described in Proposition~\ref{prop:As}, with $z_{1}(N)=z_{2}(N)=q^{-N+1}$, but using Lemma~\ref{lem:init_cond}, we can now give an explicit formula
for $A_{N}(z)$.

\begin{lemma} The matrix $A_N(z)$ has the following form:
	\begin{align*}
		A_{N}(z) &= \frac{1}{P_{N}(z)} B_{N}(z) = \frac{1}{P_{N}(z)} \begin{bmatrix} b_{N}^{11}(z) & b_{N}^{12}(z) \\[3pt]
			 b_{N}^{21}(z) & b_{N}^{22}(z)
		\end{bmatrix},\qquad\text{where}\\
		b_{N}^{11}(z) &= q^{-N}\frac{z_{3}z_{5}}{z_{4}z_{6}}
		\bigg(z-\frac{u^{2}}{z_{1}(N)}\bigg) \bigg(z-\frac{u^{2}}{z_{2}(N)}\bigg)\bigg(z-z_{3}\bigg)
			\bigg(z-z_{4}\bigg)\bigg(z-\frac{u^{2}}{z_{5}}\bigg)\bigg(z-z_{6}\bigg),\\
		b_{N}^{12}(z)&=0,\\
		b_{N}^{21}(z) &= z (z^{2} - u^{2})(k_{0}z^{2} + k_{1} z + k_{0}u^{2}),\quad \text{where }
		k_{0} = \sum_{x=0}^{N-1}q\rho_{x}\left(\frac{q^{-N}}{\alpha \beta} - q^{N-1}\right),\quad \alpha \beta = \frac{z_{4}z_{6}}{z_{3}z_{5}},\\
		k_{1} &= \sum_{x=0}^{N-1}q\rho_{x}\left(\frac{q^{-N}}{\alpha \beta}\left(q \sigma(q^{-x}) - \frac{u^{2}}{z_{1}(N)} - \frac{u^{2}}{z_{2}(N)}
		- z_{3} - z_{4} - \frac{u^{2}}{z_{5}} - z_{6}\right)\right. \\
		&\qquad \left. - q^{N-1}\left(
		\sigma(q^{-x}) - z_{1}(N) - z_{2}(N) - \frac{u^{2}}{z_{3}} - \frac{u^{2}}{z_{4}} - z_{5} - \frac{u^{2}}{z_{6}}\right)\right),\\
		b_{N}^{22}(z) &= q^{N} \bigg(z-z_{1}(N)\bigg) \bigg(z-z_{2}(N)\bigg) \bigg(z-\frac{u^{2}}{z_{3}}\bigg) \bigg(z-\frac{u^{2}}{z_{4}}\bigg)
			\bigg(z-z_{5}\bigg)\bigg(z-\frac{u^{2}}{z_{6}}\bigg).
	\end{align*}
%
%
\end{lemma}

\begin{proof}
	From Lemma~\ref{lem:init_cond} we get
	\begin{align*}
		A_{N}(z) &= m_{N}(\sigma(q^{-1}z))\, D(z)\, m_{N}^{-1}(\sigma(z)) \\
		&=
		\begin{bmatrix}
			\frac{\Phi^{+}(z)\Pi(\sigma(q^{-1}z))}{\Phi^{-}(z)\Pi(\sigma(z))} & 0 \\
			\frac{\Phi^{+}(z)\Pi(\sigma(q^{-1}z))}{\Phi^{-}(z)\Pi(\sigma(z))}
			\sum\limits_{x=0}^{N-1} \frac{\rho_{x}}{\sigma(q^{-1}z) - \pi_{x}}  -
			\frac{\Pi(\sigma(z))}{\Pi(\sigma(q^{-1}z))} \sum\limits_{x=0}^{N-1} \frac{\rho_{x}}{\sigma(z) - \pi_{x}}\quad
			& \frac{\Pi(\sigma(z))}{\Pi(\sigma(q^{-1}z))}
		\end{bmatrix}.
	\end{align*}
	Direct computation shows that
	\begin{equation*}
		\frac{\Pi(\sigma(q^{-1}z))}{\Pi(\sigma(z))} = \frac{q^{-N}(z-q)(z-u^{2}q^{N-1})}{(z-u^{2}q^{-1})(z - q^{-N+1})} =
		\frac{q^{-N}\big(z-\frac{u^{2}}{z_{1}(N)}\big) \big(z-z_{3}\big)}{\big(z-z_{1}(N)\big)\big(z-\frac{u^{2}}{z_{3}}\big) },
	\end{equation*}
	and the expressions for $b_{N}^{11}(z)$ and $b_{N}^{22}(z)$ immediately follow. From Proposition~\ref{prop:As}(iii) we know that
	$b_{N}^{21}(z) = z(z^{2} - u^{2}) (k_{0} z^{2} + k_{1} z + k_{0}u^{2})$. Moreover, for $b_{N}^{21}(z)$ we have
	\begin{equation*}
		b_{N}^{21}(z) = z (z^{2} - u^{2})(k_{0}z^{2} + k_{1} z + k_{0}u^{2}) =
		z q \sum\limits_{x=0}^{N-1}\rho_{x}\left(\frac{b_{N}^{11}(z)}{(z - q^{-(x-1)})(z - q^{x}u^{2})} -
		\frac{b_{N}^{22}(z)}{(z - q^{-x})(qz - q^{x}u^{2})}\right),
	\end{equation*}
	and therefore
	\begin{align*}
		k_{0} &= \frac{q}{u^{6}}\sum_{x=0}^{N-1}\rho_{x}\left(b_{N}^{22}(0) - \frac{b_{N}^{11}(0)}{q}\right)
		= \sum_{x=0}^{N-1}q\rho_{x}\left(\frac{q^{N} z_{1}(N)z_{2}(N)z_{5}}{z_{3}z_{4}z_{6}} - \frac{q^{-N}z_{3}^{2}}{z_{1}(N)z_{2}(N)}\right)\\
		&= \sum_{x=0}^{N-1}q\rho_{x}\left(\frac{q^{-N}}{\alpha \beta} - q^{N-1}\right),
	\end{align*}
	where $\alpha \beta = \frac{z_{4}z_{6}}{z_{3}z_{5}}$. Similarly,
	\begin{align*}
		k_{1} &= \frac{-1}{u^{2}}\lim_{z\to 0}\sum_{x=0}^{N-1} \frac{q \rho_{x}}{z}\left(
		\frac{b_{N}^{11}(z)}{(z - q^{-(x-1)})(z - q^{x}u^{2})} -
				\frac{b_{N}^{22}(z)}{(z - q^{-x})(qz - q^{x}u^{2})} - \left(b_{N}^{22}(0) - \frac{b_{N}^{11}(0)}{q}\right)\frac{z^{2}-u^{2}}{u^{4}}
		\right)\\
		&= \sum_{x=0}^{N-1}q\rho_{x}\left(\frac{q^{-N}}{\alpha \beta}\left(q \sigma(q^{-x}) - \frac{u^{2}}{z_{1}(N)} - \frac{u^{2}}{z_{2}(N)}
		- z_{3} - z_{4} - \frac{u^{2}}{z_{5}} - z_{6}\right)\right. \\
		&\qquad \left. - q^{N-1}\left(
		\sigma(q^{-x}) - z_{1}(N) - z_{2}(N) - \frac{u^{2}}{z_{3}} - \frac{u^{2}}{z_{4}} - z_{5} - \frac{u^{2}}{z_{6}}\right)\right).
	\end{align*}
\end{proof}


\subsection{The Lax Pair} 
\label{sub:the_lax_pair}

Equations \eqref{eq:next_m} and \eqref{eq:As} constitute the Lax Pair for solutions of DRHP, c.f. \cite[Section 3]{BB},
\begin{equation}\label{eq:LaxPair}
	\left\{
	\begin{aligned}
		m_{s+1}(\sigma(z)) &= \left(I + \frac{T_{s}}{\sigma(z) - \pi_{s}}\right) m_{s}(\sigma(z)),\\
		m_{s}(\sigma(q^{-1}z)) &= A_{s}(z) m_{s}(\sigma(z)) D^{-1}(z)
	\end{aligned}\right.
\end{equation}
This Lax Pair in turn gives rise to the \emph{isomonodromic dynamics} for the matrices $A_{s}(z)$,
\begin{equation}\label{eq:isom}
	A_{s+1}(z) = \left(I + \frac{T_{s}}{\sigma(q^{-1}z) - \pi_{s}}\right)A_{s}(z)\left(I - \frac{T_{s}}{\sigma(z) - \pi_{s}}\right).
\end{equation}

To run the recursion computing the gap probability function we will
need the values of  $D_{k}(k),$ $D_{k}(k+1)$ computed in the next proposition.

\begin{proposition}\label{D} $($\cite{BB}, Proposition 6.6$)$  Let
  $\mathfrak X \subset \mathbb R$
  be a discrete set, let $\{P_n(z)\}$ be the family of orthogonal
  polynomials corresponding to a strictly positive weight function
  $\omega : \mathfrak X=\{\pi_0,\dots,\pi_N\} \rightarrow \mathbb R$. Then
\begin{equation}
D_k= \frac{1}{Z} \cdot \prod\limits_{0\leq i\le j\leq k−1}^{ } (\pi_i - \pi_j )^2\cdot \prod\limits_{l=0}^{k-1} \omega(\pi_l),
\end{equation}

\begin{equation}
D_{k+1} = \omega(\pi_k) \cdot h_{k}^{-1}\cdot D_k(k)\cdot
\prod\limits_{l=0}^{k-1}(\pi_k - \pi_l)^2,\end{equation}
where $h_k$ is given by
$$h_k=\rho_k+\sum\limits_{m=0}^{k-1}\frac{\rho_m}{(\pi_k-\pi_m)^2}$$
and $\rho_k$ is defined in \eqref{eq:rho}.
\end{proposition}


\section{Moduli Space of Elliptic Connections and Discrete Painlev\'e Equations} 
\label{sec:ell}

It is possible to consider $A_s(z)$ as a matrix representation, with respect to some trivialization, of an $\mathcal{E}_u$-connection on
the vector bundle $\mathcal{L} = \mathcal{O}\oplus\mathcal{O}(-1)$.

\begin{remark}
Here we follow the approach of \cite{BA1} and twist from the trivial vector bundle to $\mathcal{L}$, since
$\mathcal{O}\oplus\mathcal{O}(-1)$ has more
gauge automorphisms, and this results in significant simplifications in computations.
\end{remark}

Thus, we consider the following class of $\mathcal E_u$-connections.
\begin{equation}\label{eq:family}
	A(z)= \frac{1}{P(z)}\begin{bmatrix}
		b_{11}(z) & \frac{b_{12}(z)}{z}\\ b_{21}(z) & b_{22}(z)
	\end{bmatrix} ,\quad  b_{21}(0)= 0,
\end{equation}
where $\deg(b_{11}(z))\leq 6$, $\deg(b_{12}(z))\leq 8$, $\deg(b_{21}(z))\leq 5$, $\deg(b_{22}(z))\leq 6$ and
\begin{equation*}
	\det A(z)= \frac{Q(z)}{P(z)},\qquad
\begin{aligned}
	P(z) &= (z-z_{1})(z - u^{2}/z_{2})(z-z_{3})(z - u^{2}/z_{4})(z-z_{5})(z - u^{2}/z_{6}),\\
	Q(z) &= \frac{z_{1}z_{3}z_{5}}{z_{2}z_{4}z_{6}}(z - u^{2}/z_{1})(z-z_{2})(z - u^{2}/z_{3})(z-z_{4})(z - u^{2}/z_{5})(z - z_{6}).
\end{aligned}
\end{equation*}
		We also require that $A(z)$ satisfies the \emph{asymptotic condition}
		\begin{equation*}
			S\left(\frac{z}{q} + \frac{u^{2}}{z}\right) A(z) S^{-1}\left(z + \frac{u^{2}}{q z}\right) \sim
			\begin{bmatrix}	d_{1} & 0 \\ 0 & d_{2} \end{bmatrix},
				\qquad\text{where}\quad
			S(z) = \begin{bmatrix}1 & 0 \\ 0 & z	\end{bmatrix},
		\end{equation*}
      and the \emph{involution condition}
\begin{equation*}
	A(u^2/z)=A^{-1}(z) \text{ and } A(u) \text{ is an
          identity matrix}.
\end{equation*}
\begin{remark}
We need to fix that either $A(u)$ is an identity or minus identity to
work with a connected component of the moduli space.
\end{remark}

Such matrix representation of a connection is not unique, since
the choice of the trivialization of $\mathcal{L}$ can be composed with an automorphism of the bundle.
Such automorphism can be written as a matrix
\begin{equation}\label{gauge}
R=\left [\begin{array}{cc} r_{11} & r_{12}\\ 0 & r_{22}
\end{array}\right],\\
\quad r_{11}, r_{22}\in \mathbb C-\{0\}, \quad r_{12}\in \Gamma(\mathbb{P}^1, \mathcal{O}(1)).
\end{equation}

	As usual, this matrix ansatz is described using the so-called
	spectral coordinates. To introduce them,
	we first observe that, using gauge transformations, we can reduce  $b_{21}(z)$ to
	\begin{equation*}
		b_{21}(z) = z(z-u^{2})(z - t)\Big(z^{2} - \frac{u^{2}}{t}\Big).
	\end{equation*}
	and so we put $t =  t_{1}/t_{2}$ to be our first spectral coordinate. The second spectral coordinate
	$p$ is the $(1,1)$-entry of $A(z)$ at $t$, $p = b_{11}(t)/P(t)$.
	Imposing the remaining conditions, such as the asymptotic condition and the determinant condition, allows us the
	express the remaining entries of $A(z)$ as rational functions of the spectral coordinates. Those rational functions
	become indeterminate at certain points, and resolving these indeterminacies via blowups and compactifying
	identifies our moduli space of $q$-connections with (a blowup of) one of the \emph{Spaces of Initial Conditions} in
	Sakai's classification scheme for discrete Painev\'e equations, \cite{Sak:2001:RSAWARSGPE}.
	
	However, the new feature of this example is that the involution condition above induces the involution on parameters,
	$t\leftrightarrow u^{2}/t$ and $p\leftrightarrow 1/p$. As a result, in the spectral coordinates $(t,p)$ we get  more
	than the usual 8 points. Specifically, we get the following six pairs of involution-conjugated points:
	\begin{alignat*}{3}
	&\left(\frac{u^{2}}{z_{1}},0\right), \left(z_{1},\infty\right),&\quad 	&\left(\frac{u^{2}}{z_{3}},0\right), 	\left(z_{3},\infty\right),&\quad 	&\left(\frac{u^{2}}{z_{5}},0\right), \left(z_{5},\infty\right),\\
	&\left(z_{2},0\right), \left(\frac{u^{2}}{z_{2}},\infty\right),&\quad 	&\left(z_{4},0\right), 	\left(\frac{u^{2}}{z_{4}},\infty\right),&\quad
	&\left(z_{6},0\right), \left(\frac{u^{2}}{z_{6}},\infty\right),\notag
	\end{alignat*}
	as well as points $(u,1)$ and $(-u,-1)$, and  points $\left(\infty,-\rho_{1} = d\right)$ and
	$\left(\infty,-\rho_{2} =
          \frac{z_{1}z_{3}z_{5}}{z_{2}z_{4}z_{6} q d}\right)$. Note
        that from the viewpoint of computations of the moduli space
        we can interchange $d_1$ and $d_2$ and by $d$ we denote one of
        the choices. In the same way $\rho_1$ and $\rho_2$ are also
        interchangeable.

	\begin{figure}[h]
		\begin{tikzpicture}[baseline=1ex,
				xscale=1.5, yscale=1.1, basept/.style={circle, draw=red!100,fill=red!100, thick, inner sep=0pt, minimum size=1.8mm},
				conjpt/.style={circle, draw=red!100,fill=white!100, thick, inner sep=0pt, minimum size=1.8mm}]
			    \draw[-] (-1,-2) -- (4.5,-2) node[right] {};
			    \draw[-] (-1,2) -- (4.5,2) node[right] {};
			    \draw[-] (-0.5,-2.5) -- (-0.5,2.5) node[right] {};
			    \draw[-] (4,-2.5) -- (4,2.5) node[right] {};
			    \draw[domain=-2.1:2.1, smooth, variable=\y, blue, very thick] plot ({\y*\y},{\y}) node [right] {$C_{0}$};
				\draw[-, blue, very thick] (4,-2.5) -- (4,2.5) node [above] {$C_{1}$};
				\node[basept] (p1) at (0.25,0.5) [label = above:$\pi_{1}$]  {};
				\node[conjpt] (p1c) at (0.25,-0.5) [label = below:$\pi_{1}'$]  {};
				\node[conjpt] (p2c) at (0.64,0.8) [label = above:$\pi_{2}'$]  {};
				\node[basept] (p2) at (0.64,-0.8) [label = below:$\pi_{2}$]  {};
				\node[basept] (p3) at (1.21,1.1) [label = above:$\pi_{3}$]  {};
				\node[conjpt] (p3c) at (1.21,-1.1) [label = below:$\pi_{3}'$]  {};
				\node[conjpt] (p4c) at (1.69,1.3) [label = below:$\pi_{4}'$]  {};
				\node[basept] (p4) at (1.69,-1.3) [label = above:$\pi_{4}$]  {};
				\node[basept] (p5) at (2.25,1.5) [label = below:$\pi_{5}$]  {};
				\node[conjpt] (p5c) at (2.25,-1.5) [label = above:$\pi_{5}'$]  {};
				\node[conjpt] (p6c) at (2.89,1.7) [label = below:$\pi_{6}'$]  {};
				\node[basept] (p6) at (2.89,-1.7) [label = above:$\pi_{6}$]  {};
				\node[basept] (p7) at (4,-1) [label = right:$\pi_{7}$]  {};
				\node[basept] (p8) at (4,1) [label = right:$\pi_{8}$]  {};
				\node[basept] (p9) at (0,0) [label = left:$\pi_{9}$]  {};
			\end{tikzpicture}
		\caption{The base points of the moduli space of $q$-connections for the $q$-Racah ensemble}	
		\label{fig:a1sur-xy}
	\end{figure}
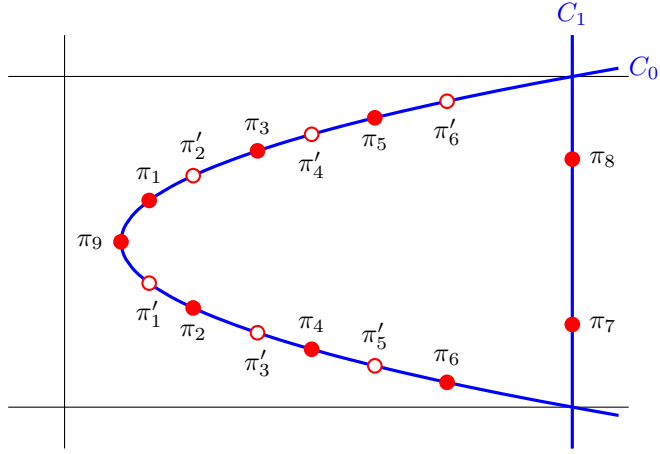
	
	To fix this, we need to introduce the \emph{involution-invariant} coordinates $x=t + \frac{u^{2}}{t}$ and
	$y = \frac{pt - u}{pu - t}$ gluing these pairs of points together.
	In the involution-invariant $(x,y)$-coordinates we get the point configuration shown on Figure~\ref{fig:a1sur-xy}.
	The points
	\begin{equation*}
		\pi_{7}\left(\infty,\rho_{1} = -d\right),\qquad
		\pi_{8}\left(\infty,\rho_{2} =  -\frac{z_{1}z_{3}z_{5}}{z_{2}z_{4}z_{6} q d}\right)
	\end{equation*}
	lie on the $(1,0)$-curve $C_{1} = V(X = 1/x)$, and the points
	\begin{equation*}
		\pi_{i}\left(z_{i} + \frac{u^{2}}{z_{i}}, \frac{z_{i}}{u} \right),\quad i=1,3,5;\qquad
		\pi_{i}\left(z_{i} + \frac{u^{2}}{z_{i}}, \frac{u}{z_{i}} \right),	\quad i = 2,4,6
	\end{equation*}
	lie on the $(1,2)$-curve
	$C_{0} = V(u(y^2+1) - xy)$; note also that when $x = z_{i} + \frac{u^{2}}{z_{i}}$,
	the equation $u(y^2+1) - xy$ factors as $u(y^2+1) - xy = u(y - y(\pi_{i}))(y - y(\pi_{i}'))$, where
	the conjugated points $\pi_{i}'$ are given by
	\begin{equation*}
		\pi_{i}'\left(z_{i} + \frac{u^{2}}{z_{i}}, \frac{u}{z_{i}} \right),\quad i=1,3,5;\qquad
		\pi_{i}'\left(z_{i} + \frac{u^{2}}{z_{i}}, \frac{z_{i}}{u} \right),	\quad i = 2,4,6.
	\end{equation*}
	The remaining two points $(t,p) = \pm(u,1)$ are fixed points of the involution and are also base points of the
	coordinate $p$. We get one final base point
	$\pi_{9}(-2u,-1)$, similar to the $q$-Hahn case. The reason
        why we are left only with this point in the computation is
        because due to
        involution-invariant change of coordinates the singularity at
        $t=u$ gets resolved (note that for $t=u$ we necessarily have $p=1$).

\subsection{Reference Example of $q$-$P\left(A_{1}^{(1)}\right)$} 
\label{sub:reference_example_of_q_p_left_a__1_1_e__7_1_right}

The goal of this section is to show that the isomonodromic dynamics corresponding to the parameter evolution
\begin{equation}\label{eq:parevolve}
	\bar{z}_{2} = q z_{2},\quad \bar{z}_{4} = q z_{4},\quad \bar{d} = q^{-1} d,\quad \bar{z}_{i} = z_{i}\text{ otherwise. }
\end{equation}
is in fact equivalent to the standard $q$-$P(A_{1}^{(1)})$ discrete Painlev\'e dynamic and to give the explicit change of
variables from the  involution-invariant spectral coordinates $(x,y)$ to the Painlev\'e coordinates $(f,g)$.
The approach here is similar to that of \cite{DzhTak:2018:OSAOSGTODPE}, so we shall be brief and refer the reader to that paper for details.
Below we review the geometric setting of Sakai's theory, as well as introduce some notation. We only consider a generic setup here, see
\cite{Sak:2001:RSAWARSGPE} and especially \cite{KajNouYam:2017:GAOPE} for careful and detailed exposition that also
includes special cases.

\subsubsection{The Root Data} 
\label{ssub:the_root_data}

A discrete Painlev\'e equation describes dynamics on a
certain family $\mathbb{X}$ of rational algebraic surfaces obtained by blowing up $\mathbb{P}^{1}\times \mathbb{P}^{1}$
at eight, possibly infinitely close, points $p_{i}$ that lie on a, possibly reducible, bi-quadratic curve $\Gamma$.
Let $([f_{0}:f_{1}],[g_{0}:g_{1}])$ be homogeneous coordinates on $\mathbb{P}^{1} \times \mathbb{P}^{1}$. Then
$\mathbb{P}^{1} \times \mathbb{P}^{1}$ is covered by four affine charts, $(f = f_{0}/f_{1},g = g_{0}/g_{1})$,
$(F = 1/f,g)$, $(f,G = 1/g)$, and $(F,G)$. Parameters $\mathbf{b}= \{b_{i}\}$ of the family are essentially the
coordinates of the blowup points. However, since we need to account for various gauge actions, a better choice
of parameters is given by the so-called \emph{root variables} $\mathbf{a}= \{a_{i}\}$, as we explain later.
Then a typical surface in the family is $\mathcal{X}_{\mathbf{b}} = \operatorname{Bl}_{p_{1},\ldots,p_{8}}(\mathbb{P}^{1}\times
\mathbb{P}^{1})\overset{\eta}\to \mathbb{P}^{1}\times \mathbb{P}^{1}$. The \emph{Picard lattices} for all of these
surfaces are isomorphic,
\begin{equation*}
	\operatorname{Pic}(\mathcal{X}_{\mathbf{b}}) \simeq \operatorname{Pic}(\mathcal{X}) = H^{1}(\mathcal{X},\mathcal{O}_{\mathcal{X}}^{*}) =
	\operatorname{Div}(\mathcal{X})/\operatorname{P}(\mathcal{X}) =
	\operatorname{Span}_{\mathbb{Z}}\{\mathcal{H}_{f},\mathcal{H}_{g},\mathcal{F}_{1},\ldots,\mathcal{F}_{8}\},
\end{equation*}
where $\mathcal{H}_{f}$ (resp.~$\mathcal{H}_{g}$) are the classes of the total transforms of the vertical (resp.~horizontal)
lines on $\mathbb{P}^{1} \times \mathbb{P}^{1}$ and $\mathcal{F}_{i}$ are the classes of the total transforms of the
exceptional divisors of the blowup at $p_{i}$ under the full blowup map $\eta$.
A generic surface $\mathcal{X}_{\mathbf{b}}$ in the family is then a generalized
Halphen surface, i.e., it has a \emph{unique} anti-canonical divisor $-K_{\mathcal{X}} = \eta^{*}(\Gamma)$
of \emph{canonical type}. That is, if
\begin{equation*}
	-K_{\mathcal{X}} = 2H_{f} + 2H_{g} - F_{1} - \cdots - F_{8} =  \sum_{i} m_{i} d_{i},
\end{equation*}
is the decomposition of the anti-canonical
divisor into irreducible components $d_{i}$ with the multiplicities $m_{i}$, then $d_{i}$ is orthogonal to $-K_{\mathcal{X}}$
w.r.t. the intersection form, $d_{i}\bullet (-K_{\mathcal{X}}) = 0$.

We associate with this geometric data two sub-latticed in the Picard lattice $\operatorname{Pic}(\mathcal{X})$: the \emph{surface sub-lattice}
$\Pi(R) = \operatorname{Span}_{\mathbb{Z}}\{\delta_{i} = [d_{i}]\}\triangleleft \operatorname{Pic}(\mathcal{X})$ that encodes the geometry of
the point configuration, and its orthogonal complement $Q$ in $\operatorname{Pic}(\mathcal{X})$, which is called the \emph{symmetry sub-lattice}.
Both of these sub-lattices are \emph{root lattices}, i.e., they have bases of simple roots, $R = \{\delta_{i}\mid \delta_{i}^{2} = -2\}$,
and $R^{\perp} = \{\alpha_{j}\mid \alpha_{j}^{2} = -2, \alpha_{j}\bullet \delta_{i} = 0\}$. Then
$Q = \Pi(R^{\perp}) = \operatorname{Span}_{\mathbb{Z}}\{\alpha_{i}\}$. Further, $R$ and $R^{\perp}$ can be described by
affine Dynkin diagrams $\mathcal{D}_{1}$ and $\mathcal{D}_{2}$, whose types are then called the \emph{surface} (resp.~\emph{symmetry}) type of the
corresponding discrete Painlev\'e equation (and the surface family). To the symmetry root diagram $\mathcal{D}_{2}$ we can associate an
affine Weyl group $W(\mathcal{D}_{2})$ whose action on $\operatorname{Pic}(\mathcal{X})$ is generated by reflections in the basis symmetry roots,
$w_{i}: \mathcal{C}\mapsto w_{i}(\mathcal{C}) = \mathcal{C} + (\alpha_{i}\bullet \mathcal{C}) \alpha_{i}$. Extending this group by the
group of automorphism of the Dynkin diagram (same for both $\mathcal{D}_{i}$) we get the full extended affine Weyl group
$\widetilde{W}\left(\mathcal{D}_{2}\right) = \operatorname{Aut}\left(\mathcal{D}_{2}\right) \ltimes W\left(\mathcal{D}_{2}\right)$, again acting
on $\operatorname{Pic}(\mathcal{X})$ and preserving both sub-lattices (and thus preserving the surface family, that's why it is called the
\emph{symmetry group}). In the cases we are interested in, this group coincides with the group of \emph{Cremona isometries} of $\mathcal{X}$,
its action on $\operatorname{Pic}(\mathcal{X})$ can be extended to maps on the family, and a discrete Painlev\'e equation is a
discrete dynamical system on $\mathbb{X}$ that corresponds to a translation element of $\operatorname{Cr}(\mathcal{X})$.

Let us now consider a particular example of the $q$-$P(A_{1}^{(1)})$-equation, as written in \cite{KajNouYam:2017:GAOPE}.
It is characterized by the following Dynkin diagrams
(where the numbers at the nodes are the coefficients
of the linear combination describing the class of the anti-canonical divisor $\delta = - \mathcal{K}_{X}$ in terms of the root classes):
\begin{center}
	\begin{tabular}{cc}
		$
		{\begin{tikzpicture}[
				elt/.style={circle,draw=black!100,thick, inner sep=0pt,minimum size=2mm}]
			\path 	( 0,0) 	node 	(D0) [elt] {}
			        ( 1,0) 	node 	(D1) [elt] {};
			\draw [black,line width=1pt ] (D0.north east) -- (D1.north west);
			\draw [black,line width=1pt ] (D0.south east) -- (D1.south west);
				\node at ($(D0.south east) + (-0.5,-0.1)$) 	{$\delta_{0}$};
				\node at ($(D1.south west) + (+0.5,-0.1)$) 	{$\delta_{1}$};
				\node at ($(D0.north) + (0,0.2)$) 	{\color{red}\small$1$};
				\node at ($(D1.north) + (0,0.2)$) 	{\color{red}\small$1$};
				\end{tikzpicture}}$ \qquad &\qquad  $\displaystyle
				\begin{tikzpicture}[
							elt/.style={circle,draw=black!100,thick, inner sep=0pt,minimum size=2mm}]
						\path 	(-3,0) 	node 	(a1) [elt] {}
						        (-2,0) 	node 	(a2) [elt] {}
						        ( -1,0) node  	(a3) [elt] {}
						        ( 0,0) 	node  	(a4) [elt] {}
						        ( 1,0) 	node 	(a5) [elt] {}
						        ( 2,0)	node 	(a6) [elt] {}
						        ( 3,0)	node 	(a7) [elt] {}
						        ( 0,1)	node 	(a0) [elt] {};
						\draw [black,line width=1pt ] (a1) -- (a2) -- (a3) -- (a4) -- (a5) --  (a6) -- (a7) (a4) -- (a0);
							\node at ($(a1.south) + (0,-0.2)$) 	{$\alpha_{1}$};
							\node at ($(a2.south) + (0,-0.2)$)  {$\alpha_{2}$};
							\node at ($(a3.south) + (0,-0.2)$)  {$\alpha_{3}$};
							\node at ($(a4.south) + (0,-0.2)$)  {$\alpha_{4}$};	
							\node at ($(a5.south) + (0,-0.2)$)  {$\alpha_{5}$};		
							\node at ($(a6.south) + (0,-0.2)$) 	{$\alpha_{6}$};	
							\node at ($(a7.south) + (0,-0.2)$) 	{$\alpha_{7}$};	
							\node at ($(a0.north) + (0,0.2)$) 	{$\alpha_{0}$};		
							\node at ($(a0.east) + (0.2,0)$) 	{\color{red}\small$2$};
							\node at ($(a1.north) + (0,0.2)$) 	{\color{red}\small$1$};
							\node at ($(a2.north) + (0,0.2)$) 	{\color{red}\small$2$};
							\node at ($(a3.north) + (0,0.2)$) 	{\color{red}\small$3$};
							\node at ($(a4.north east) + (0.2,0.23)$) {\color{red}\small$4$};
							\node at ($(a5.north) + (0,0.2)$) 	{\color{red}\small$3$};
							\node at ($(a6.north) + (0,0.2)$) 	{\color{red}\small$2$};
							\node at ($(a7.north) + (0,0.2)$) 	{\color{red}\small$1$};
							\end{tikzpicture}$\\
				Dynkin diagram $A_{1}^{(1)}$ & Dynkin diagram $E_{7}^{(1)}$\\[5pt]
				$\begin{aligned}
				\delta_{0} &= \mathcal{H}_{f} + \mathcal{H}_{g} -
				\mathcal{F}_{1}  - \mathcal{F}_{2}  - \mathcal{F}_{3}  - \mathcal{F}_{4}\\
				\delta_{1} &= \mathcal{H}_{f} + \mathcal{H}_{g} -
				\mathcal{F}_{5}  - \mathcal{F}_{6}  - \mathcal{F}_{7}  - \mathcal{F}_{8} \\
				\qquad \\ \qquad \\
				\end{aligned}$ &
				$\begin{aligned}
				\alpha_{0} &= \mathcal{H}_{f} - \mathcal{H}_{g} & \quad \alpha_{4} &= \mathcal{H}_{g} - \mathcal{F}_{1} - 				\mathcal{F}_{5}\\
				\alpha_{1} &= \mathcal{F}_{3} - \mathcal{F}_{4} & \alpha_{5} &= \mathcal{F}_{5} - \mathcal{F}_{6} \\
				\alpha_{2} &= \mathcal{F}_{2} - \mathcal{F}_{3} & \alpha_{6} &= \mathcal{F}_{6} - \mathcal{F}_{7} \\
				\alpha_{3} &= \mathcal{F}_{1} - \mathcal{F}_{2} & \alpha_{7} &= \mathcal{F}_{7} - \mathcal{F}_{8}
				\end{aligned}$ \\[3pt]
	The surface data & The symmetry data
	\end{tabular}
\end{center}


\subsubsection{The Point Configuration} 
\label{ssub:the_point_configuration}
To describe a model $A_{1}^{(1)}$-surface, we
start with the following point configuration on $\mathbb{P}^{1}\times \mathbb{P}^{1}$.
Take two divisors $d_{i}\in \delta_{i}$ and consider their pushdown
$\eta_{*}(d_{i})$ on $\mathbb{P}^{1}\times \mathbb{P}^{1}$. In the affine chart $(f,g)$ these divisors are given by
\begin{equation*}
		\eta_{*}(d_{i}) = V(A_{i} fg + B_{i} f + C_{i} g + D_{i}),\qquad i=0,1.
\end{equation*}
Since $\delta_{0}\bullet \delta_{1} = 2$, generically
$|\eta_{*}(d_{0})\cap \eta_{*}(d_{1})| = 2$ and we can assume that these two points do not lie on the same
horizontal or vertical line. 
Then, using the $\mathbf{PGL}_{2} \times \mathbf{PGL}_{2}$ action,
we can arrange that the intersection points are $(0,\infty)$ (thus, $C_{i} = 0$) and $(\infty,0)$ (thus, $B_{i} = 0$).
Using rescaling, we can arrange
\begin{equation*}
	\eta_{*}(d_{0}) = V(fg - 1) \text{ and } \eta_{*}(d_{1}) = V(fg - \kappa),
	\text{ where $\kappa = \kappa_{1}/\kappa_{2}$ is some parameter.}
\end{equation*}
Assigning four blowup points to each of the curves,
\begin{equation*}
	p_{i}\left(\nu_{i},\frac{1}{\nu_{i}}\right),\quad i = 1,\dots,4;\qquad
	p_{i}\left(\frac{\kappa_{1}}{\nu_{i}},\frac{\nu_{i}}{\kappa_{2}}\right),\quad i = 5,\dots,8.
\end{equation*}
we get the $A_{1}^{(1)}$ point configuration as in \cite{KajNouYam:2017:GAOPE}, see Figure~\ref{fig:a1sur-fg}.
\begin{figure}[h]
	\begin{tikzpicture}[baseline=20ex,
		xscale=4, yscale=2.6, basept/.style={circle, draw=red!100,fill=red!100, thick, inner sep=0pt, minimum size=1.8mm}]
	    \draw[-] (0.3,0.5) -- (2.3,0.5) node[right] {};
	    \draw[-] (0.3,2) -- (2.3,2) node[right] {};
	    \draw[-] (2,0.3) -- (2,2.3) node[right] {};
	    \draw[-] (0.5,0.3) -- (0.5,2.3) node[right] {};
	    \draw[domain=10/23:2.3, smooth, variable=\x, blue, very thick] plot ({\x},{1/\x}) node [right] {$\eta_{*}(d_{0})$};
		\draw[domain=45/22:0.3, smooth, variable=\x, blue, very thick] plot ({\x},{(10*\x-21)/(4*\x-10)}) node [left] 		{$\eta_{*}(d_{1})$};
		\node[basept] (p1) at (7/10,10/7) [label = below:$p_{1}$]  {};
		\node[basept] (p2) at (9/10,10/9) [label = below:$p_{2}$]  {};
		\node[basept] (p3) at (11.3/10,10/11.3) [label = below:$p_{3}$]  {};
		\node[basept] (p4) at (14/10,10/14) [label = below:$p_{4}$]  {};
		\node[basept] (p5) at (11/10,25/14) [label = above:$p_{5}$]  {};
		\node[basept] (p6) at (13/10,5/3) [label = above:$p_{6}$]  {};
		\node[basept] (p7) at (3/2,3/2) [label = above:$p_{7}$]  {};
		\node[basept] (p8) at (17/10,5/4) [label = above:$p_{8}$]  {};
	\end{tikzpicture}
	\caption{The standard point configuration for the $A_{1}^{(1)}$-surface}
	\label{fig:a1sur-fg}
\end{figure}

This point configuration has $10$ parameters, $\nu_{i}$, $i=1,\dots 8$ and $\kappa_{1}$, $\kappa_{2}$, however, there are two rescaling actions,
one is internal on the parameters $(\kappa_{1},\kappa_{2})\sim(\mu \kappa_{1}, \mu \kappa_{2})$, and the other is the scaling of the axes preserving
the curve $fg=1$,
\begin{equation*}
	\left(\begin{matrix}
		\nu_{1} & \nu_{2} & \nu_{3} & \nu_{4} \\ \nu_{5} & \nu_{6} & \nu_{7} & \nu_{8}
	\end{matrix};\quad
	\begin{matrix}	\kappa_{1} \\ \kappa_{2} \end{matrix} ;\quad  \begin{matrix} f \\ g	\end{matrix}
	\right)\sim
	\left(\begin{matrix}
	\lambda\nu_{1} & \lambda\nu_{2} & \lambda\nu_{3} & \lambda\nu_{4} \\
	\frac{1}{\lambda}\nu_{5} & \frac{1}{\lambda}\nu_{6} & \frac{1}{\lambda}\nu_{7} & \frac{1}{\lambda}\nu_{8}
	\end{matrix};\quad
	\begin{matrix}	\mu\kappa_{1} \\ \mu\kappa_{2} \end{matrix} ;\quad  \begin{matrix} \lambda f \\ \frac{1}{\lambda} g	\end{matrix}
	\right),
\end{equation*}
so the actual number of parameters is $8$. As usual, the invariant parameterization is given by the root variables $a_{i}$ that can be
obtained using the \emph{period map}.

\subsubsection{The Period Map} 
\label{ssub:the_period_map}

\begin{proposition} For our model of the $A_{1}^{(1)*}$-surface, the period map and the root variables
	$a_{i} = \chi(\alpha_{i})$ are given by
\begin{equation}\label{eq:root-params-a1-fg}
	a_{0} = \frac{\kappa_{1}}{\kappa_{2}}, \quad a_{1} = \frac{\nu_{3}}{\nu_{4}}, \quad a_{2} = \frac{\nu_{2}}{\nu_{3}}, \quad a_{3} = \frac{\nu_{1}}{\nu_{2}},
	\quad a_{4} = \frac{\kappa_{2}}{\nu_{1}\nu_{5}}, \quad a_{5} = \frac{\nu_{5}}{\nu_{6}}, \quad a_{6} = \frac{\nu_{6}}{\nu_{7}},
	\quad a_{7} = \frac{\nu_{7}}{\nu_{8}}.	
\end{equation}	
	This gives us the following parameterization by the root variables $a_{i}$
	\begin{equation*}
	\left(\begin{matrix}
		\nu_{1} & \nu_{2} & \nu_{3} & \nu_{4} \\ \nu_{5} & \nu_{6} & \nu_{7} & \nu_{8}
	\end{matrix};\quad
	\begin{matrix}	\kappa_{1} \\ \kappa_{2} \end{matrix} ;\quad  \begin{matrix} f \\ g	\end{matrix}
	\right)
	= \left(\begin{matrix}
			\nu_{1} & \frac{\nu_{1}}{a_{3}} & \frac{\nu_{1}}{a_{2}a_{3}} & \frac{\nu_{1}}{a_{1} a_{2} a_{3}} \\
			\nu_{5} & \frac{\nu_{5}}{a_{5}} & \frac{\nu_{5}}{a_{5}a_{6}} & \frac{\nu_{5}}{a_{5} a_{6} a_{7}} \\			
		\end{matrix};\quad
	\begin{matrix}	a_{4} \nu_{1} \nu_{5} \\ a_{0} a_{4} \nu_{1} \nu_{5} \end{matrix} ;\quad  \begin{matrix} f \\ g	\end{matrix}
	\right).
	\end{equation*}
\end{proposition}	

\begin{proof}
	To compute the period map, we first need to define a symplectic form $\omega$ whose pole divisor is $\eta_{*}(d_{0}) + \eta_{*}(d_{1})$. Let us put
	$s = fg$. Then, up to a normalization constant $C$, we can take $\omega$ (in the affine $(f,g)$-chart) to be
	\begin{align*}
		\omega &= C \frac{df\wedge dg}{(fg - 1)(fg - \kappa)} = C \frac{df\wedge ds}{f(s-1)(s-\kappa)} = C \frac{ds\wedge dg}{g(s-1)(s-\kappa)},\\
		\operatorname{res}_{\eta_{*}(d_{0})} \omega & = \frac{C}{\kappa-1} \frac{df}{f} = - \frac{C}{\kappa-1} \frac{dg}{g},\qquad
		\operatorname{res}_{\eta_{*}(d_{1})} \omega  = -\frac{C}{\kappa-1} \frac{df}{f} = \frac{C}{\kappa-1} \frac{dg}{g}.
	\end{align*}
	Then
	\begin{align*}
		\chi(\alpha_{0})&=\chi(\mathcal{H}_{f} - \mathcal{H}_{g}) =\chi([H_{f} - F_{1}] - [H_{g} - F_{1}])  =
		\int\limits_{(H_{g} - F_{1})\cap \eta_{*}(d_{1})}^{(H_{f} - F_{1})\cap \eta_{*}(d_{1})} \operatorname{res}_{\pi_{*}(d_{1})} \omega
		= -\frac{C}{\kappa-1} \int_{\kappa \nu_{1}}^{\nu_{1}} \frac{df}{f} \\
		& = \frac{C}{\kappa-1}\log\left(\frac{\kappa_{1}}{\kappa_{2}}\right), \\
		\chi(\alpha_{1})&=\chi(\mathcal{F}_{3} - \mathcal{F}_{4}) = \chi([F_{3}] - [F_{4}])  =
		\int\limits_{p_{3}}^{p_{4}} \operatorname{res}_{\pi_{*}(d_{0})} \omega  = \frac{C}{\kappa-1}\log\left(\frac{\nu_{3}}{\nu_{4}}\right), \\
		\chi(\alpha_{2})&=\chi(\mathcal{F}_{2} - \mathcal{F}_{3})  = \frac{C}{\kappa-1}\log\left(\frac{\nu_{2}}{\nu_{3}}\right), \qquad
		\chi(\alpha_{3})=\chi(\mathcal{F}_{1} - \mathcal{F}_{2})  = \frac{C}{\kappa-1}\log\left(\frac{\nu_{1}}{\nu_{2}}\right), \\
		\chi(\alpha_{4})&=\chi(\mathcal{H}_{g} - \mathcal{F}_{1} - \mathcal{F}_{5}) =\chi([H_{g} - F_{1}] - [F_{5}])  =
		\int\limits_{(H_{g} - F_{1})\cap \pi_{*}(d_{1})}^{F_{5}\cap \eta_{*}(d_{1})} \operatorname{res}_{\eta_{*}(d_{1})} \omega
		= \frac{C}{\kappa-1} \int_{\frac{1}{\nu_{1}}}^{\frac{\nu_{5}}{\kappa_{2}}} \frac{dg}{g} \\
		& = \frac{C}{\kappa-1}\log\left(\frac{\kappa_{2}}{\nu_{1}\nu_{2}}\right), \\
		\chi(\alpha_{5})&=\chi(\mathcal{F}_{5} - \mathcal{F}_{6}) = \chi([F_{5}] - [F_{6}])  =
		\int\limits_{p_{5}}^{p_{6}} \operatorname{res}_{\pi_{*}(d_{1})} \omega  = \frac{C}{\kappa-1}\log\left(\frac{\nu_{5}}{\nu_{6}}\right), \\
		\chi(\alpha_{6})&=\chi(\mathcal{F}_{6} - \mathcal{F}_{7})  = \frac{C}{\kappa-1}\log\left(\frac{\nu_{6}}{\nu_{7}}\right), \qquad
		\chi(\alpha_{7})=\chi(\mathcal{F}_{7} - \mathcal{F}_{8})  = \frac{C}{\kappa-1}\log\left(\frac{\nu_{7}}{\nu_{8}}\right).
	\end{align*}
	From these computations we see that it is convenient to choose the normalization constant $C=\kappa-1$, which gives \eqref{eq:root-params-a1-fg}.
	Also, note that, using the decomposition of the anti-canonical divisor class,
	$\delta = \delta_{0} + \delta_{1} = 2 \alpha_{0} + \alpha_{1} + 2 \alpha_{2} + 3 \alpha_{3} + 4 \alpha_{4} + 3 \alpha_{5} + 2 \alpha_{6} + \alpha_{7}$,
	we get the following expression for the step $q$ of the dynamics:
	\begin{equation*}
		q = \exp(\chi(\delta)) = a_{0}^{2} a_{1} a_{2}^{2} a_{3}^{3} a_{4}^{4} a_{5}^{3} a_{6}^{2} a_{7} =
		\frac{\kappa_{1}^{2} \kappa_{2}^{2}}{\nu_{1}\nu_{2}\nu_{3}\nu_{4}\nu_{5}\nu_{6}\nu_{7}\nu_{8}}.
	\end{equation*}	
\end{proof}


\subsubsection{The Symmetry Group} 
\label{ssub:the_symmetry_group}
The symmetry group of the $A_{1}^{(1)*}$-surface family is the extended affine Weyl group
$\widetilde{W}\left(E_{7}^{(1)}\right) = \operatorname{Aut}\left(E_{7}^{(1)}\right) \ltimes W\left(E_{7}^{(1)}\right)$,
where $\operatorname{Aut}\left(E_{7}^{(1)}\right)\simeq \mathbb{Z}_{2}$ and the affine Weyl group $W\left(E_{7}^{(1)}\right)$
is defined in terms of generators $w_{i} = w_{\alpha_{i}}$ and relations that
are encoded by the affine Dynkin diagram $E_{7}^{(1)}$,
\begin{equation*}
	W\left(E_{7}^{(1)}\right) = W\left(\raisebox{-12pt}{\begin{tikzpicture}[
			elt/.style={circle,draw=black!100,thick, inner sep=0pt,minimum size=1.5mm},scale=0.9]
		\path 	(-1.5,0) 	node 	(a1) [elt] {}
				(-1,0) 		node 	(a2) [elt] {}
		        (-0.5,0) 	node 	(a3) [elt] {}
		        ( 0,0) 		node  	(a4) [elt] {}
		        ( 0.5,0) 	node  	(a5) [elt] {}
		        ( 1,0) 		node 	(a6) [elt] {}
		        ( 1.5,0) 	node 	(a7) [elt] {}
		        ( 0,0.5)	node 	(a0) [elt] {};
		\draw [black] (a1) -- (a2) -- (a3) -- (a4) -- (a5) -- (a6) -- (a7)   (a4) -- (a0);
			\node at ($(a0.east)  + (0.3,0)$) 	{$\alpha_{0}$};
			\node at ($(a1.south) + (0,-0.2)$)  {$\alpha_{1}$};
			\node at ($(a2.south) + (0,-0.2)$)  {$\alpha_{2}$};
			\node at ($(a3.south) + (0,-0.2)$)  {$\alpha_{3}$};
			\node at ($(a4.south) + (0,-0.2)$)  {$\alpha_{4}$};
			\node at ($(a5.south) + (0,-0.2)$)	{$\alpha_{5}$};
			\node at ($(a6.south) + (0,-0.2)$)	{$\alpha_{6}$};
			\node at ($(a7.south) + (0,-0.2)$)	{$\alpha_{7}$};
			\end{tikzpicture}} \right)
	=
	\left\langle w_{0},\dots, w_{6}\left|
	\begin{alignedat}{2}
    w_{i}^{2} = e,\quad  w_{i}\circ w_{j} &= w_{j}\circ w_{i}& &\text{ when
   				\raisebox{-0.08in}{\begin{tikzpicture}[
   							elt/.style={circle,draw=black!100,thick, inner sep=0pt,minimum size=1.5mm}]
   						\path   ( 0,0) 	node  	(ai) [elt] {}
   						        ( 0.5,0) 	node  	(aj) [elt] {};
   						\draw [black] (ai)  (aj);
   							\node at ($(ai.south) + (0,-0.2)$) 	{$\alpha_{i}$};
   							\node at ($(aj.south) + (0,-0.2)$)  {$\alpha_{j}$};
   							\end{tikzpicture}}}\\
    w_{i}\circ w_{j}\circ w_{i} &= w_{j}\circ w_{i}\circ w_{j}& &\text{ when
   				\raisebox{-0.17in}{\begin{tikzpicture}[
   							elt/.style={circle,draw=black!100,thick, inner sep=0pt,minimum size=1.5mm}]
   						\path   ( 0,0) 	node  	(ai) [elt] {}
   						        ( 0.5,0) 	node  	(aj) [elt] {};
   						\draw [black] (ai) -- (aj);
   							\node at ($(ai.south) + (0,-0.2)$) 	{$\alpha_{i}$};
   							\node at ($(aj.south) + (0,-0.2)$)  {$\alpha_{j}$};
   							\end{tikzpicture}}}
	\end{alignedat}\right.\right\rangle.
\end{equation*}


\subsubsection{The Standard Dynamic} 
\label{ssub:the_dynamic}
The evolution of the parameters considered in \cite{KajNouYam:2017:GAOPE} is $\overline{\kappa}_{1} = \frac{\kappa_1}{q}$,
$\overline{\kappa}_{2} = q \kappa_{2}$, and $\overline{\nu}_{i} = \nu_{i}$ for all $i$ which, in the
$q$-$P\left(A_{1}^{(1)}\right)$-case gives us equations (8.7) in Section~8.1.3 of \cite{KajNouYam:2017:GAOPE}:
\begin{equation}\label{eq:dynamic-a1-fg}
	\left\{
	\begin{aligned}
		\frac{\left(fg - \frac{\kappa_{1}}{\kappa_{2}}\right)(\overline{f} g - \frac{\kappa_{1}}{q \kappa_{2}})}{(fg - 1) (\overline{f}g - 1)} &=
		\frac{\left(g - \frac{\nu_{5}}{\kappa_{2}}\right) \left(g - \frac{\nu_{6}}{\kappa_{2}}\right) \left(g - \frac{\nu_{7}}{\kappa_{2}}\right)
		\left(g - \frac{\nu_{8}}{\kappa_{2}}\right)}{ \left(g - \frac{1}{\nu_{1}}\right) \left(g - \frac{1}{\nu_{2}}\right) \left(g - \frac{1}{\nu_{3}}\right)
		\left(g - \frac{1}{\nu_{4}}\right)},\\
		\frac{\left(fg - \frac{\kappa_{1}}{\kappa_{2}}\right)(f \underline{g} - \frac{q \kappa_{1}}{\kappa_{2}})}{(fg - 1) (f\underline{g} - 1)} &=
		\frac{\left(f - \frac{\kappa_{1}}{\nu_{5}}\right) \left(f - \frac{\kappa_{1}}{\nu_{6}}\right) \left(f - \frac{\kappa_{1}}{\nu_{7}}\right)
		\left(f - \frac{\kappa_{1}}{\nu_{8}}\right)}{ \left(f - \nu_{1}\right) \left(f - \nu_{2}\right) \left(f - \nu_{3}\right)
		\left(f - \nu_{4}\right)}.		
	\end{aligned}	
	\right.
\end{equation}

Further, we get the following action on the root variables: $\overline{a}_{0} = q^{-2} a_{0}$, $\overline{a}_{4} = q a_{4}$ and $\overline{a}_{i} = a_{i}$
otherwise.

\begin{remark}
\label{rem:period-action}
In computing the birational representation, the following observation is very helpful. Let
$w\in \widetilde{W}\left(E_{6}^{(1)}\right)$, and let $\eta: \mathcal{X}_{\mathbf{b}} \to \mathcal{X}_{\bar{\mathbf{b}}}$
be the corresponding mapping, i.e., $w = \eta_{*}$ and $w^{-1} = \eta^{*}$, where $\eta_{*}$ and $\eta^{*}$ are the induced
push-forward and pull-back actions on the divisors (and hence on $\operatorname{Pic(\mathcal{X})}$) that are inverses of each other.
Since  $\eta$ is just a change of the blowdown structure that the period
map $\chi$ does not depend on, $\chi_{\mathcal{X}}(\alpha_{i}) = \chi_{\eta(\mathcal{X})}(\eta_{*}(\alpha_{i}))$. Thus, we can
compute the evolution of the root variables directly from the action on $\operatorname{Pic}(\mathcal{X})$
via the formula
\begin{equation}
	\bar{a}_{i} = \chi_{\eta(\mathcal{X})}(\bar{\alpha}_{i}) = \chi_{\mathcal{X}} (\eta^{*}(\bar{\alpha}_{i}))
	= \chi_{\mathcal{X}} (w^{-1}(\bar{\alpha}_{i})).
\end{equation}
Thus the action of $\eta$ on the root variables is \emph{inverse} to the action of $w$ on the roots. This is not essential
for the generating reflections, that are involutions, but it is important for composed maps.
\end{remark}

In view of Remark~\ref{rem:period-action}, we can now
identify the translation element in $\widetilde{W}\left(E_{7}^{(1)}\right)$ (w.r.t. the given choice of root vectors) through its action
on the symmetry roots:
\begin{equation}
	\varphi_{*}: \upalpha =  \langle \alpha_{0}, \alpha_{1}, \alpha_{2}, \alpha_{3}, \alpha_{4}, \alpha_{5}, \alpha_{6}, \alpha_{7}  \rangle
	\mapsto \overline{\upalpha} = \upalpha + \langle 2,0,0,0,-1,0,0,0 \rangle \delta.
\end{equation}

\begin{proposition}
	The $q$-$P\left(A_{1}^{(1)}\right)$ discrete Painlev\'e dynamics, given our choses, corresponds to the following element of the
	$\widetilde{W}\left(E_{7}^{(1)}\right)$, written in terms of the generators
	\begin{equation}\label{eq:word-a1-fg}
		\varphi_{*} = w_{0} w_{4} w_{5} w_{3} w_{4} w_{6} w_{5} w_{2} w_{3} w_{4} w_{1} w_{2} w_{3} w_{0} w_{4}
				w_{7} w_{6} w_{5} w_{4} w_{3} w_{0} w_{4} w_{6} w_{5} w_{2} w_{3} w_{4} w_{7} w_{6} w_{5}
				w_{1} w_{2} w_{3} w_{4}.
	\end{equation}
\end{proposition}

\begin{proof} The proof of this statement is a standard computation, see
	\cite{DzhTak:2018:OSAOSGTODPE} for an example.
\end{proof}

\begin{remark} Note that this representation of the dynamic allows us to recover the action
of the mapping on $\operatorname{Pic}(X)$ and, provided that we know the bitrational representation
of the symmetry group $\widetilde{W}\left(E_{7}^{(1)}\right)$, the equation itself (although there are better ways
of obtaining the equation).
\end{remark}


\subsection{Matching the Isomonodromic and the Standard Dynamic} 
\label{sub:matching_the_isomonodromic_and_the_standard_dynamic}
We are now ready to prove the following Theorem.

\begin{theorem} The isomonodromic dynamics corresponding to the parameter evolution \eqref{eq:parevolve} is
	equivalent to the standard dynamics through the following change of coordinates
	from isomonodromic to Painlev\'e:
	\begin{equation}\label{eq:chvars-fg-xy-a1}
		\begin{aligned}
		f(x,y) &= \frac{\sigma_{3}(xy + u(y-1)) - u^{2}(x^{2} - \sigma_{1} x + \sigma_{2}(y+1)) + u^{3}(1-y)(\sigma_{1} - x)+ u^{4}(1+ y)}{
		\sigma_{3}x(xy + u(y-1))- u^2(\sigma_{2} xy + \sigma_{3}(y+1)) + u^{3}\sigma_{2}(1-y) + u^{4}(\sigma_{1}(1 + y) - x) + u^{5}(y - 1)},\\
		g(x,y) &= \frac{x y z_{6} + u z_{6}(y - 1) - u^{2}(1 + y)}{z_{6}(1 + y) - x - u(1 + y)},	
		\end{aligned}
	\end{equation}
	where $\sigma_{i}$ are the standard symmetric functions,  $\sigma_{1} = z_{2} + z_{4} + z_{6}$, $\sigma_{2} = z_{2} z_{4} + z_{4} z_{6} + z_{6}z_{2}$,
	and $\sigma_{3} = z_{2} z_{4} z_{6}$. The inverse change of coordinates is given by		
	\begin{align}\label{eq:chvars-xy-fg-a1}
		x(f,g) &= \frac{(\kappa_{1} - \kappa_{2})g + \nu_{6}(1 + \kappa_{1} \kappa_{2}) (1 - fg)  + \nu_{6}^{2}(\kappa_{1} - \kappa_{2})f}{
		\kappa_{1} - \kappa_{2} f g},\\
		y(f,g) &= \frac{
		\nu_{1} \nu_{6} (1-f g) (\nu_{6} \kappa_{1} - (1 + \kappa_{1} \kappa_{2})g) + \kappa_{2} fg ((\nu_{1} \nu_{6} - 1)g - \nu_{6})
		+ \nu_{1} \kappa_{2} g^{2} + \kappa_{1}(1 - \nu_{1} g)(g + \nu_{6})
		}{
		(1 - fg)(\nu_{6} - \kappa_{2}(g - \nu_{6} \kappa_{1})) - \nu_{6} ((g + \nu_{6})(\kappa_{1} \nu_{1} + \kappa_{2}f(1 - g \nu_{1}))
		- \kappa_{1}(1 + \nu_{6} f)) }.
	\end{align}
	
	Corresponding to these changes of coordinates we have the following matching of parameters:
		\begin{equation}\label{eq:chpars-fg-xy-a1}
			\nu_{1} = \frac{1}{z_{6}},\,
			\nu_{2} = \frac{1}{z_{1}},\,
			\nu_{3} = \frac{1}{z_{3}},\,
			\nu_{4} = \frac{1}{z_{5}},\,
			\nu_{5} = \frac{u z_{4}}{z_{2}},\,
			\nu_{6} = u,\,
			\nu_{7}= \frac{-\rho_{1} z_{4} z_{6}}{u},\,
			\nu_{8}=  \frac{-\rho_{2} z_{4} z_{6}}{u},\,
			\kappa_{1} = \frac{u}{z_{2}},\,
			\kappa_{2} = \frac{z_{4}}{u}.
		\end{equation}
	\end{theorem}		
	
\begin{proof}
	Looking at the point configuration for the moduli space of the $q$-connections in the $q$-Racah case, we see that it is not minimal
	(the divisor $\Delta_{0}$ has self-intersection degree $-3$), so we need to blow down one of the $-1$-curves. Instead, it is easier to
	blow up a point on one of the curves $\eta_{*}(d_{i})$ for the $A_{1}^{(1)}$-surface model. Without loss of generality, we can
	let this point $p_{9}$ be on the curve $\eta_{*}(d_{0})$, see Figure~\ref{fig:a1sur-xy-fg}.
	
	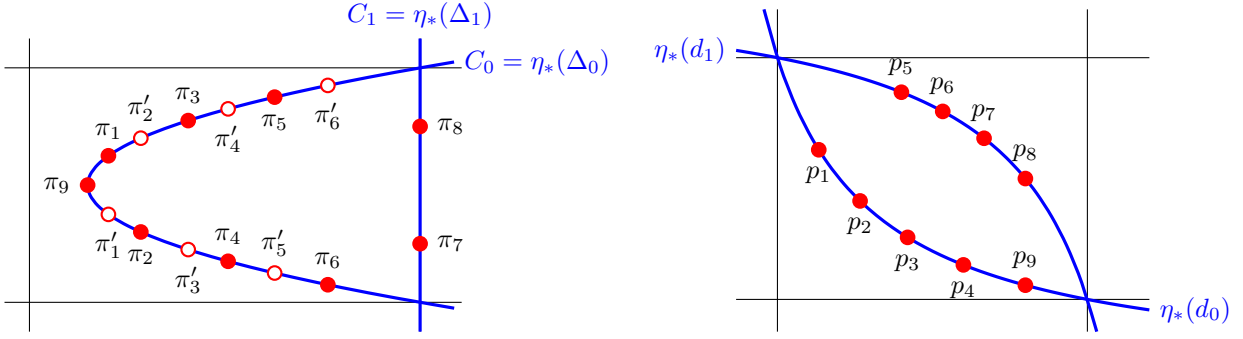
\begin{figure}[h]
		\begin{tikzpicture}[
				xscale=1.7, yscale=1.2, basept/.style={circle, draw=red!100,fill=red!100, thick, inner sep=0pt, minimum size=1.8mm},
				conjpt/.style={circle, draw=red!100,fill=white!100, thick, inner sep=0pt, minimum size=1.8mm},scale = 0.65]
			    \draw[-] (-1,-2) -- (4.5,-2) node[right] {};
			    \draw[-] (-1,2) -- (4.5,2) node[right] {};
			    \draw[-] (-0.7,-2.5) -- (-0.7,2.5) node[right] {};
			    \draw[-] (4,-2.5) -- (4,2.5) node[right] {};
			    \draw[domain=-2.1:2.1, smooth, variable=\y, blue, very thick] plot ({\y*\y},{\y}) node [right] {$C_{0}=\eta_{*}(\Delta_{0})$};
				\draw[-, blue, very thick] (4,-2.5) -- (4,2.5) node [above] {$C_{1}=\eta_{*}(\Delta_{1})$};
				\node[basept] (p1) at (0.25,0.5) [label = above:$\pi_{1}$]  {};
				\node[conjpt] (p1c) at (0.25,-0.5) [label = below:$\pi_{1}'$]  {};
				\node[conjpt] (p2c) at (0.64,0.8) [label = above:$\pi_{2}'$]  {};
				\node[basept] (p2) at (0.64,-0.8) [label = below:$\pi_{2}$]  {};
				\node[basept] (p3) at (1.21,1.1) [label = above:$\pi_{3}$]  {};
				\node[conjpt] (p3c) at (1.21,-1.1) [label = below:$\pi_{3}'$]  {};
				\node[conjpt] (p4c) at (1.69,1.3) [label = below:$\pi_{4}'$]  {};
				\node[basept] (p4) at (1.69,-1.3) [label = above:$\pi_{4}$]  {};
				\node[basept] (p5) at (2.25,1.5) [label = below:$\pi_{5}$]  {};
				\node[conjpt] (p5c) at (2.25,-1.5) [label = above:$\pi_{5}'$]  {};
				\node[conjpt] (p6c) at (2.89,1.7) [label = below:$\pi_{6}'$]  {};
				\node[basept] (p6) at (2.89,-1.7) [label = above:$\pi_{6}$]  {};
				\node[basept] (p7) at (4,-1) [label = right:$\pi_{7}$]  {};
				\node[basept] (p8) at (4,1) [label = right:$\pi_{8}$]  {};
				\node[basept] (p9) at (0,0) [label = left:$\pi_{9}$]  {};
			\end{tikzpicture}\quad
			\begin{tikzpicture}[
				xscale=4.1, yscale=3.2, basept/.style={circle, draw=red!100,fill=red!100, thick, inner sep=0pt, minimum size=1.8mm},scale = 0.67]
			    \draw[-] (0.3,0.5) -- (2.3,0.5) node[right] {};
			    \draw[-] (0.3,2) -- (2.3,2) node[right] {};
			    \draw[-] (2,0.3) -- (2,2.3) node[right] {};
			    \draw[-] (0.5,0.3) -- (0.5,2.3) node[right] {};
			    \draw[domain=10/23:2.3, smooth, variable=\x, blue, very thick] plot ({\x},{1/\x}) node [right] {$\eta_{*}(d_{0})$};
				\draw[domain=45/22:0.3, smooth, variable=\x, blue, very thick] plot ({\x},{(10*\x-21)/(4*\x-10)}) node [left] 		{$\eta_{*}(d_{1})$};
				\node[basept] (p1) at (7/10,10/7) [label = below:$p_{1}$]  {};
				\node[basept] (p2) at (9/10,10/9) [label = below:$p_{2}$]  {};
				\node[basept] (p3) at (11.3/10,10/11.3) [label = below:$p_{3}$]  {};
				\node[basept] (p4) at (14/10,10/14) [label = below:$p_{4}$]  {};
				\node[basept] (p9) at (17/10,10/17) [label = above:$p_{9}$]  {};
				\node[basept] (p5) at (11/10,25/14) [label = above:$p_{5}$]  {};
				\node[basept] (p6) at (13/10,5/3) [label = above:$p_{6}$]  {};
				\node[basept] (p7) at (3/2,3/2) [label = above:$p_{7}$]  {};
				\node[basept] (p8) at (17/10,5/4) [label = above:$p_{8}$]  {};
			\end{tikzpicture}			
		\caption{Matching the moduli spaces of $q$-connections with $A_{1}^{(1)}$-surface}	
		\label{fig:a1sur-xy-fg}
	\end{figure}

Next, we need to find a map (change of basis) from $\operatorname{Pic}(\mathcal{X}^{R})$ to $\operatorname{Pic}(\mathcal{X}^{P})$ that will transform
the components of the anti-canonical divisor class $\Delta_{i}$ to $d_{i}$ and then extend this map to the isomorphism between the surfaces, which,
when written as a birational map $(x,y)\dashrightarrow(f,g)$, will give us the required change of variables.
However, finding such an identification between
the surfaces does not guarantee that the dynamics will also match. First, it may turn out that the dynamics are non-equivalent. Second, even if they are
equivalent, our preliminary change of variables may result in a \emph{conjugated} translation vector. Below we explain that there is a systematic procedure
that resolves this issue.

First, comparing the expressions for the irreducible components of the anti-canonical divisor class in
$\operatorname{Pic}(\mathcal{X}^{R})$ and $\operatorname{Pic}(\mathcal{X}^{P})$,
\begin{align*}
	\delta_{0} &= \mathcal{H}_{f} + \mathcal{H}_{g} -
	\mathcal{F}_{1}  - \mathcal{F}_{2}  - \mathcal{F}_{3}  - \mathcal{F}_{4} - \mathcal{F}_{9} =
	\mathcal{H}_{x} + 2\mathcal{H}_{y} - \mathcal{E}_{1}  - \mathcal{E}_{2}  - \mathcal{E}_{3}  -
	\mathcal{E}_{4} - \mathcal{E}_{5} - \mathcal{E}_{6} - \mathcal{E}_{9},
	\\
	\delta_{1} &= \mathcal{H}_{f} + \mathcal{H}_{g} -
	\mathcal{F}_{5}  - \mathcal{F}_{6}  - \mathcal{F}_{7}  - \mathcal{F}_{8} = \mathcal{H}_{x}   - \mathcal{E}_{7}  - 		\mathcal{E}_{8},
\end{align*}
we see that we can preliminary do the following change of bases:
	\begin{align*}
		\mathcal{H}_{f} & = \mathcal{H}_{x} +  \mathcal{H}_{y}  - \mathcal{E}_{2} - \mathcal{E}_{9},\qquad
		& \mathcal{H}_{x} & = \mathcal{H}_{f} +  \mathcal{H}_{g}  - \mathcal{F}_{7} - \mathcal{F}_{8}, \\
		\mathcal{H}_{g} & = \mathcal{H}_{x} +  \mathcal{H}_{y}  - \mathcal{E}_{4} - \mathcal{E}_{9},\qquad
		& \mathcal{H}_{y} & = \mathcal{H}_{f} +  \mathcal{H}_{g}  - \mathcal{F}_{7} - \mathcal{F}_{9}, \\
		\mathcal{F}_{1} & = \mathcal{E}_{1},\qquad
		& \mathcal{E}_{1} & = \mathcal{F}_{2},\ \\
		\mathcal{F}_{2} & = \mathcal{E}_{6},\qquad
		& \mathcal{E}_{2} & = \mathcal{H}_{g} - \mathcal{F}_{7}, \\
		\mathcal{F}_{3} & = \mathcal{E}_{3},\qquad
		& \mathcal{E}_{3} & = \mathcal{F}_{3}, \\
		\mathcal{F}_{4} & = \mathcal{E}_{5},\qquad
		& \mathcal{E}_{4} & = \mathcal{H}_{f} - \mathcal{F}_{7}, \\
		\mathcal{F}_{5} & = \mathcal{E}_{7},\qquad
		& \mathcal{E}_{5} & = \mathcal{F}_{4}, \\
		\mathcal{F}_{6} & = \mathcal{E}_{8},\qquad
		& \mathcal{E}_{6} & = \mathcal{F}_{1}, \\
		\mathcal{F}_{7} & = \mathcal{H}_{x} +  \mathcal{H}_{y}  - \mathcal{E}_{2} - \mathcal{E}_{4} - \mathcal{E}_{9},\qquad
		& \mathcal{E}_{7} & = \mathcal{F}_{5}, \\
		\mathcal{F}_{8} & = \mathcal{H}_{y}  - \mathcal{E}_{9},\qquad
		& \mathcal{E}_{8} & = \mathcal{F}_{6}, \\
		\mathcal{F}_{9} & = \mathcal{H}_{x}  - \mathcal{E}_{9},\qquad
		& \mathcal{E}_{9} & = \mathcal{H}_{f} +  \mathcal{H}_{g}  - \mathcal{F}_{7} - \mathcal{F}_{8} - \mathcal{F}_{9}.
	\end{align*}
	From this correspondence we see that the $f$ is a coordinate on a pencil of $(1,1)$-curves in the $(x,y)$-plane passing through the
	points $\pi_{2}$ and $\pi_{9}$. Taking $u^{2}(y+1) - z_{2}(x y + u (y-1))$ and $z_{2}(y+1) - x + u(y-1)$ to be the basis of this
	pencil, we get
	\begin{align*}
		f &= \frac{f_{0}}{f_{1}} = \frac{A (u^{2}(y+1) - z_{2}(x y + u (y-1))) + B(z_{2}(y+1) - x + u(y-1))}{
		C (u^{2}(y+1) - z_{2}(x y + u (y-1))) + D(z_{2}(y+1) - x + u(y-1))}.
		\intertext{Similarly,}
		g& = \frac{g_{0}}{g_{1}} = \frac{K (u^{2}(y+1) - z_{4}(x y + u (y-1))) + L(z_{4}(y+1) - x + u(y-1))}{
		M (u^{2}(y+1) - z_{4}(x y + u (y-1))) + N(z_{4}(y+1) - x + u(y-1))}.
	\end{align*}
	Adjusting the coefficients $A$, $B$, $C$, $D$, $K$, $L$, $M$, and $N$ of the M\"obius transformations using the mapping between exceptional divisors,
	we get the following change of coordinates:
	\begin{equation*}
		f = \frac{x - u(y-1) - z_{2}(y+1)}{u^{2}(y+1) - z_{2}(x y + u (y-1))},\qquad g = \frac{u^{2}(y+1) - z_{4}(x y + u (y-1))}{x - u(y-1) - z_{4}(y+1)}.
	\end{equation*}
	This change of variables results in the following identification between two sets of parameters:
	\begin{equation*}
		\nu_{1} = \frac{z_{6}}{u^{2}},\,
		\nu_{2} = \frac{1}{z_{1}},\,
		\nu_{3} = \frac{1}{z_{3}},\,
		\nu_{4} = \frac{1}{z_{5}},\,
		\nu_{5} = -\rho_{1} z_{2} z_{4},\,
		\nu_{6} = -\rho_{2} z_{2} z_{4},\,
		\nu_{7}= u^{2},\,
		\nu_{8}=  z_{2}z_{4},\,
		\kappa_{1} = z_{4},\,
		\kappa_{2} = z_{2}.
	\end{equation*}
	From here, using \eqref{eq:root-params-a1-fg}, we can recompute the root variables in terms of the parameters of $q$-Racah setting,
		\begin{equation*}
			a_{0} = \frac{z_{4}}{z_{2}},\, a_{1} = \frac{z_{5}}{z_{3}}, \, a_{2} = \frac{z_{3}}{ z_{1}},
			\, a_{3} = \frac{z_{1}z_{6}}{u^{2}},
			\, a_{4} = -\frac{u^{2}}{\rho_{1}z_{4}z_{6}}, \, a_{5} = \frac{\rho_{1}}{\rho_{2}},
			\, a_{6} = -\frac{\rho_{2} z_{2} z_{4}}{u^{2}},
			\, a_{7} = \frac{u^{2}}{z_{2}z_{4}},
		\end{equation*}
	we see, using Remark~\ref{rem:period-action}, that the corresponding translation vector is	
	\begin{equation*}
		\psi_{*}: \upalpha =  \langle \alpha_{0}, \alpha_{1}, \alpha_{2}, \alpha_{3}, \alpha_{4}, \alpha_{5}, \alpha_{6}, \alpha_{7}  \rangle
		\mapsto \overline{\upalpha} = \upalpha + \langle 0,0,0,0,0,0,-1,2 \rangle \delta,
	\end{equation*}
	and that it turns out to be \emph{different} from the standard translation vector
	\begin{equation*}
		\varphi_{*}: \upalpha =  \langle \alpha_{0}, \alpha_{1}, \alpha_{2}, \alpha_{3}, \alpha_{4}, \alpha_{5}, \alpha_{6}, \alpha_{7}  \rangle
		\mapsto \overline{\upalpha} = \upalpha + \langle 2,0,0,0,-1,0,0,0 \rangle \delta.
	\end{equation*}
	However, these elements are \emph{conjugated}. This can be observed, for example, by looking at the corresponding
	words in the affine Weyl symmetry group:
	\begin{align*}
		&\psi_{*}:  w_{7} w_{6} w_{5} w_{4} w_{3} w_{0} w_{4} ( w_{5} w_{2} w_{3}
		w_{4} w_{1} w_{2} w_{3} w_{0} w_{4} w_{6} w_{5} w_{4} w_{3} w_{0} w_{4} w_{6} w_{5} w_{2}
		w_{3} w_{4} w_{1} w_{2} w_{3}) w_{0} w_{4} w_{5} w_{6},\\
		&\varphi_{*}: w_{0} w_{4} w_{5} w_{3} w_{4} w_{6} w_{7} ( w_{5} w_{2} w_{3} w_{4} w_{1} w_{2}
		w_{3} w_{0} w_{4} w_{6} w_{5} w_{4} w_{3} w_{0} w_{4} w_{6} w_{5} w_{2} w_{3} w_{4}
			w_{1} w_{2} w_{3}) w_{7} w_{6} w_{5} w_{4},
	\end{align*}
	and then using the far commutativity and the braid relations in ${W}\left(E_{7}^{(1)}\right)$ to write
	\begin{align*}
		\psi_{*} &= (w_{6}w_{5}w_{4} w_{0}w_{7}w_{6} w_{5}w_{4})  \varphi_{*}
		(w_{6}w_{5}w_{4} w_{0}w_{7}w_{6} w_{5}w_{4})^{-1}.
	\end{align*}

	Conjugating by the element $w_{6}w_{5}w_{4} w_{0}w_{7}w_{6} w_{5}w_{4}$ adjusts the divisor matching as
		\begin{align*}
			\mathcal{H}_{f} & = 2\mathcal{H}_{x} +  \mathcal{H}_{y}  - \mathcal{E}_{2} - \mathcal{E}_{4} -
			\mathcal{E}_{6} - \mathcal{E}_{9},\qquad
			& \mathcal{H}_{x} & = \mathcal{H}_{f} +  \mathcal{H}_{g}  - \mathcal{F}_{5} - \mathcal{F}_{6}, \\
			\mathcal{H}_{g} & = \mathcal{H}_{x} +  \mathcal{H}_{y}  - \mathcal{E}_{6} - \mathcal{E}_{9},\qquad
			& \mathcal{H}_{y} & = \mathcal{H}_{f} +  2\mathcal{H}_{g}  - \mathcal{F}_{1} - \mathcal{F}_{5} -
			\mathcal{F}_{6} - \mathcal{F}_{9}, \\
			\mathcal{F}_{1} & = \mathcal{H}_{x} - \mathcal{E}_{6},\qquad
			& \mathcal{E}_{1} & = \mathcal{F}_{2},\ \\
			\mathcal{F}_{2} & = \mathcal{E}_{1},\qquad
			& \mathcal{E}_{2} & = \mathcal{H}_{g} - \mathcal{F}_{5}, \\
			\mathcal{F}_{3} & = \mathcal{E}_{3},\qquad
			& \mathcal{E}_{3} & = \mathcal{F}_{3}, \\
			\mathcal{F}_{4} & = \mathcal{E}_{5},\qquad
			& \mathcal{E}_{4} & = \mathcal{H}_{g} - \mathcal{F}_{6}, \\
			\mathcal{F}_{5} & = \mathcal{H}_{x} +  \mathcal{H}_{y}  - \mathcal{E}_{2} - \mathcal{E}_{6} - \mathcal{E}_{9},\qquad
			& \mathcal{E}_{5} & = \mathcal{F}_{4}, \\
			\mathcal{F}_{6} & = \mathcal{H}_{x} +  \mathcal{H}_{y}  - \mathcal{E}_{4} - \mathcal{E}_{6} - \mathcal{E}_{9},\qquad
			& \mathcal{E}_{6} & = \mathcal{H}_{f} +  \mathcal{H}_{g}  - \mathcal{F}_{1} - \mathcal{F}_{5} - \mathcal{F}_{6}, \\
			\mathcal{F}_{7} & = \mathcal{E}_{7},\qquad
			& \mathcal{E}_{7} & = \mathcal{F}_{7}, \\
			\mathcal{F}_{8} & = \mathcal{E}_{8},\qquad
			& \mathcal{E}_{8} & = \mathcal{F}_{8}, \\
			\mathcal{F}_{9} & = \mathcal{H}_{x}  - \mathcal{E}_{9},\qquad
			& \mathcal{E}_{9} & = \mathcal{H}_{f} +  \mathcal{H}_{g}  - \mathcal{F}_{5} - \mathcal{F}_{6} - \mathcal{F}_{9}.
		\end{align*}
		Proceeding as before, we get the final change of variables \eqref{eq:chvars-fg-xy-a1}, as well as the matching of parameters
		\eqref{eq:chpars-fg-xy-a1}. The inverse change of variables \eqref{eq:chvars-xy-fg-a1} can be computed in a similar way.

		Finally, it is now easy to verify that the parameter dynamic $\bar{z}_{2} = q z_{2}$, $\bar{z}_{4} = q z_{4}$, $\bar{d} = q^{-1} d$ (and so
		$\bar{\rho}_{i} = q^{-1}\rho_{i}$) gives us the correct translation element:
		\begin{equation*}
			\psi_{*}: \upalpha =  \langle \alpha_{0}, \alpha_{1}, \alpha_{2}, \alpha_{3}, \alpha_{4}, \alpha_{5}, \alpha_{6}, \alpha_{7}  \rangle
			\mapsto \overline{\upalpha} = \upalpha + \langle 2,0,0,0,-1,0,0,0 \rangle \delta.
		\end{equation*}

\end{proof}	

\section{Degeneration from the $q$-Racah to the $q$-Hahn case} 
\label{sec:degeneration_from_the_q_racah_to_the_q_hahn_case}

Note that, as shown in Figure~\ref{fig:sakai-scheme}, there exists a degeneration cascade for the $q$-Racah weight that matches (a part of) the degeneration
scheme of discrete Painlev\'e equations.	
In this section we show that our choice of coordinates is compatible with the weight degeneration from the $q$-Racah to the $q$-Hahn case. We
plan to consider the degenerations to Racah and Hahn cases in a separate publication.
The $q$-Hahn case was considered in detail in \cite{K}, however, to match the
$q$-$P\left(A_{2}^{(1)}\right)$-equation as written in \cite{KajNouYam:2017:GAOPE}, we need to make a slightly different change of coordinates.
Below we briefly summarize the relevant data.

\subsection{Reference Example of $q$-$P(A_{2}^{(1)})$} 
\label{sub:reference_example_of_q_p_a__2_1}

\subsubsection{The Root Data} 
\label{ssub:the_root_data-a2}
As before, we take the standard example of the $q$-$P(A_{2}^{(1)})$-equation from \cite{KajNouYam:2017:GAOPE}.
It is characterized by the following Dynkin diagrams:
\begin{center}
	\begin{tabular}{cc}
		$
		{\begin{tikzpicture}[
				elt/.style={circle,draw=black!100,thick, inner sep=0pt,minimum size=2mm}]
			\path 	( 0,0) 	node 	(D1) [elt] {}
					( 2,0) 	node 	(D2) [elt] {}
			        ( 1,1.5) 	node 	(D0) [elt] {};
			\draw [black,line width=1pt ] (D1.east) -- (D2.west);
			\draw [black,line width=1pt ] (D2.north west) -- (D0.south east);
			\draw [black,line width=1pt ] (D1.north east) -- (D0.south west);
				\node at ($(D1.south east) + (-0.5,-0.1)$) 	{$\delta_{1}$};
				\node at ($(D2.south west) + (+0.5,-0.1)$) 	{$\delta_{2}$};
				\node at ($(D0.north) + (0,0.2)$) 	{$\delta_{0}$};
				\node at ($(D1.north) + (0,0.2)$) 	{\color{red}\small$1$};
				\node at ($(D2.north) + (0,0.2)$) 	{\color{red}\small$1$};
				\node at ($(D0.east) + (0.2,0)$) 	{\color{red}\small$1$};
				\end{tikzpicture}}$ \qquad &\qquad  $\displaystyle
				\begin{tikzpicture}[
							elt/.style={circle,draw=black!100,thick, inner sep=0pt,minimum size=2mm}]
						\path 	(-2,0) 	node 	(a1) [elt] {}
						        (-1,0) 	node 	(a2) [elt] {}
						        ( 0,0) node  	(a3) [elt] {}
						        ( 1,0) 	node  	(a4) [elt] {}
						        ( 2,0) 	node 	(a5) [elt] {}
						        ( 0,1)	node 	(a6) [elt] {}
						        ( 0,2)	node 	(a0) [elt] {};
						\draw [black,line width=1pt ] (a1) -- (a2) -- (a3) -- (a4) -- (a5)  (a3) --  (a6) -- (a0);
							\node at ($(a1.south) + (0,-0.2)$) 	{$\alpha_{1}$};
							\node at ($(a2.south) + (0,-0.2)$)  {$\alpha_{2}$};
							\node at ($(a3.south) + (0,-0.2)$)  {$\alpha_{3}$};
							\node at ($(a4.south) + (0,-0.2)$)  {$\alpha_{4}$};	
							\node at ($(a5.south) + (0,-0.2)$)  {$\alpha_{5}$};		
							\node at ($(a6.west) + (-0.2,0)$) 	{$\alpha_{6}$};	
							\node at ($(a0.west) + (-0.2,0)$) 	{$\alpha_{0}$};		
							\node at ($(a0.east) + (0.2,0)$) 	{\color{red}\small$1$};
							\node at ($(a1.north) + (0,0.2)$) 	{\color{red}\small$1$};
							\node at ($(a2.north) + (0,0.2)$) 	{\color{red}\small$2$};
							\node at ($(a3.north) + (0.3,0.2)$)	{\color{red}\small$3$};
							\node at ($(a4.north) + (0,0.2)$) {\color{red}\small$2$};
							\node at ($(a5.north) + (0,0.2)$) 	{\color{red}\small$1$};
							\node at ($(a6.east) + (0.2,0)$) 	{\color{red}\small$2$};
							\end{tikzpicture}$\\
				Dynkin diagram $A_{2}^{(1)}$ & Dynkin diagram $E_{6}^{(1)}$\\[5pt]
				$\begin{aligned}[t]
				\delta_{0} &= \mathcal{H}_{f} + \mathcal{H}_{g} -
				\mathcal{F}_{1}  - \mathcal{F}_{2}  - \mathcal{F}_{3}  - \mathcal{F}_{4}\\
				\delta_{1} &= \mathcal{H}_{f} -
				\mathcal{F}_{5}  - \mathcal{F}_{6}\\
				\delta_{2} &= \mathcal{H}_{g}  - \mathcal{F}_{7}  - \mathcal{F}_{8} \\
				\qquad \\ \qquad \\
				\end{aligned}$ &
				$\begin{aligned}[t]
				\alpha_{0} &= \mathcal{F}_{7} - \mathcal{F}_{8} & \quad \alpha_{4} &= \mathcal{F}_{2} - \mathcal{F}_{3}\\
				\alpha_{1} &= \mathcal{F}_{6} - \mathcal{F}_{5} & \alpha_{5} &= \mathcal{F}_{3} - \mathcal{F}_{4} \\
				\alpha_{2} &= \mathcal{H}_{g} - \mathcal{F}_{1} - \mathcal{F}_{6} & \alpha_{6} &= \mathcal{H}_{f} - \mathcal{F}_{1} - \mathcal{F}_{7} \\
				\alpha_{3} &= \mathcal{F}_{1} - \mathcal{F}_{2} &
				\end{aligned}$ \\[3pt]
	The surface data & The symmetry data
	\end{tabular}
\end{center}


\subsubsection{The Point Configuration} 
\label{ssub:the_point_configuration-a2}
Our model $A_{2}^{(1)}$-surface is obtained from the $A_{1}^{(1)}$-surface on Figure~\ref{fig:a1sur-fg} via the following degeneration.
We rescale parameters $\kappa_{1}\rightsquigarrow \varepsilon \kappa_{1}$, $\nu_{7}\rightsquigarrow \varepsilon \nu_{7}$,
and $\nu_{8}\rightsquigarrow \varepsilon \nu_{8}$ and then let $\varepsilon\to 0$.
%
%
Under this degeneration, $\eta_{*}(d_{1}) = V(fg- \kappa_{1}/\kappa_{2}) = V(f) + V(g)$ decomposes into
$\eta_{*}(d_{1}) = V(f) = H_{f} - F_{5} - F_{6}$ and $\eta_{*}(d_{2}) = V(g) = H_{g} - F_{7} - F_{8}$,
$\eta_{*}(d_{0}) = H_{f} + H_{g} - F_{1} - F_{2} - F_{3} - F_{4}$ remains unchanged, and the new point
configuration becomes
\begin{equation*}
	p_{i}\left(\nu_{i},\frac{1}{\nu_{i}}\right),\quad i = 1,\dots,4;\qquad
	p_{i}\left(0,\frac{\nu_{i}}{\kappa_{2}}\right),\quad i = 5,6;\qquad
	p_{i}\left(\frac{\kappa_{1}}{\nu_{i}},0\right),\quad i = 7,8.
\end{equation*}
see Figure~\ref{fig:a2sur-fg}.
\begin{figure}[h]
	\begin{tikzpicture}[baseline=20ex,
	xscale=3.5, yscale=2.5, basept/.style={circle, draw=red!100,fill=red!100, thick, inner sep=0pt, minimum size=1.8mm}]
    \draw[-] (0.3,0.5) -- (2.3,0.5) node[right] {};
    \draw[-] (0.3,2) -- (2.3,2) node[right] {};
    \draw[-] (2,0.3) -- (2,2.3) node[right] {};
    \draw[-] (0.5,0.3) -- (0.5,2.3) node[right] {};
    \draw[domain=10/23:2.3, smooth, variable=\x, blue, very thick] plot ({\x},{1/\x}) node [right] {$\eta_{*}(d_{0})$};
	\draw[blue, very thick, -] (2.3,0.5) -- (0.3,0.5) node[left] {$\eta_{*}(d_{2})$};
	\draw[blue, very thick, -] (0.5,0.3) -- (0.5,2.3) node[right] {$\eta_{*}(d_{1})$};
	\node[basept] (p1) at (7/10,10/7) [label = below:$p_{1}$]  {};
	\node[basept] (p2) at (9/10,10/9) [label = below:$p_{2}$]  {};
	\node[basept] (p3) at (11.3/10,10/11.3) [label = below:$p_{3}$]  {};
	\node[basept] (p4) at (14/10,10/14) [label = below:$p_{4}$]  {};
	\node[basept] (p5) at (0.5, 1.2) [label = left:$p_{5}$]  {};
	\node[basept] (p6) at (0.5,0.8) [label = left:$p_{6}$]  {};
	\node[basept] (p7) at (0.8,0.5) [label = below:$p_{7}$]  {};
	\node[basept] (p8) at (1.2,0.5) [label = below:$p_{8}$]  {};
	\end{tikzpicture}
	\caption{The standard point configuration for the $A_{2}^{(1)}$-surface}
	\label{fig:a2sur-fg}
\end{figure}


\subsubsection{The Period Map} 
\label{ssub:the_period_map-a2}
The $A_{1}^{(1)}$ symplectic form $\omega$ degenerates (in the affine $(f,g)$-chart) to
\begin{align*}
	\omega &= (\kappa - 1) \frac{df\wedge dg}{(fg - 1)(fg - \kappa)} \rightsquigarrow  \omega = - \frac{df\wedge dg}{fg(fg - 1)}\\
	\operatorname{res}_{\eta_{*}(d_{0})} \omega & = \frac{df}{f} = -  \frac{dg}{g},\qquad
	\operatorname{res}_{\eta_{*}(d_{1})} \omega  = \frac{dg}{g}, \qquad \operatorname{res}_{\eta_{*}(d_{1})} \omega  = -\frac{dg}{g}.
\end{align*}
The same computation as before gives us the following \emph{root variables} $a_{i} = \exp(\chi(\alpha_{i}))$:
\begin{equation*}
	a_{0} = \frac{\nu_{7}}{\nu_{8}}, \quad a_{1} = \frac{\nu_{6}}{\nu_{5}}, \quad a_{2} = \frac{\kappa_{2}}{\nu_{1}\nu_{6}},
	\quad a_{3} = \frac{\nu_{1}}{\nu_{2}}, \quad a_{4} = \frac{\nu_{2}}{\nu_{3}}, \quad a_{5} = \frac{\nu_{3}}{\nu_{4}}.
	\quad a_{6} = \frac{\kappa_{1}}{\nu_{1}\nu_{7}},
\end{equation*}

To reduce the number of parameters, it is convenient to introduce the variables $b_{i}$ for the coordinates of the blowup points
as follows:
\begin{equation*}
	p_{i}\left(b_{i} = \nu_{i},\frac{1}{b_{i}}\right),\quad i = 1,\dots,4;\qquad
	p_{i}\left(0,\frac{1}{b_{i}} = \frac{\nu_{i}}{\kappa_{2}}\right),\quad  i = 5,6;\qquad
	p_{i}\left(b_{i} = \frac{\kappa_{1}}{\nu_{i}},0\right),\quad i = 7,8.
\end{equation*}
This gives us $8$ parameters, however, there is still a rescaling action of the axes preserving
the curve $fg=1$,
\begin{equation*}
	\left(\begin{matrix}
		b_{1} & b_{2} & b_{3} & b_{4} \\ b_{5} & b_{6} & b_{7} & b_{8}
	\end{matrix};\   f, g\right)\sim
	\left(\begin{matrix}
	\lambda b_{1} & \lambda b_{2} & \lambda b_{3} & \lambda b_{4} \\
	\lambda b_{5} & \lambda b_{6} & \lambda b_{7} & \lambda b_{8}	\end{matrix};
	\   \lambda f, \frac{1}{\lambda}g\right),
\end{equation*}
so the true number of parameters is $7$ and they are given by the root variables $a_{i}$. We use the parameter $b_{4}$
as a free parameter and normalize birational maps to keep it fixed. We then have the following relationship between
$b_{i}$ and the root variables $a_{i}$:
\begin{equation}\label{eq:e6a-b}
	a_{0} = \frac{b_{8}}{b_{7}},\quad 	a_{1} = \frac{b_{5}}{b_{6}},\quad 	a_{2} = \frac{b_{6}}{b_{1}},\quad
	a_{3} = \frac{b_{1}}{b_{2}},\quad 	a_{4} = \frac{b_{2}}{b_{3}},\quad 	a_{5} = \frac{b_{3}}{b_{4}},\quad
	a_{6} = \frac{b_{7}}{b_{1}},
\end{equation}
and the root variable parameterization
\begin{equation}\label{eq:e6b-a}
	\left(\begin{matrix}
		b_{1} & b_{2} & b_{3} & b_{4} \\ b_{5} & b_{6} & b_{7} & b_{8}
	\end{matrix};\   f, g\right) =
	\left(\begin{matrix}
		a_{3}a_{4}a_{5}b_{4} & a_{4} a_{5} b_{4} & a_{5} b_{4} & b_{4} \\
		a_{2}a_{3}a_{4}a_{5}b_{4}  & a_{1}a_{2}a_{3}a_{4}a_{5}b_{4}  & a_{3}a_{4}a_{5}a_{6}b_{4} & a_{0}a_{3}a_{4}a_{5}a_{6}b_{4}
	\end{matrix};\   f, g\right).
\end{equation}

Using the decomposition of the anti-canonical divisor class,
$\delta = \delta_{0} + \delta_{1} + \delta_{2} =  \alpha_{0} + \alpha_{1} + 2 \alpha_{2} + 3 \alpha_{3} + 2 \alpha_{4} +  \alpha_{5} + 2 \alpha_{6}$,
we get the following expression for the step $q$ of the dynamics:
\begin{equation*}
	q = \exp(\chi(\delta)) = a_{0} a_{1} a_{2}^{2} a_{3}^{3} a_{4}^{2} a_{5} a_{6}^{2}  =
	\frac{\kappa_{1}^{2} \kappa_{2}^{2}}{\nu_{1}\nu_{2}\nu_{3}\nu_{4}\nu_{5}\nu_{6}\nu_{7}\nu_{8}} =
	\frac{b_{5} b_{6} b_{7} b_{8}}{b_{1} b_{2} b_{3} b_{4}}.
\end{equation*}


\subsubsection{The Standard Dynamic} 
\label{ssub:the_dynamic-a2}
Using the same parameter evolution as in the $A_{1}^{(1)}$-case, $\overline{\kappa}_{1} = \frac{\kappa_1}{q}$,
$\overline{\kappa}_{2} = q \kappa_{2}$, and $\overline{\nu}_{i} = \nu_{i}$ for all $i$, gives us the evolution
$\bar{a}_{2} = q a_{2}$, $\bar{a}_{6} = q^{-1} a_{6}$, $\bar{a}_{i} = a_{i}$ otherwise. This, in view of Remark~\ref{rem:period-action},
gives us the translation
\begin{equation}
	\phi_{*}: \upalpha =  \langle \alpha_{0}, \alpha_{1}, \alpha_{2}, \alpha_{3}, \alpha_{4}, \alpha_{5}, \alpha_{6} \rangle
	\mapsto \overline{\upalpha} = \upalpha + \langle 0,0,-1,0,0,0,1 \rangle \delta.
\end{equation}
on the symmetry sub-lattice which, when written in terms of generators of $\widetilde{W}\left(E_{6}^{(1)}\right)$, becomes
\begin{equation*}
	\phi_{*} = r w_{2} w_{3}\, w_{1} w_{2} w_{6}\, w_{3} w_{4} w_{0}\, w_{6} w_{3} w_{5}\, w_{4} w_{2} w_{3}\, w_{1} w_{2},
\end{equation*}
where $r= (\delta_{0} \delta_{1} \delta_{2}) = (\alpha_{0} \alpha_{5} \alpha_{1}) (\alpha_{2} \alpha_{6} \alpha_{4})$ is a Dynkin diagram automorphism
corresponding to the counterclockwise rotation of the diagram.
The resulting dynamics, written in the affine chart $(f,g)$, is given by  equations (8.8) in Section~8.1.3 of \cite{KajNouYam:2017:GAOPE}:
	\begin{equation*}
		\left\{
		\begin{aligned}
		\frac{(fg - 1)(\bar{f} g - 1)}{f\bar{f}} &= \frac{\left(g - \frac{1}{\nu_{1}}\right)\left(g - \frac{1}{\nu_{2}}\right)
		\left(g - \frac{1}{\nu_{3}}\right)\left(g - \frac{1}{\nu_{4}}\right)}{\left(g - \frac{\nu_{5}}{k_{2}}\right)
		\left(g - \frac{\nu_{6}}{k_{2}}\right)}\\		
		\frac{(fg - 1)(f \underline{g} - 1)}{g \underline{g}} &=
		\frac{(f - \nu_{1})(f - \nu_{2})(f - \nu_{3})(f - \nu_{4})}{\left(f - \frac{k_{1}}{\nu_{7}}\right)
		\left(f - \frac{k_{1}}{\nu_{8}}\right)}
		\end{aligned}.
		\right.
	\end{equation*}



\subsection{Moduli space for the $q$-Hahn connections} 
\label{sub:moduli_space_for_the_q_hahn_connections}

As shown in \cite{K}, the $q$-Hahn case corresponds to the moduli space of $q$-connections of type
$\lambda = (z_{1},\dots, z_{6}; {\bf u}, q v, w, w; 3)$
on the $\mathcal{O}\oplus \mathcal{O}(-1)$ bundle over
$\mathbb{P}^{1}$.
(We write here ${\bf u}$ in bold to distinguish it from the parameter
$u$ in the context of the present paper). After a trivialization, a generic connection
of this type is represented by a matrix $A(z)$ that has the following form:
	\begin{equation*}
		A(z)= \begin{bmatrix}
			a_{11} (z)& a_{12}(z)\\ a_{21}(z) & a_{22}(z)
		\end{bmatrix} ,\quad  A(0)= \begin{bmatrix}
			w & 0\\ 0 & w
		\end{bmatrix},   				
	\end{equation*}
	where $\deg(a_{11})\leq 3$, $\deg(a_{12})\leq 2$,
	$\deg(a_{21})\leq 2$, $\deg(a_{22})\leq 3$
	and
	\begin{equation*}
		\det A(z)={\bf u} v(z-z_1)(z-z_2)(z-z_3)(z-z_4)(z-z_5)(z-z_6).
	\end{equation*}
	We also impose the following asymptotic conditions:
	\begin{equation*}
		\det\det(S^{-1}(qz)A(z)S(z)) = q {\bf u} v
                z^6+\mathcal{O}(z^5)\qquad
                \operatorname{tr}(S^{-1}(qz)A(z)S(z)) = ({\bf u} + q v)z^3+\mathcal O(z^2),
	\end{equation*}
	where $S(z) = \begin{bmatrix} 1 & 0 \\ 0 & z^{-1} \end{bmatrix}$ gives the trivialization of the bundle in the neighborhood of $z = \infty$.
	We consider these  matrices modulo gauge transformations of the form $\hat{A}(z) = R(qz) A(z) R^{-1}(z)$, where the gauge matrix $R(z)$
	has the form
	\begin{equation*}
		R(z)= \begin{bmatrix}
			r_{11} (z)& r_{12}(z)\\ 0 & r_{22}(z)
		\end{bmatrix} ,\qquad  	\deg(r_{11})\leq 1,\ \deg(r_{12})\leq 2,\ \deg(r_{22})\leq 1.
	\end{equation*}

The isomonodromic dynamic $A(z) \rightarrow \overline{A}(z)$ that we consider corresponds to the following parameter evolution:
\begin{equation*}
	(z_1, z_2,\dots, z_6, {\bf u}_s, v_s, w_s) \rightarrow
		(\overline{z}_1,\overline{z}_2,\dots, \overline{z}_{6}, \overline{u}, \overline{v}, \overline{w})
		= ({z}_1,q {z}_2,z_{3},q z_{4}, z_{5}, {z}_{6}, {{\bf u}}, {v}, q {w}).
\end{equation*}
	
	Let us now explicitly describe the moduli space of $q$-Hahn connections of type
	$\lambda = (z_{1},\dots, z_{6}; {\bf u}, qv, w, w,; 3)$. Ater gauging we can put $a_{21}(z) = z(z - t)$, where
	$t =  t_{0}/t_{1}$ is our first spectral coordinate. The second spectral coordinate we adjust slightly
	and put
	\begin{equation*}
		p = \frac{p_{0}}{p_{1}} = \frac{z_{1} z_{3} z_{5} \, a_{11}(t)}{(t - z_{1})(t - z_{3})(t - z_{5})}.
	\end{equation*}
	If we just use $p=a_{11}(t)$, we get singular points $(z_{i},0)$ that results in a $-6$ curve that appears
	after we resolve the singularities of the parameterization using blowup, the above change of variables
	results in two $-3$-curves that are easier to handle. In the coordinates $(t,p)$ we get the following base
	points:
	\begin{equation*}
		\pi_{i}\left(z_{i},0 \right),\  i=1,2,3;\quad
		\pi_{i}\left(z_{i},\infty \right),\  i = 4,5,6;\quad
		\pi_{7}\left(\infty,\rho_{1} = \frac{w^{2}}{v z_{1} z_{2} z_{3}}\right),\
		\pi_{8}\left(\infty,\rho_{2} = \frac{v z_{4}z_{5}z_{6}}{ q}\right);\
		\pi_{9}\left(0,w\right).		
	\end{equation*}
	This gives us the point configuration shown on Figure~\ref{fig:a2sur-xy-fg} on the right. Note that the resulting surface is
	$q$-Hahn surface is again \emph{not minimal}
	and requires \emph{blowing down} the $-1$-curve $t=0$. This
        follows from the properties of the matrix and the nature of
        the parametrization: for $t=0$ we will always have $p=w.$ To match it with the standard $A_{2}^{(1)}$-surface, is easier to
	first blow up the point $\pi_{9}(\infty,0)$ in the standard $(f,g)$-coordinates and
	establishing the identification on the level of Picard lattices, and then extending it to the birational change of coordinates.
	
	\begin{figure}[h]
		\begin{tikzpicture}[
				xscale=1.7, yscale=1.2, basept/.style={circle, draw=red!100,fill=red!100, thick, inner sep=0pt, minimum size=1.8mm},
				conjpt/.style={circle, draw=red!100,fill=white!100, thick, inner sep=0pt, minimum size=1.8mm},scale = 0.7]
			    \draw[-] (-1,-2) -- (4.5,-2) node[right] {};
			    \draw[-] (-1,2) -- (4.5,2) node[right] {};
			    \draw[-] (-0.5,-2.5) -- (-0.5,2.5) node[right] {};
			    \draw[-] (4,-2.5) -- (4,2.5) node[right] {};
				\draw[-, blue, very thick] (4,-2.5) -- (4,2.5) node [above] {$\pi_{*}(d_{2})$};
			    \draw[-,green!60!black, very thick] (-1,-2) -- (4.5,-2) node[right] {};
			    \draw[-,green!60!black, very thick] (-1,2) -- (4.5,2) node[right] {};
			    \draw[-,red, very thick] (-0.5,-2.5) -- (-0.5,2.5) node[right] {};
				\node[basept] (p1) at (1,-2) [label = below:$\pi_{1}$]  {};
				\node[basept] (p2) at (2,-2) [label = below:$\pi_{2}$]  {};
				\node[basept] (p3) at (3,-2) [label = below:$\pi_{3}$]  {};
				\node[basept] (p4) at (1.2,2) [label = above:$\pi_{4}$]  {};
				\node[basept] (p5) at (2.2,2) [label = above:$\pi_{5}$]  {};
				\node[basept] (p6) at (3.2,2) [label = above:$\pi_{6}$]  {};
				\node[basept] (p7) at (4,-1) [label = right:$\pi_{7}$]  {};
				\node[basept] (p8) at (4,1) [label = right:$\pi_{8}$]  {};
				\node[basept] (p9) at (-0.5,0.5) [label = left:$\pi_{9}$]  {};
			\end{tikzpicture}\quad
			\begin{tikzpicture}[
				xscale=4.1, yscale=3.2, basept/.style={circle, draw=red!100,fill=red!100, thick, inner sep=0pt, minimum size=1.8mm},scale = 0.7]
			    \draw[-] (0.3,0.5) -- (2.3,0.5) node[right] {};
			    \draw[-] (0.3,2) -- (2.3,2) node[right] {};
			    \draw[-] (2,0.3) -- (2,2.3) node[right] {};
			    \draw[-] (0.5,0.3) -- (0.5,2.3) node[right] {};
			    \draw[domain=10/23:2.3, smooth, variable=\x, blue, very thick] plot ({\x},{1/\x}) node [right] {$\eta_{*}(d_{0})$};
				\draw[blue, very thick, -] (2.3,0.5) -- (0.3,0.5) node[left] {$\eta_{*}(d_{2})$};
				\draw[blue, very thick, -] (0.5,0.3) -- (0.5,2.3) node[right] {$\eta_{*}(d_{1})$};
				\node[basept] (p1) at (7/10,10/7) [label = below:$p_{1}$]  {};
				\node[basept] (p2) at (9/10,10/9) [label = below:$p_{2}$]  {};
				\node[basept] (p3) at (11.3/10,10/11.3) [label = below:$p_{3}$]  {};
				\node[basept] (p4) at (14/10,10/14) [label = below:$p_{4}$]  {};
				\node[basept] (p5) at (0.5, 1.2) [label = left:$p_{5}$]  {};
				\node[basept] (p6) at (0.5,0.8) [label = left:$p_{6}$]  {};
				\node[basept] (p7) at (0.8,0.5) [label = below:$p_{7}$]  {};
				\node[basept] (p8) at (1.2,0.5) [label = below:$p_{8}$]  {};
				\node[basept] (p9) at (2,0.5) [label = below left:$p_{9}$]  {};
			\end{tikzpicture}			
		\caption{Matching the moduli spaces of $q$-Hahn connections with $A_{2}^{(1)}$-surface}	
		\label{fig:a2sur-xy-fg}
	\end{figure}
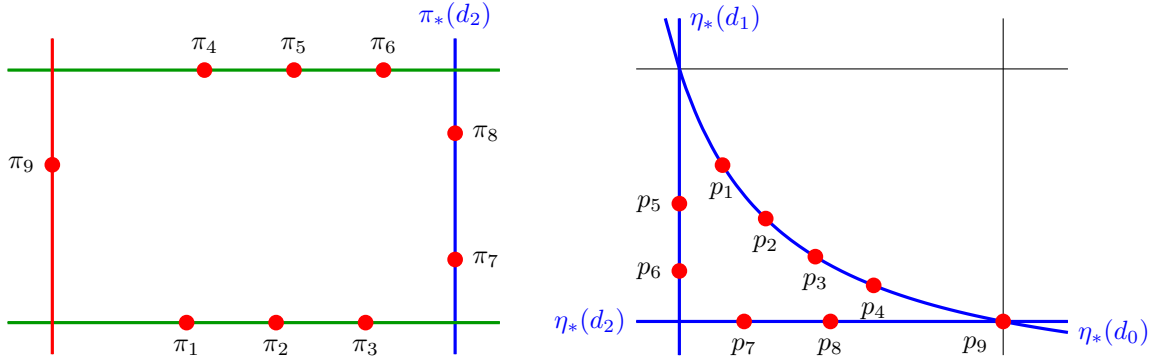

	Using the same techniques as before we get the following change of basis on the level of the Picard lattice,
	\begin{align*}
		\mathcal{H}_{f} &= \mathcal{H}_{t} &
		\mathcal{F}_{1} &= \mathcal{E}_{1}, &
		\mathcal{F}_{3} &= \mathcal{E}_{3}, &
		\mathcal{F}_{5} &= \mathcal{E}_{7}, &
		\mathcal{F}_{7} &= \mathcal{E}_{2}, &
		\mathcal{F}_{9} &= \mathcal{H}_{t} - \mathcal{E}_{9}, \\		
		\mathcal{H}_{g} &= \mathcal{H}_{t} + \mathcal{H}_{p} - \mathcal{E}_{6} - \mathcal{E}_{9}, &
		\mathcal{F}_{2} &= \mathcal{H}_{t} - \mathcal{E}_{6}, &
		\mathcal{F}_{4} &= \mathcal{E}_{5}, &
		\mathcal{F}_{6} &= \mathcal{E}_{8}, &
		\mathcal{F}_{8} &= \mathcal{E}_{4},
	\end{align*}
	as well as the corresponding change of variables
	\begin{equation}\label{eq:chvars-fg-xy-a2}
		f = \frac{1}{t},\qquad g = \frac{t w z_{6} }{z_{6}(p - w) + t w}.
	\end{equation}
	The resulting parameter matching is
	\begin{align*}
		k_{1} &= \frac{1}{w}, &
		\nu_{1} &= \frac{1}{z_{1}}, &
		\nu_{3} &= \frac{1}{z_{3}}, &
		\nu_{5} &= \rho_{1} z_{6}, &
		\nu_{7} &= \frac{z_{2}}{w}, & \\
		k_{2} &= w, &
		\nu_{2} &= \frac{1}{z_{6}}, &
		\nu_{4} &= \frac{1}{z_{5}}, &
		\nu_{6} &= \rho_{2} z_{6}, &
		\nu_{8} &= \frac{z_{4}}{w},
	\end{align*}
	(note that there is a parameter constraint in $q$-Hahn, $w^{2}
        = {\bf u} v z_{1}\cdots z_{6}$. 
	With this identification the spectral coordinates evolution under isomonodromic transformations coincides with
	$q\text{-P}\left(A_{2}^{1}/E_{6}^{(1)}\right)$ of \cite{KajNouYam:2017:GAOPE}.
	The following Corollary is now immediate.
	
\begin{corollary}As $u\to 0$, the change of variables \eqref{eq:chvars-fg-xy-a1} degenerates into \eqref{eq:chvars-fg-xy-a2}. Moreover, the
	base point configuration for the $q$-Racah moduli space shown on Figure~\ref{fig:a1sur-xy} degenerates into the
	base point configuration for the $q$-Hahn moduli space shown on Figure~\ref{fig:a2sur-xy-fg} (left).
\end{corollary}		

\section{Algorithm for Computing Gap Probabilities} 
\label{sec:algorithm_for_computing_gap_probabiliites}
We are now ready to present a recursive procedure for computing the gap probabilities $D_{s}$. Our strategy is the following.
In Section~\ref{sub:initial_conditions} we computed the initial conditions for the recursion and the explicit form of the
connection matrix $A_{N}(z)$. Using $q$-P$\left(E_7^{(1)}/A_{1}^{(1)}\right)$ Painl\'eve equation, we can effectively
compute the evolution $A_{s}(z)\mapsto A_{s+1}(z)$ in the Painlev\'e coordinates $(f, g)$. Further, we have an explicit formula
\eqref{eq:first_gap} for the ratio $D_{s+1}/D_{s}$ obtained in Proposition~\ref{prop:gap} in terms of the matrix elements of the
transition matrix $T_{s}$. Our objective is then to rewrite this formula in terms of the matrix elements of the matrix $A_{s}(z)$.
For that, we use the Lax pair representation~\eqref{eq:isom} obtained in Section~\ref{sub:the_lax_pair},
\begin{equation}\label{eq:start}
A_{s+1}(z)=R_{s}(q^{-1}z)A_s(z)R_{s}^{-1}(z),\text{ where}\quad R_{s}(z)=I+\frac{T_s}{\sigma(z)-\sigma(q^{-s})}.
\end{equation}

First, recall that our nilpotent matrix $T_{s}$ has the form~\eqref{eq:nilp-par},
\begin{equation*}
	T_{s} = \begin{bmatrix}
		t_{s}^{11} & t_{s}^{12} \\ t_{s}^{21} & - t_{s}^{11}
	\end{bmatrix} = \lambda v v_{1}^t,\ \text{where }
	v\in \operatorname{Ker}(T_{s}) = \operatorname{Span}\left\{\begin{bmatrix}
		t_{s}^{11},\\t_{s}^{21}
	\end{bmatrix}\right\},
	\  v_{1}^{t}\in \operatorname{Ker}(T_{s}^{t}) = \operatorname{Span}\left\{\begin{bmatrix}
		t_{s}^{11} & t_{s}^{12}
	\end{bmatrix}\right\}.
\end{equation*}
Fix vectors $v$ and $v_1.$ Then take $v_{2}\notin\operatorname{Span}\{v\},$ note
that $v_{1}^{t} v_{2}\neq0$. Then
$T_{s} v_{2} \in \operatorname{Span}\{v\}$, so after rescaling we can
pick a vector $v_{2}$ so  that $T_{s}v_{2} = v$;
such $v_{2}$ is unique up to adding
a multiple of $v\in \operatorname{Ker}(T_{s})$. Then $\lambda = (v_{1}^{t}v_{2})^{-1}$ and
\begin{equation*}
	T_s=\frac{v v_1^{t}}{v_1^t v_2}.
\end{equation*}

\begin{proposition}\label{prop:gap-det}
	In this parameterization of $T_{s}$, if we take $v = \begin{bmatrix}m_{s}^{11}(\pi_s) \\ m_{s}^{21}(\pi_s)\end{bmatrix}$,
	expression~\eqref{eq:nilp-par} for gap probabilities takes the form
	\begin{equation}
		\frac{D_{s+1}}{D_{s}} = \omega(s) \det[v,v_{2}].
	\end{equation}
\end{proposition}
\begin{proof}
	We can take $v_{2} = \mu \begin{bmatrix}1\\0\end{bmatrix}$. We know from~\eqref{eq:resPs} that
	\begin{equation*}
		T_{s}v = \begin{bmatrix}
			t_{s}^{11}m_{s}^{11}(\pi_{s}) + t_{s}^{12}m_{s}^{21}(\pi_{s})\\ t_{s}^{21}m_{s}^{11}(\pi_{s}) - t_{s}^{11}m_{s}^{21}(\pi_{s})
		\end{bmatrix} = 0,\qquad T_{s}v_{2} = \mu \begin{bmatrix}t_{s}^{11}\\t_{s}^{21}\end{bmatrix} =
		\begin{bmatrix}m_{s}^{11}(\pi_{s})\\m_{s}^{21}(\pi_{s})\end{bmatrix},\text{ and so } \mu = \frac{m_{s}^{11}(\pi_{s})}{t_{s}^{11}}.
	\end{equation*}
	Thus,
	\begin{equation*}
		\det[v,v_{2}] = - \mu m_{s}^{21}(\pi_{s}) = \frac{m_{s}^{11}(\pi_{s})}{t_{s}^{11}}\cdot
		\frac{t_{s}^{11} m_{s}^{11}(\pi_{s})}{t_{s}^{12}}  = \frac{1}{\omega(s)}\cdot \frac{D_{s+1}}{D_{s}}.
	\end{equation*}
The computation is the same as in \ref{prop:gap}
\end{proof}

Consider now the singularity structure of matrices in the Lax Pair \eqref{eq:start}. Since
\begin{align*}
R_{s}(z) &= I + \frac{T_{s}}{\sigma(z) - \sigma(q^{-s})} = I + T_{s}\frac{z }{(z - q^{-s})(z - u^{2}q^{s - 1})}	
= I + \frac{T_{s}}{q^{-s} - u^{2}q^{s-1}}\left(\frac{q^{-s}}{z - q^{-s}} - \frac{u^{2}q^{s-1}}{z - u^{2}q^{s-1}}\right),
\end{align*}
we see that $R_{s}(z)$ and $R^{-1}(z)$ have simple poles at $z = q^{-s}$ and $z = u^{2}q^{s-1}$,
$R_{s}(q^{-1}z)$ and $R^{-1}(q^{-1}z)$ have simple poles at $z = q^{-s+1}$ and $z = u^{2}q^{s}$,
$A_{s}(z)$ and $A_{s+1}(z)$ share simple poles at $z_{3}$, $u^{2}/z_{4}$, $z_{5}$, $u^{2}/z_{6}$,
$A_{s}^{-1}(z)$ and $A_{s+1}^{-1}(z)$ share simple poles at $u^{2}/z_{3}$, $z_{4}$, $u^{2}/z_{5}$, $z_{6}$.
Finally, $A_{s}(z)$ and $A_{s}^{-1}(z)$ have simple poles at $z_{1}(s) = q^{-s+1}$ and $u^{2}/z_{2}(s) = u^{2}q^{s-1}$ and
$A_{s+1}(z)$ and $A_{s+1}^{-1}(z)$ have simple poles at $z_{1}(s+1) = q^{-s}$ and $u^{2}/z_{2}(s+1) = u^{2}q^{s}$. Thus,
\begin{equation*}
	A_{s}(z) = R_{s}(q^{-1}z)H_{s}(z),\qquad A_{s}^{-1}(z) = H_{s}^{-1}(z) R_{s}^{-1}(q^{-1}z),\qquad\text{where}\quad H_{s}(z) = A_{s+1}(z) R_{s}(z)
\end{equation*}
and both $H_{s}(z)$ and $H_{s}^{-1}(z)$ are regular at $q^{-s+1}$. Thus,
\begin{equation*}
	\residue_{z=q^{-s+1}}A_{s}(z) = \frac{v v_{1}^{t}H_{s}(q^{-s+1})q^{-s+1}}{(v_{1}^{t}v_{2})(q^{-s} - u^{2}q^{s-1})},\qquad
	\operatorname{Im}\left(\residue_{z=q^{-s+1}}A_{s}(z)\right) = \operatorname{Span}\{v\},
\end{equation*}
and we can take any $v\in \operatorname{Im}\left(\residue_{z=q^{-s+1}}A_{s}(z)\right)$. Similarly,
$v_{1}\in \operatorname{Im}\left(\residue_{z=q^{-s+1}}(A_{s}^{-1}(z))^{t}\right)$. To find $v_{2}$, observe that if we impose the condition
$T_{s}v_{2} = v$ then $v_2$ is characterized by
\begin{align*}
	R_{s}(q^{-1}z) \left( \frac{(z - q^{-s+1}) (q^{-s} - u^{2} q^{s-1})}{q^{-s+1}} v_{2}\right) &= v + O(z - q^{-s+1}).
\end{align*}
Thus,
\begin{align*}
	\lim_{z\to q^{-s+1}} R_{s}(q^{-1}z) \left(v -  \frac{(z - q^{-s+1}) (q^{-s} - u^{2} q^{s-1})}{q^{-s+1}} v_{2}\right) &=0,
	\intertext{and since $H_{s}^{-1}(z)$ is regular at $z = q^{-s+1}$,}
	\lim_{z\to q^{-s+1}} A_{s}^{-1}(z) \left(v -  \frac{(z - q^{-s+1}) (q^{-s} - u^{2} q^{s-1})}{q^{-s+1}} v_{2}\right) &=0.
\end{align*}

We now show how, given the triple $(v,v_{1}^{t},v_{2})$ for $A_{s}(z)$, to compute the triple $(\hat{v},\hat{v}_{1}^{t},\hat{v}_{2})$ for $A_{s+1}(z)$.
Since
\begin{equation*}
	\hat{v}\in \operatorname{Im}\left(\residue_{z=q^{-s}}A_{s+1}(z)\right) = \operatorname{Span}\{R_{s}(q^{-s-1}) A_{s}(q^{-s})v \},
\end{equation*}
we can take $\hat{v} = R_{s}(q^{-s-1}) A_{s}(q^{-s})v$. Similarly,
$\hat{v}_{1}^{t} = v_{1}^{t}A_{s}^{-1}(q^{-s}) R_{s}^{-1}(q^{-s-1})$.

Finally, to find $\hat v_2$ we will solve
\begin{equation*}
\lim_{z\to q^{-s}} A_{s+1}^{-1}(z) \left(\hat v - \frac{(z - q^{-s})
    (q^{-s-1} - u^{2} q^{s})}{q^{-s}}  \hat v_{2}\right) =0.
\end{equation*}
Since $A^{-1}_s(z)R^{-1}_s(q^{-1}z)$ is regular at $z=q^{-s}$ we can
replace it in the expression for $A_{s+1}^{-1}(z)$ with the series expansion near $z = q^{-s}$ to get
\begin{align*}
	A_{s+1}^{-1}(z) &= \left( I + \frac{T_{s}}{q^{-s} - u^{2}q^{s-1}}\left(\frac{q^{-s}}{z -
      q^{-s}} - \frac{u^{2}q^{s-1}}{z -
      u^{2}q^{s-1}}\right)\right)\\
	  &\qquad \cdot \left(A^{-1}_s(q^{-s})R^{-1}_s(q^{-s-1})+\dfrac{d(A^{-1}_s(z)R^{-1}_s(q^{-1}z))}{d
  z}\bigg \vert_{z=q^{-s}}\cdot(z-q^{-s}) + O(z - q^{-s})^{2}
\right).
\end{align*}
Thus,
\begin{align*}
&\lim_{z\to q^{-s}} A_{s+1}^{-1}(z) \left(\hat v - \frac{(z - q^{-s})
    (q^{-s-1} - u^{2} q^{s})}{q^{-s}}  \hat v_{2}\right) = v+ \\
 &\qquad\frac{q^{-s}T_s}{q^{-s}-u^2 q^{s-1}} \left( \dfrac{d (A^{-1}_s(z)R^{-1}_s(q^{-1}z))}{d
  z}\bigg \vert_{z=q^{-s}}
\hat v - \frac{q^{-s-1}-u^2 q^s}{q^{-s}}A^{-1}_s(q^{-s})R^{-1}_s(q^{-s-1})\hat v_2\right) = 0.
\end{align*}
Since $v = T_{s}v_{2}$, we can now solve for $\hat v_2$ (again, modulo adding a vector in the kernel of $T_{s}$):
\begin{equation*}
\frac{q^{-s}}{q^{-s}-u^2 q^{s-1}} \left( \dfrac{d(A^{-1}_s(z)R^{-1}_s(q^{-1}z))}{d
  z}\bigg \vert_{z=q^{-s}}
\hat v - \frac{q^{-s-1}-u^2 q^s}{q^{-s}}A^{-1}_s(q^{-s})R^{-1}_s(q^{-s-1})\hat v_2\right)=-v_2,
\end{equation*}
or
\begin{align*}\hat v_2=
	\frac{R_{s}\left(q^{-s-1}\right)A_s\left(q^{-s}\right)}{q^{-s-1}-u^2q^s}
	\left(q^{-s}\left.\frac{d\left(A^{-1}_s\left(z\right)R_{s}^{-1}\left(q^{-1}z\right)\right)}{dz}\right|_{z=q^{-s}}  \hat{v}
	+ (q^{-s} -u^{2} q^{s-1})v_{2}\right).
\end{align*}
%
We are now ready to prove the following result.

\begin{proposition}\label{prop-Ds} Gap probabilities can be computed with the help of the following recursion:
	\begin{equation}
		\frac{D_{s+2}D_{s}}{D_{s+1}^{2}} = \frac{\omega(s+1)}{\omega(s)}\left(\frac{\Phi^{-}(q^{-s})}{\Phi^{+}(q^{-s})}\right)^{2}
		\frac{\det[\hat{v},\hat{v}_{2}]}{\det[v,v_{2}]}.
	\end{equation}	
\end{proposition}
\begin{proof}
	Let $v = \begin{bmatrix}m_{s}^{11}(\pi_s) \\ m_{s}^{21}(\pi_s)\end{bmatrix}$. From \eqref{eq:nilp-par} and \eqref{eq:As} we see that
	\begin{equation*}
		m_{s+1} (\sigma(q^{-1}q^{-s})) = R_{s}(q^{-1}q^{-s}) m_{s}(\sigma(q^{-1}q^{-s})) = R(q^{-s-1})A_{s}(q^{-s}) m_{s}(\sigma(q^{-s}))
		D^{-1}(q^{-s}).
	\end{equation*}
	Multiplying on the right by $D(q^{-s})$ and looking at the first column, we get
	\begin{equation*}
		\begin{bmatrix}
			m^{11}_{s+1}(\pi_{s+1})\\m^{11}_{s+1}(\pi_{s+1})
		\end{bmatrix}\frac{\Phi^{+}(q^{-s})}{\Phi^{-}(q^{-s})} = R_{s}(q^{-s-1})A_{s}(q^{-s})
		\begin{bmatrix}m_{s}^{11}(\pi_s) \\ m_{s}^{21}(\pi_s)\end{bmatrix} = R_{s}(q^{-s-1})A_{s}(q^{-s} )v = \hat{v}.
	\end{equation*}
	In view of linearity, $\hat{v}_{2}$ scales in the same way as $\hat{v}$, and so by Proposition~\ref{prop:gap-det}, we get
	\begin{equation*}
		\det[\hat{v},\hat{v}_{2}] = \left(\frac{\Phi^{+}(q^{-s})}{\Phi^{-}(q^{-s})}\right)^{2}\cdot \frac{1}{ \omega(s+1) }
		\frac{D_{s+2}}{D_{s+1}}.
	\end{equation*}
	Thus,
	\begin{equation*}
		\frac{D_{s+2}D_{s}}{D_{s+1}^{2}} = \frac{\omega(s+1)}{\omega(s)}\left(\frac{\Phi^{-}(q^{-s})}{\Phi^{+}(q^{-s})}\right)^{2}
		\frac{\det[\hat{v},\hat{v}_{2}]}{\det[v,v_{2}]},
	\end{equation*}
	where the final equation is now independent of the choice of the initial scaling for $v$.
\end{proof}

\bibliographystyle{amsalpha}

\providecommand{\bysame}{\leavevmode\hbox to3em{\hrulefill}\thinspace}
\providecommand{\MR}{\relax\ifhmode\unskip\space\fi MR }
\providecommand{\href}[2]{#2}

\end{document}